\documentclass[12pt,english]{article}
\usepackage[T1]{fontenc}
\usepackage[utf8]{inputenc}
\usepackage{geometry}
\usepackage{color}
\usepackage{babel}
\usepackage{verbatim}
\usepackage{url}
\usepackage{amsmath}
\usepackage{amsthm}
\usepackage{amssymb}
\usepackage{graphicx}
\usepackage{setspace}
\usepackage[authoryear]{natbib}
\usepackage{subcaption}
\usepackage[export]{adjustbox}
\usepackage[unicode=true,pdfusetitle,
 bookmarks=true,bookmarksnumbered=false,bookmarksopen=false,
 breaklinks=false,pdfborder={0 0 1},backref=false,colorlinks=false]
 {hyperref}
\hypersetup{
 pdfborderstyle=,colorlinks,citecolor=blue,linkcolor=black,urlcolor=black}

\makeatletter

\theoremstyle{remark}
\newtheorem{rem}{\protect\remarkname}
\theoremstyle{plain}
\newtheorem{prop}{\protect\propositionname}
\theoremstyle{plain}
\newtheorem{assumption}{\protect\assumptionname}
\theoremstyle{plain}
\newtheorem{thm}{\protect\theoremname}
\theoremstyle{plain}
\newtheorem{lem}{\protect\lemmaname}

\usepackage{babel}

\usepackage{appendix}
\usepackage{multirow}
\usepackage{longtable}
\usepackage{rotating}
\usepackage{makecell}
\usepackage{threeparttable}
\usepackage[table]{xcolor}

\renewcommand{\hat}{\widehat}
\renewcommand{\tilde}{\widetilde}

\providecommand{\assumptionname}{Assumption}
\providecommand{\lemmaname}{Lemma}
\providecommand{\remarkname}{Remark}
\providecommand{\theoremname}{Theorem}

\makeatother

\providecommand{\assumptionname}{Assumption}
\providecommand{\lemmaname}{Lemma}
\providecommand{\propositionname}{Proposition}
\providecommand{\remarkname}{Remark}
\providecommand{\theoremname}{Theorem}

\usepackage{xr}
\makeatletter

\newcommand*{\addFileDependency}[1]{
\typeout{(#1)}
\@addtofilelist{#1}
\IfFileExists{#1}{}{\typeout{No file #1.}}
}\makeatother

\begin{document}
\title{Nickell Bias in Panel Local Projection:\\
 Financial Crises Are Worse Than You Think$^{\dag}$}
\author{Ziwei Mei$^a$, Liugang Sheng$^b$, Zhentao Shi$^b$ \vspace{0.6em}
\\   
 $^a$University of Macau \\  
$^b$The Chinese University of Hong Kong}
\date{\vspace{-5ex}}
\maketitle

\singlespacing
\begin{abstract}

Panel local projection (LP) with fixed-effects (FE) is widely adopted for evaluating the economic consequences of financial crises across countries. This paper highlights a fundamental methodological issue: the presence of the Nickell bias in the panel FE estimator due to inherent dynamic structures of predictive specifications, even if the regressors have no lagged dependent variables. The Nickell bias invalidates the standard inferential procedure based on the $t$-statistic. We propose a split-panel jackknife (SPJ) estimator as a simple, easy-to-implement, and yet effective solution to eliminate the bias and restore valid statistical inference. We revisit four influential empirical studies on the impact of financial crises, and find that the FE method underestimates the economic losses of financial crises relative to the SPJ estimates. Replication files are available at \url{https://metricshilab.github.io/panel-lp-replication/}, with links to R and Stata packages.

\end{abstract}
\vspace{0.8cm}



\vspace{1cm}

{\small{}$^{\dag}$ Ziwei Mei, }\texttt{\small{}ziweimei@um.edu.mo}{\small{}.
Liugang Sheng, }\texttt{\small{}lsheng@cuhk.edu.hk}{\small{}. Zhentao Shi (Corresponding
author),} \texttt{\small{}zhentao.shi@cuhk.edu.hk}{\small{},
9/F Esther Lee Building, Department of Economics, The Chinese University of Hong Kong, Shatin, New Territories, Hong Kong SAR, China. 
Sheng thanks the financial support from the Research Grants Council of the Hong Kong SAR (Project No.~14501821). Shi acknowledges the partial financial support from the National Natural Science Foundation of China (Project No.~72425007). We thank Qiyu Dai, Shu Shen, and Ji Pan for their excellent research assistance. We are also grateful to Geert Dhaene, Raffaella Giacomini, Frank Kleibergen, Byoungchan Lee, {\`O}scar Jord{\`a}, Xun Lu, Ryo Okui, and Liangjun Su for their helpful comments.} 

\newpage

\onehalfspacing 
\section{Introduction\label{sec:Introduction}}
Financial crises inflict lasting damage on economies, causing output contractions, persistent unemployment, and economic scarring. Ben
S.~Bernanke, the 2022 Nobel laureate, demonstrated how bank runs played a decisive role in the Great Depression of the 1930s, the worst economic crisis in modern history \citep{bernanke1983nonmonetary}. The Great Recession following the 2008 financial crisis reignited interest in understanding the relationship between financial shocks and economic recessions.  Numerous studies have since demonstrated that financial crises generally produce deeper and more persistent economic losses compared to typical recessions
\citep{reinhart2009aftermath,laeven2013systemic,schularick2012credit,jorda2013credit}. 

A central question in this burgeoning literature is to quantify the magnitude and persistence of economic losses caused by financial crises. Accurate assessment of crisis-induced output losses is crucial not only for resolving theoretical debates about crisis persistence, transmission channels, and the speed of economic recovery, but also for guiding policymakers in crafting effective macroeconomic responses, such as monetary easing, fiscal stimulus, and macroprudential regulation, to mitigate economic damage and promote recovery \citep{sufi2021financial}. As economists employ various measures of financial crises and collect panels of countries with financial shocks,  the panel version of \cite{jorda2005estimation}'s local projection (LP) has become widely used by recent empirical analyses of crisis-related economic losses for its simplicity, robustness, and flexible specifications. Moreover, virtually all studies applying panel LP techniques employ fixed-effect (FE) estimators to control for unobserved heterogeneity, cementing FE as the \emph{de facto}  approach in these empirical studies. For example, by applying the FE to various cross-country panel data, \citet{romer2017new}, \citet{baron2021banking}, and
\citet{mian2017household} show that financial distress, banking crises,
and household debts lead to severe economic contractions, respectively.\footnote{Many other studies, including \citet{jorda2013credit}, \citet{jorda2015leveraged},
\citet{jorda2016great}, \citet{zeev2017capital}, \citet{bhattarai2021local}, as well as the recent survey on financial crises conducted by \citet{sufi2021financial},
also use FE to examine the impacts of financial
crises on economic contraction. }

Despite its widespread adoption, the panel LP model is subject to a fundamental limitation: the FE estimator in panel LP inherently incurs the Nickell bias \citep{nickell1981biases}, even in the absence of lagged dependent variables. To clarify this issue, consider the following illustrative $h$-period-ahead panel LP model:
\begin{align}
y_{i,t+h} & =\mu_{i}^{(h)y}+\beta^{(h)}x_{i,t}+e_{i,t+h}^{(h)},\ \ \text{for }t=1,2,\ldots,T-h,\text{ and }h=0,1,\ldots,H\label{eq:modelh},
\end{align}
where $x_{i,t}$ is a measure of financial shock, $y_{i,t+h}$ is the outcome variable, e.g., the logarithm of real GDP, $e_{i,t+h}^{(h)}$ is the
error term uncorrelated with $x_{i,t}$, and $\mu_{i}^{(h)y}$ is the
individual-specific heterogeneity, or the fixed effect. The IRF $(\beta^{(h)})_{h=0}^{H}$
is of central interest for understanding the dynamic impact of financial crises on economic activities.  Despite the absence of the lagged dependent variable in (\ref{eq:modelh}),
FE for the panel LP model is asymptotically biased
when the number of cross-sectional units $N$ and the time periods
$T$ are both large; to be precise, the leading case is $(N,T)\to\infty$
and $N/T\to c$ for some constant $c\in(0,\infty)$.
It is a type of Nickell bias in the FE estimator for multiple-equation panel vector autoregression (VAR)---the full dynamic model behind
the single-equation LP. The Nickell bias has important implications for inference. Standard inference procedures, such as $t$-statistics with critical values based on the standard normal distribution, become unreliable, leading to distorted test sizes, and thus potentially misleading empirical findings. To the best of our knowledge, no systematic study has examined how Nickell bias in panel LP affects estimates of the economic losses associated with financial crises, both qualitatively and quantitatively.

To facilitate the theoretical analysis of the bias, we first present a
prototype model where the dependent variable $y_{i,t+1}$ is generated
from the following data generating process (DGP) 
\begin{equation}
y_{i,t+1}=\mu_{i}^{(0)y}+\beta^{(0)}x_{i,t+1}+u_{i,t+1}^{y},\label{eq:DGP}
\end{equation}
 and the variable of interest $x_{i,t}$ follows a panel autoregression (AR) of order 1:
\begin{align}
x_{i,t+1} & =\mu_{i}^{x}+\rho x_{i,t}+u_{i,t+1}^{x},\label{eq:AR1}
\end{align}
where $|\rho|<1$ ensures stationarity. This simple model allows
us to derive the \textit{analytical} expression of the bias of the FE estimator of the IRF from this underlying
dynamic DGP, which delivers four important implications. Firstly, 
FE suffers from an \emph{intrinsic} Nickell bias 
due to Equation (\ref{eq:modelh})'s violation of \emph{strict exogeneity}.\footnote{\label{footnote: strict exogeneity} Technically, strict exogeneity is violated if $
\mathbb{E}\left[e_{i,t+h}^{(h)}|({x}_{i,s})_{1\leq s\leq T}\}\right]  \neq0  $; see (\ref{eq: violate strict exo}) for justifications.
}
Secondly, FE exhibits an \emph{attenuation bias} and underestimates the true dynamic effects when $\rho > 0$ in the leading case of economic autoregressive relationships. 
This result suggests that previous studies on financial crises using FE may have underestimated the impact of financial crises on economic contractions. Third, the Nickell bias tends to rise with the horizon ($h$) and the persistence of the regressor ($\rho$). This implies that FE is more likely to underestimate the long-term impact of financial shocks, and the bias would be larger for more persistent financial shocks. Lastly, the distortion in inference caused by the Nickell bias tends to be more severe for a panel with larger $N$ and smaller $T$.

To eliminate the first-order Nickell bias and restore valid statistical inference in panel LP regressions, we propose the split-panel jackknife (SPJ) estimator \citep{Dhaene2015,chudik2018half} as a simple, effective, and general solution. It is easy to implement, using the formula:
\[
\widehat{\beta}^{(h)\mathrm{spj}}=2\widehat{\beta}^{(h)\mathrm{fe}}-(\widehat{\beta}_{a}^{(h)\mathrm{fe}}+\widehat{\beta}_{b}^{(h)\mathrm{fe}})/2,
\]
where $\widehat{\beta}^{(h)\mathrm{fe}}$, $\widehat{\beta}_{a}^{(h)\mathrm{fe}}$,
and $\widehat{\beta}_{b}^{(h)\mathrm{fe}}$ are the plain FE estimators
from all the time periods, the first half ($t\leq T/2$), and the
second half ($t>T/2$), respectively. Given that most empirical
applications of LP involve at least a moderate $T$, we show that $\widehat{\beta}^{(h)\mathrm{spj}}$ is asymptotically unbiased and follows a zero-mean normal distribution when $N/T^{3}\to0$. Our formal theory of SPJ is developed in a linear model that accommodates extra control variables. The theory can be further extended to entertain
other practical specifications in panel data analysis (See Appendix B).   

SPJ \label{para: spj advantage} preserves the key advantages of the LP method. It allows researchers to focus on specifying the main equation of interest linking $y_{i,t+h}$ and $x_{i,t}$ without requiring a full VAR system, and provides a unified regression framework that accommodates a broad range of model specifications, including different lag orders, control variables, and dependent variable forms. It is a practical and theoretically grounded solution for panel data inference.

To illustrate the harm of the Nickell bias and evaluate the effectiveness of the SPJ correction, we revisit four seminal studies examining the macroeconomic consequences of financial crises. These applications are chosen for their prominence in the literature, their diverse yet interconnected economic contexts, and their ability to demonstrate the broad relevance of bias correction methods in empirical macro-finance analysis. Specifically, they cover a comprehensive spectrum of financial disturbances, including general financial distress, banking crises, household debt, and currency collapses.

The first example is \citet{romer2017new}, which creates a narrative semiannual measure of financial distress index on a 16-bin scale for 24 advanced countries from 1967 to 2012, and shows that the output experiences significant and persistent declines following typical financial distress shocks, highlighting the general macroeconomic vulnerability to financial disruptions. Next, focusing on banking crises, \citet{baron2021banking} construct a new data set of banking crises based on bank equity crash of more
than 30\% for 46 advanced and emerging economies over 1870-2016, and find that banking crises cause sharp output contractions and severe credit constraints. Furthermore, \citet{mian2017household} examine the impact of household debts. They construct a panel of 30 countries from 1960 to 2012, and show how elevated household indebtedness contributes significantly to recessions, particularly evident in the medium-run aftermath of housing booms. The last one is currency crisis, a central concern in international macroeconomics, given its acute and enduring impacts on national economies. \citet{cerra2008growth} construct a panel of currency crisis episodes across 175 countries from 1965 to 2000, and find that currency crises have persistent negative effects on output.

Our reexamination broadly confirms the original findings from these seminal studies: financial shocks consistently and significantly reduce future economic output. However, compared to the asymptotically unbiased SPJ, FE tends to underestimate the medium- and long-term economic contractions; such underestimation can be substantial in some applications, epitomized by \citet{cerra2008growth}. The underestimation by FE is consistent with our theoretical finding that the Nickell bias shrinks the FE estimator toward zero; see Remark \ref{rem:attentuation} for details.  Across the four applications we revisited, the relative differences between SPJ and FE at the peak year---when output loss is greatest---range from 16\% to 75\%. Specifically, SPJ in the \citet{romer2017new} indicates a peak decline in output of 6.3\% approximately three and a half years after a typical financial shock, whereas FE yields a smaller decline of 5.4\%. Similarly, for banking crises analyzed by \citet{baron2021banking}, the FE estimates suggest a cumulative GDP decrease of 3.4\% four years after crises, while the SPJ reveals a more pronounced decline of about 4.2\%. In \citet{mian2017household}, a 10 percentage-point increase in the household debt-to-GDP ratio corresponds to a peak GDP contraction of 3.9\% using FE estimates, compared to a larger drop of 5.6\% estimated with SPJ. Finally, the case of currency crises exhibits the largest discrepancy: SPJ indicates a peak cumulative output loss of approximately 7.9\% a decade after the crisis, which is roughly 75\% greater than the FE estimates. This substantial discrepancy likely arises due to the large panel of countries ($N=175$) accompanied by a relatively short time length ($T=35$).
In summary, our analysis reveals a consistent pattern of underestimation by the FE estimator, particularly for long horizons, when evaluating output losses associated with financial crises. This systematic underestimation mirrors the Nickell bias featured in our prototype model, which biases FE toward zero. 

The above empirical findings underscore that the Nickell bias of FE in panel LP is not merely an accidental occurrence but a systematic pathology that must be taken care of
by economists who are serious about methodology, and policymakers who are serious about economic consequences. If policymakers rely on the biased empirical estimation method, they are likely to underestimate both the severity and duration of economic downturns, resulting in insufficient or prematurely withdrawn monetary and fiscal interventions. Consequently, stimulus packages, stabilization policies, and financial assistance programs calibrated based on biased estimates will fail to fully address the true magnitude of the crisis. To illustrate this point, we discuss the fiscal policy implications of bias correction in Section \ref{subsec: summary empirical}. We demonstrate that correcting this bias through SPJ is both essential and economically significant, as it provides policymakers with accurate and unbiased assessments, whereas FE tends to underestimate the required fiscal stimulus to counteract the potential maximum output loss after a typical financial crisis for a given government spending multiplier.

\textbf{Literature Review and Contributions}. 
Our paper makes econometric contributions to the literature on financial crises by highlighting the Nickell bias of FE in panel LP.
The methodology extends beyond the literature on financial crises. The panel LP with fixed effects is also widely used in empirical macroeconomics, for example, to assess the dynamic impacts of political shocks, monetary and fiscal policies, and pandemics \citep{acemoglu2019democracy,jorda2022longer,gilchrist2022sovereign}. Moreover, its recent adoption in applied microeconomic studies at sectoral, firm, product, or household levels further underscores the broad relevance of addressing Nickell bias, particularly given the large cross-sectional dimensions and short time horizons typical of these datasets \citep{bahaj2020home, ottonello2020financial,  caldara2022measuring, boehm2023long}.\footnote{See \citet{jorda2025local} for a comprehensive review of the local projection method and its broad applications.} Our SPJ can also be easily applied in those studies to deal with the Nickell bias.  

This study is built upon an econometric basis. The finite sample bias of time series AR models is raised by
\citet{kendall1954}. The detrimental effect of such bias is amplified
in panel data. \citet{nickell1981biases} showcases
the treacherous nature of panel data: a seemingly innocuous procedure
may face unexpected difficulty when a transformation takes care of
the individual-specific heterogeneity embodied by the fixed effects; see a recent
survey by \citet{okui20211}. Regarding the estimation of IRFs, the Nickell bias is well-known in dynamic panel VAR literature \citep{holtz1988estimating,hahn2002asymptotically,Arellano2003,greenaway2013multistep}. With the increasing popularity of single-equation LP for IRF estimation due to its simplicity, the Nickell bias and its implications in the panel LP framework have so far remained unrecognized in the literature. 
To the best of our knowledge, this paper is the first to point out the omnipresence of Nickell bias even if the lagged dependent variables are absent from the regressors \citep{jorda2025local}.

For general dynamic panel models, the Nickell bias is usually addressed with the instrumental variable (IV) method \citep{Anderson1982,arellano1991some,arellano1995another,blundell1998initial,Hsiao2002} or analytical formulas \citep{Kiviet1995,bun2005bias,alvarez2022robust}. The IV method is subject to the lack of efficiency and performs poorly with finite samples when the regressor is persistent \citep{fernandez2018fixed}. The analytical methods require the availability of closed-form formulas. Therefore, case-by-case adjustments and mathematical derivations are necessary for the IV and analytical methods.
\citet{Dhaene2015} tackle the Nickell bias by SPJ, which is an ``automated'' estimator that spares applied researchers from the potential weak instrument issue and complex analytical derivations. The application of SPJ goes beyond single-equation linear dynamic panel to multiple-equation models \citep{dhaene2016bias} and nonlinear models 
\citep{weidner2021bias}, due to its simplicity and robustness. We therefore recommend using SPJ in panel LP to preserve the attractive features of LP. Though \label{para: literature} SPJ is adopted from \citet{Dhaene2015}, this paper is the first to highlight the necessity of using SPJ specifically in the context of panel LPs to remove the Nickell bias for valid inference.

This paper is closely related to, while different from \citet{chudik2018half} and \citet{herbst2021bias}. \citet{chudik2018half} explore the Nickell bias in a generic linear panel model with weakly exogenous regressors and then extend SPJ to correct
it in two-way FE models. Our paper, on the other hand, is specifically motivated
by panel LP, and the solutions and empirical examples are tailored
for the panel LP. As panel LP involves a series of regressions, 
the regressor $x_{i,t}$ in (\ref{eq:modelh}) violates strict exogeneity when $h\geq 1$. Therefore, weak exogeneity
is deduced, rather than assumed, as an intrinsic feature of panel
LP in learning the economic IRF. Moreover, we provide an in-depth
study of panel VAR($\infty$) about when strict exogeneity is violated
in panel LP (see Appendix A.2). 
\citet{herbst2021bias} focus on point estimates and 
investigate the finite sample bias in the time series LP model 
with and without (in their Appendix A.2) lagged dependent variables and 
then extend their discussion into panel data;
they employ \citet{bao2007second} for analytical bias correction. 
Our paper points out that all regressors in panel LP are subject to Nickell bias regardless of the linear specification, and our solution is free of analytical bias calculation. Beyond the impacts on point estimates, we highlight that in panel LP the Nickell bias is not a finite sample issue in statistical inference; it sustains in large samples by shifting the mean of the asymptotic normal distribution of FE, thereby spoiling the theoretical foundation for standard inference using the $t$-statistics. The impact of the bias on statistical inference arises from the asymptotic normal distribution rather than the finite sample standard error. 

\textbf{Data and Software}. 
SPJ is easy to implement and can be extended to many variants of specifications. All empirical applications can be replicated with data and code available at \url{https://metricshilab.github.io/panel-lp-replication/}.
This repository also provides links to R and Stata packages to facilitate empirical practices.
Given input data, a single-line command of the main function will automatically compute the IRF, standard error, and the corresponding confidence intervals.

\textbf{Organization}. The rest of the paper is organized as follows.
In Section \ref{sec:bias}, we first use 
a simple two-equation model to demonstrate the presence of 
Nickell bias without lagged dependent variables
and provide the analytical expression of this bias. 
We then show that the Nickell bias is a generic phenomenon 
in panel LP and propose using the SPJ to restore asymptotic normality centered at
zero. Moreover, Monte Carlo simulations are carried out to verify
the theoretical predictions. 
Section \ref{sec:Empirical-Application}
estimates the IRFs by FE and SPJ in four empirical examples of
macro-finance to evaluate the impact of financial crises.
Proofs, theoretical extensions, and additional empirical
works are relegated to the Appendix.

\section{Models and Theory \label{sec:bias}}

The FE estimator is the most popular method that estimates the slope
coefficient in the linear panel regression and in the meantime controls
the unobservable individual-specific heterogeneity. For simplicity
let us consider the regression (\ref{eq:modelh}). Let $T_{h}=T-h$
be the effective sample size and $\mathcal{T}^{h}=[T_{h}]=\{1,2,\ldots,T_{h}\}$
be the corresponding index set, where throughout the paper we use
$[q]=\{1,2,\cdots,q\}$ to denote
the set of positive integers up to a natural number $q$. The FE estimator is 
\begin{equation}
\hat{\beta}^{(h){\rm \mathrm{fe}}}=\sum_{i\in[N]}\sum_{t\in\mathcal{T}^{h}}\tilde{x}_{i,t}\tilde{y}_{i,t+h}\bigg/\sum_{i\in[N]}\sum_{t\in\mathcal{T}^{h}}\tilde{x}_{i,t}^{2},\label{eq:FEest}
\end{equation}
where $\tilde{y}_{i,t+h}=y_{i,t+h}-T_{h}^{-1}\sum_{t\in\mathcal{T}^{h}}y_{i,t+h}$
is the within-group demeaned dependent variable, and similarly within-group
demeaned is the regressor $\tilde{x}_{i,t}$. To conduct statistical
inference about $\beta^{(h)}$, the standard deviation of $\hat{\beta}^{(h){\rm \mathrm{fe}}}$
is calculated as 
\[
\hat{s}^{(h){\rm \mathrm{fe}}}=(\sum_{i\in[N]}\sum_{t,s\in\mathcal{T}^{h}}\tilde{x}_{i,t}\tilde{x}_{i,s}\hat{e}_{i,t+h}^{(h)\mathrm{fe}}\hat{e}_{i,s+h}^{(h)\mathrm{fe}})^{1/2}/{\sum_{i\in[N]}\sum_{t\in\mathcal{T}^{h}}\tilde{x}_{i,t}^{2}}
\]
where $\hat{e}_{i,t+h}^{(h)\mathrm{fe}}=\tilde{y}_{i,t+h}-\tilde{x}_{i,t}\hat{\beta}^{(h){\rm \mathrm{fe}}}$
is the estimation residual. We then construct the $t$-statistic 
\[
(\hat{\beta}^{(h){\rm \mathrm{fe}}}-\beta^{(h)\mathrm{null}})/\hat{s}^{(h){\rm \mathrm{fe}}},
\]
where $\beta^{(h)\mathrm{null}}$ is a hypothesized value under a
null of economic interest, and compare the value of the $t$-statistic
with a critical value drawn from the standard normal distribution
$\mathcal{N}(0,1)$. For example, for a two-sided test with size 5\%,
we reject the null if the absolute value of the $t$-statistic is
larger than 1.96. Does the $t$-statistic based on $\hat{\beta}^{(h){\rm \mathrm{fe}}}$
provide valid statistical inference for the true IRF $\beta^{(h)}$?
This question is complicated by the within-group transformation
in the dynamic setting.

\subsection{Intrinsic Nickell Bias in Panel LP{\normalsize{} \label{subsec: implicit bias}}}

The procedure described above is the common practice based on FE.
However, there is an intrinsic Nickell bias built into panel LP. \citet{jorda2005estimation}
constructs the time series LP by a VAR system. Here we use (\ref{eq:DGP})
and (\ref{eq:AR1})---a stylized panel VAR(1)---as the true DGP
for demonstration. 
The dependent variable $y_{i,t+1}$ in (\ref{eq:DGP}) is linked to the
regressor $x_{i,t+1}$ by the slope parameter $\beta^{(0)}$. 
In (\ref{eq:AR1}) $x_{i,t+1}$ follows one of the simplest time series models---a
stationary AR(1) model. 
Let $\mathbf{x}_{i}^{t}=\left(x_{i,0},x_{i,1},\ldots,x_{i,t}\right)$
be the time series of the regressor from time 0 up to $t$. When assuming
$u_{i,t+1}^{x}$ and $u_{i,t+1}^{y}$ independently and identically
(i.i.d.)~distributed across $i$ and $t$, we have $\mathbb{E}\left[u_{i,t+1}^{y}|\mathbf{x}_{i}^{T}\right]=0$
and thus in (\ref{eq:DGP}) strict exogeneity holds. The FE estimator is asymptotically unbiased when $h=0$ in the regression (\ref{eq:DGP}).

LP is a series of linear regressions across different horizons. For
$h=1$ we substitute (\ref{eq:AR1}) into (\ref{eq:DGP}) and for
$h\geq2$ we repeat the substitution to produce (\ref{eq:modelh})
 with the following closed-form expressions 
 \begin{equation}\label{eq: close form prototype}
\begin{aligned}
\beta^{(h)} & =\rho^{h}\beta^{(0)},\\
\mu_{i}^{(h)y} & =\mu_{i}^{(0)y}+\beta^{(0)}\mu_{i}^{x}(1+\sum_{s=1}^{h-1}\rho^{s}),\\
e_{i,t+h}^{(h)} & =u_{i,t+h}^{y}+\beta^{(0)}(u_{i,t+h}^{x} + \sum_{s=1}^{h-1}\rho^{s}u_{i,t+h-s}^{x}).
\end{aligned}
 \end{equation}
Denote $\mathbf{u}_{i}^{x,t}=(u_{i,s}^{x})_{s=1}^{t}$. The composite error term $e_{i,t+h}^{(h)}$ is weakly exogenous\footnote{\citet{chudik2018half} define weak exogeneity in terms of unconditional
moments. We define weak exogeneity in terms of mean independence conditional
on the past information, following \citet{Mikusheva2023}.} in that 
\[
\mathbb{E}\left[e_{i,t+h}^{(h)}|\mathbf{x}_{i}^{t}\right]=\mathbb{E}\left[\beta^{(0)}(u_{i,t+h}^{x} + \sum_{s=1}^{h-1}\rho^{s}u_{i,t+h-s}^{x})+u_{i,t+h}^{y}\Bigg|x_{i,0},\mathbf{u}_{i}^{x,t}\right]=0.
\]
However, it is not strictly exogenous as 
\begin{equation}\label{eq: violate strict exo}
\begin{aligned}
\mathbb{E}\left[e_{i,t+h}^{(h)}|\mathbf{x}_{i}^{T}\right] & =\mathbb{E}\left[\beta^{(0)} (u_{i,t+h}^{x} + \sum_{s=1}^{h-1}\rho^{s}u_{i,t+h-s}^{x})+u_{i,t+h}^{y}\Bigg|x_{i,0},\mathbf{u}_{i}^{x,T}\right]\\
& =\beta^{(0)} (u_{i,t+h}^{x} + \sum_{s=1}^{h-1}\rho^{s}u_{i,t+h-s}^{x})\neq0
\end{aligned}
\end{equation}
if $\beta^{(0)}\neq0$. 
\begin{rem}
When $\rho=0$,
the regressor $x_{i,t}$ is independent across time, but  
strict exogeneity is still violated for all $h\geq1$ due to
$\mathbb{E}\left[e_{i,t+h}^{(h)}|\mathbf{x}_{i}^{T}\right] = \beta^{(0)}u_{i,t+h}^{x}\neq0$.
Strict exogeneity occurs in (\ref{eq:modelh}) if and only if $\beta^{(0)}=0$, under
which $(x_{i,t})$ and $(y_{i,t})$ are two autonomous time series
with no connection. 
\end{rem}
At first glance, (\ref{eq:modelh}) is a seemingly innocuous regression
of $y_{i,t+h}$ on another variable $x_{i,t}$. It turns out that
(\ref{eq:DGP}) and (\ref{eq:AR1}) consist of a two-equation panel
\emph{vector} AR model. When the within-group transformation is used
to eliminate the fixed effects, the Nickell bias is present even though the lagged $y$ does not explicitly appear on the right-hand
side of (\ref{eq:modelh}). 

The following proposition characterizes the bias in the
asymptotic distribution.

\begin{prop} \label{thm:biasFEAR1} 
Suppose the zero-mean innovations
$u_{i,t}^{y}$ and $u_{i,t}^{x}$ are i.i.d.~across $i$ and $t$. If

\begin{equation}
s_{x}^{2}:=\dfrac{1}{NT_{h}}\sum_{i\in[N]}\sum_{t\in\mathcal{T}^{h}}\tilde{x}_{i,t}^{2}\stackrel{p}{\to}\sigma_{x}^{2}>0\label{eq:plim_varx}
\end{equation}
\begin{equation}
\dfrac{1}{\sqrt{NT_{h}}}\sum_{i\in[N]}\sum_{t\in\mathcal{T}^{h}}\left(\tilde{x}_{i,t}e_{i,t+h}^{(h)}-\mathbb{E}\left[\tilde{x}_{i,t}e_{i,t+h}^{(h)}\right]\right)\stackrel{d}{\to}\mathcal{N}\left(0,\sigma_{xe,h}^{2}\right)\label{eq:clt_xe_bias}
\end{equation}
either as $N\to\infty$ with a fixed $T$, or $(N,T)\to\infty$ jointly, then
\begin{equation}
\sqrt{NT_{h}}\left(\hat{\beta}^{(h){\rm \mathrm{fe}}}-\beta^{(h)}\right)+\beta^{(0)}\cdot\dfrac{\sigma_{u_{x}}^{2}}{s_{x}^{2}}\sqrt{\dfrac{N}{T_{h}}}f_{T,h}(\rho)\stackrel{d}{\to}   \mathcal{N}\left(0,\ \frac{\sigma_{xe,h}^{2}}{\sigma_{x}^{4}}\right),
\label{eq:biasFEAR1}
\end{equation}
where 
\[
f_{T,h}(\rho)=\frac{(1-\rho^{h})}{\left(1-\rho\right)^{2}}-\frac{h}{\left(T-h\right)\left(1-\rho\right)^{2}}+\frac{\rho^{T-2h+1}\left(1-\rho^{2h}\right)}{\left(T-h\right)\left(1-\rho\right)^{2}(1-\rho^{2})},
\]
$\sigma_{u_{x}}^{2}=\mathrm{var}\left[u_{i,t}^{x}\right]$ and 
$\sigma_{x}^{2}=\mathrm{var}\left[x_{i,t}\right]$. 
\end{prop}

\begin{rem}
Simplifying assumptions are commonly employed in the large-$N$-large-$T$ panel data literature to help make the key point clear.
Here, the i.i.d.~innovations keep concise the expressions of the bias 
and variance, and the high-level conditions in
(\ref{eq:plim_varx}) and (\ref{eq:clt_xe_bias}) 
avoid technical distractions highlighted in \citet{phillips1999linear}'s panel data joint $(N,T)$ asymptotics. 
\end{rem}

Let $\mathcal{Z}\sim\mathcal{N}\left(0,\sigma_{xe,h}^{2}/\sigma_{x}^{4}\right)$
denote the zero-mean normal random variable following the limit distribution in (\ref{eq:biasFEAR1}).  Proposition \ref{thm:biasFEAR1} implies 
\begin{equation}
\hat{\beta}^{(h){\rm \mathrm{fe}}}\stackrel{a}{\sim}\beta^{(h)}-\dfrac{\beta^{(0)}\sigma_{u_{x}}^{2}}{T_{h}s_{x}^{2}}f_{T,h}(\rho)+\frac{\mathcal{Z}}{\sqrt{NT_{h}}},\label{eq:FE error order}
\end{equation}
where ``$\stackrel{a}{\sim}$'' signifies asymptotic similarity.
The bias is of order $1/T$. In terms of point estimates, if $T$ is fixed, $\hat{\beta}^{(h){\rm \mathrm{fe}}}$
is inconsistent as $N\to\infty$.
When $(N,T)\to\infty$ jointly with $N/T\to c \in (0,\infty)$, the FE estimator
$\hat{\beta}^{(h){\rm \mathrm{fe}}}$ is consistent for $\beta^{(h)}$, and hence $\rho$, $\sigma_{u_{x}}^{2}$, $\sigma_{x}^{2}$
and $\sigma_{xe,h}^{2}$ are also consistently estimable. However, when statistical inference is the purpose, we need to examine the bias in the asymptotic distribution.  When $N/T\to c \in (0,\infty)$, the bias  in (\ref{eq:biasFEAR1}) satisfying 
\begin{equation}
-\beta^{(0)}\cdot\dfrac{\sigma_{u_{x}}^{2}}{s_{x}^{2}}\sqrt{\dfrac{N}{T_{h}}}f_{T,h}(\rho)\to-\text{\ensuremath{\beta^{(0)}\cdot\dfrac{\sigma_{u_{x}}^{2}}{\sigma_{x}^{2}}}\ensuremath{\ensuremath{\sqrt{c}}}}\cdot\frac{(1-\rho^{h})}{\left(1-\rho\right)^{2}}\label{eq:bias_limit}
\end{equation}
does not vanish asymptotically and it will distort the test size of the usual inference based on the $t$-statistic.\footnote{
When the lagged dependent variables are present, 
the well-known explicit Nickell bias in panel LP is pointed out by \citet{teulings2014economic}. 
They augment the regression specification as a way to fix the bias, which is valid when the regressors have no dynamics.
By contrast, we show that intrinsic Nickell bias occurs even when the lagged dependent variable is not included.} This size distortion is not caused by finite sample standard errors, and it therefore remains even if the asymptotic variances are known. Following \citet{hahn2002asymptotically}
and \citet{okui2010asymptotically}, one can correct the bias based
on (\ref{eq:biasFEAR1}) for valid asymptotic inference.

\medskip{}

\begin{rem}
\label{rem:attentuation} LP is a series of regressions where the direction
and the magnitude of bias depend on many population parameters in
the model. Consider $\rho>0$ as in most economic autoregressive relationships,
and we summarize the following features that are peculiar to the Nickell
bias in panel LP. 
\begin{enumerate}
\item The limit expression (\ref{eq:bias_limit}) implies that given all
the parameters in the DGP, the bias worsens with large $h$, as
$1-\rho^{h}$ increases with $h$. Also, the bias becomes more substantial as $\rho$ increases. 

\item The bias shrinks the FE estimator toward zero as (\ref{eq:FE error order})
becomes 
\[
\hat{\beta}^{(h){\rm \mathrm{fe}}}\stackrel{a}{\sim}\beta^{(h)}\left(1-\frac{1}{T_{h}}\cdot\dfrac{\sigma_{u_{x}}^{2}f_{T,h}(\rho)}{\sigma_{x}^{2}\rho^{h}}\right)+\frac{\mathcal{Z}}{\sqrt{NT_{h}}} \]
by noticing the true IRF $\beta^{(h)}=\beta^{(0)}\rho^{h}$ and $0<\sigma_{u_{x}}^{2}f_{T,h}(\rho)/\left(\sigma_{x}^{2}\rho^{h}\right)=O(1)$.
No matter whether the true $\beta^{(h)}$ is positive or negative,
the FE estimator exhibits an \emph{attenuation bias} and underestimates
$|\beta^{(h)}|$.
\item For statistical inference, the bias term in (\ref{eq:biasFEAR1}) suggests that a larger relative magnitude of sample sizes $N/T$ causes more severe distortion of the coverage probabilities of FE's confidence intervals. Therefore, the impact of Nickell bias on inference is more substantial with a larger $N$ and smaller $T$. 
\end{enumerate}
\end{rem}
In simulation studies and empirical applications, we find that these
intriguing phenomena of the FE estimator are common in the simple
AR(1) specification as well as more general cases. Although the true
DGPs in real data studies are unknown, in the empirical applications
in Section \ref{sec:Empirical-Application} we observe that the bias
correction enlarges the magnitude of the IRFs and the biggest discrepancy
often occurs on a relatively long horizon. 

\subsection{Main Equation Based on Panel VAR{\normalsize{} \label{subsec:General-Setting}}}

While we have demonstrated the bias in the FE estimation of (\ref{eq:modelh})
from the prototype DGP (\ref{eq:DGP}) and (\ref{eq:AR1}), the Nickell
bias looms in general dynamic systems. In practical use of the panel
AR model, researchers may want to include lagged dependent variables
as well as other control variables. For example, \citet[p.1424]{nickell1981biases}
considers a panel ARX (in our notation) $y_{i,t+1}=\mu_{i}^{y}+\gamma y_{i,t}+\beta x_{i,t+1}+u_{i,t+1}^{y}$.

Though LP focuses on single-equation regressions instead of simultaneous-equation
systems, it is helpful to present the underlying VAR system which implies
the series of regressions. In this section, we consider data
generated from a panel VAR. Let $\mathbf{x}_{i,t}$ be a $K$-dimensional
random vector, and the observed data at time $t$ is combined into
a $(K+1)$-vector $\mathbf{w}_{i,t}=(y_{i,t},\mathbf{x}_{i,t}^{\prime})^{\prime}$.
We write down a panel structural VAR($p$) model 
\begin{equation}
\mathbf{A}_{0}\mathbf{w}_{i,t+1}=\boldsymbol{\mu}_{i}^{(0)}+\sum_{s=1}^{p}\mathbf{A}_{s}\mathbf{w}_{i,t+1-s}+\mathbf{u}_{i,t+1}\label{eq:compact_VAR}
\end{equation}
as the DGP, where $\mathbf{A}_{s}$, $s=0,1,\ldots,p$, are $(1+K)\times(1+K)$
coefficient matrices,\footnote{Without loss of generality, here we write the numbers of the lags
of all components in $\mathbf{w}_{i,t}$ being the same for all components.
The researcher has the discretion to choose the number of lags for
specific variables either by prior knowledge or by some information
criterion, and in this case $p$ here is viewed as the largest number
of lags with the shorter lags being imputed by known zero coefficients.
} and $\boldsymbol{\mu}_{i}^{(0)}$ is the vector of individual-specific
fixed effects. The Wold-causal order requests the left-bottom block of
$\mathbf{A}_{0}$ to be a $K$-vector of zeros \citep{jorda2005estimation},
and we standardize its diagonal line as 1. We rewrite (\ref{eq:compact_VAR})
as 
\begin{equation}
\begin{pmatrix}1 & -\mathbf{a}_{0,yx}^{\prime}\\
\boldsymbol{0} & \mathbf{A}_{0,x}
\end{pmatrix}\begin{pmatrix}y_{i,t+1}\\
\mathbf{x}_{i,t+1}
\end{pmatrix}=\begin{pmatrix}\mu_{i}^{(0)y}\\
\boldsymbol{\mu}_{i}^{(0)x}
\end{pmatrix}+\sum_{s=1}^{p}\begin{pmatrix}a_{s,y} & \mathbf{a}_{s,yx}^{\prime}\\
\mathbf{a}_{s,xy} & \mathbf{A}_{s,x}
\end{pmatrix}\begin{pmatrix}y_{i,t+1-s}\\
\mathbf{x}_{i,t+1-s}
\end{pmatrix}+\begin{pmatrix}u_{i,t+1}^{y}\\
\mathbf{u}_{i,t+1}^{x}
\end{pmatrix},\label{eq:block_VAR}
\end{equation}
where the matrices/vectors are partitioned in a compatible manner.
The structural form of the first equation is
\begin{equation}
y_{i,t+1}=\mu_{i}^{(0)y}+\mathbf{a}_{0,yx}^{\prime}\mathbf{x}_{i,t+1}+(1,\boldsymbol{0}^{\prime})\times\sum_{s=1}^{p}\mathbf{A}_{s}\mathbf{w}_{i,t+1-s}+u_{i,t+1}^{y}.\label{eq:stru_y}
\end{equation}
As derived in Appendix A.1, the predictive
equation for $y_{i,t+h}$ is
\begin{equation}
y_{i,t+h}=\mu_{i}^{(h)y}+\mathbf{W}_{i,t}^{\prime}\boldsymbol{\theta}^{(h)}+e_{i,t+h}^{(h)},\:\ h=1,2,\ldots,H,\label{eq:pred_h}
\end{equation}
where $\mathbf{W}_{i,t}=(\mathbf{w}_{i,t}^{\prime},\mathbf{w}_{i,t-1}^{\prime},\ldots,\mathbf{w}_{i,t-p+1}^{\prime})^{\prime}$
is the $p(K+1)$ long vector of regressors, $\boldsymbol{\theta}^{(h)}$
is the corresponding coefficient vector, and the error term 
\[
e_{i,t+h}^{(h)}=(1,\boldsymbol{0}^{\prime})\sum_{s=0}^{h-1}\mathbf{A}_{0}^{-s}\mathbf{A}_{1}^{s}\mathbf{A}_{0}^{-1}\mathbf{u}_{i,t+h-s}.
\]
We impose the following assumption to ensure that $e_{i,t+1}^{(h)}$ satisfies
weak exogeneity.
\begin{assumption}
\label{assu:indep_general} (a) The times series $(\mathbf{u}_{i,t})_{t\in[T]}$ is independent across $i$,  strictly stationary, and is a martingale difference sequence (m.d.s.) and further satisfies $\mathbb{E}[u_{i,t}^{y}|\mathbf{u}_{i,t}^{x}]=0$.
(b) All roots of the determinant equation 
\[
g(z)=\mathrm{det}\left(\mathbf{A}_{0}-\sum_{s=1}^{p}\mathbf{A}_{s}z^{s}\right)=0
\]
 stay outside of the unit circle on the complex plane.
\end{assumption}

Condition (a) rules out cross-sectional dependence for simplicity.
The m.d.s.~assumption is necessary for the panel LP linear coefficient to be interpreted as IRF. 
The conditional mean
$\mathbb{E}[u_{i,t}^{y}|\mathbf{u}_{i,t}^{x}]=0$
further guarantees that $u_{i,t+1}^{y}$ is mean independent
of all the right-hand side regressors in (\ref{eq:stru_y}). As $e_{i,t+h}^{(h)}$
is a linear combination of $(\mathbf{u}_{i,s})_{s=t+1}^{h}$, it is
mean-independent of past information. 
Condition (b) ensures that
the observed variables $\mathbf{w}_{i,t}$ are strictly stationary
over time, exhibiting no unit root or explosive behavior. Analyzing panel LP with highly persistent panels is technically involved and beyond the scope of this paper. Our follow-up work \citep*{liao2024nickell} develops a novel method to handle persistent panel data. See Appendix C.5 for details.

When $\mathbf{x}_{i,t}$ is multivariate, without loss of generality
we can denote the variable of main economic interest as the first
scalar $x_{i,t}^{(1)}$, and the rest of the vector $\mathbf{x}_{i,t}^{(2)}$
as additional control variables to make more plausible the mean independence
in (\ref{eq:stru_y}) and weak exogeneity in (\ref{eq:pred_h}). In estimation, however, $x_{i,t}^{(1)}$ and $\mathbf{x}_{i,t}^{(2)}$
are symmetric in that they share the same status as regressors accompanying
the potential lagged dependent variables in the predictive equation (\ref{eq:pred_h}).
In the panel LP regression, all regressors in $\mathbf{w}_{i,t}$ incur the Nickell
bias regardless of the number of lags, which is elaborated in Appendix A.2 for a VAR($\infty$) model.

\begin{rem}
This fact has profound implications for the GMM method for the dynamic panel regression using internal IVs \citep{arellano1991some}. The above discussion made clear that if we intend to seek IVs, we must prepare instruments not only for $(y_{i,t+1-s})_{s=1}^{p}$, but also for every variable in $(\mathbf{x}_{i,t+1-s})_{s=1}^{p}$. In practice, many IVs can bring about poor finite sample performance \citep{Roodman2009}, not to mention the weak IV's further complications \citep{Andrews2019}. Thus, the IV approach is not an ideal solution for panel LP.
\end{rem}

Closed-form expressions of the bias, as in Proposition \ref{thm:biasFEAR1}
for the prototype model, are intractable in a full-scale dynamic
system such as (\ref{eq:pred_h}). We suggest an automated bias correction
method in the next section.

\subsection{Split-Panel Jackknife}

Panel LP requires a correct specification of the main equation of
interest, but keeps an agnostic attitude toward the dynamics of $\mathbf{x}_{i,t}$
in the lower block of (\ref{eq:block_VAR}). Unlike the analytic bias
correction, \citet{Dhaene2015}'s SPJ is a data-driven method that
suits well with the single-equation panel LP.

Denote the FE estimator of (\ref{eq:pred_h}) as $\hat{\boldsymbol{\theta}}^{(h)\mathrm{fe}}$,
and let $\hat{\boldsymbol{\theta}}_{a}^{(h)\mathrm{fe}}$ and $\hat{\boldsymbol{\theta}}_{b}^{(h)\mathrm{fe}}$
be the FE estimators using the first half of observations in the set
$\mathcal{T}_{a}^{h}:=\left\{ t\leq T_{h}/2\right\} $ and the second
half of observations in the set $\mathcal{T}_{b}^{h}:=\left\{ T_{h}/2<t\leq T_{h}\right\} $,
respectively. The SPJ estimator is defined as 
\begin{equation}\label{eq:spj}
    \hat{\boldsymbol{\theta}}^{(h){\rm spj}}=2\hat{\boldsymbol{\theta}}^{(h)\mathrm{fe}}-\dfrac{1}{2}\left(\hat{\boldsymbol{\theta}}_{a}^{(h)\mathrm{fe}}+\hat{\boldsymbol{\theta}}_{b}^{(h)\mathrm{fe}}\right).
\end{equation} 

\begin{rem}
There have been many proposed solutions in the literature of Nickell bias, and therefore
it is not our intention to invent yet another new method. 
In a recent research independent of ours,
\citet[April]{dube2023local} is aware of the explicit Nickell
bias in panel LP with the lagged dependent variables included and refers to \citet{chen2019mastering}'s SPJ as
a potential solution. Our contribution here is to verify that
the theory of SPJ goes through in panel LP with the series of predictive
regressions. A key ingredient is our Lemma 1 in the
Appendix, which establishes the order of $\sum_{t\in\mathcal{T}_{a}^{h}}\mathbb{E}\left[\bar{\mathbf{W}}_{i,b}e_{i,t+h}^{(h)}\right]$
that crosses the blocks $\mathcal{T}_{a}^{h}$ and $\mathcal{T}_{b}^{h}$.
This term is peculiar to panel LP and its order depends on $h$, for
$e_{i,t+h}^{(h)}$ is not an arbitrary exogenous shock but the error
term yielded from iterating the first equation of the reduced-form
VAR. 
\end{rem}
To establish the asymptotic properties of the SPJ estimator, we impose
the following Assumption \ref{assu:limit_HJ} that consists of high-level
conditions of the law of large numbers and central limit theorem commonly
seen in the literature of panel data, say \citet{bai2009panel}. Define
\[
\widehat{\mathbf{Q}}_{k}:=\frac{1}{NT_{h}/2}\sum_{i\in[N]}\sum_{t\in\mathcal{T}_{k}^{h}}\tilde{\mathbf{W}}_{i,t}\tilde{\textbf{\ensuremath{\mathbf{W}}}}_{i,t}^{\top},
\]
for $k\in\{a,b\}$ associated with the two halves of the data over
the time dimension, where $\tilde{\textbf{\ensuremath{\mathbf{W}}}}_{i,t}=\textbf{\ensuremath{\mathbf{W}}}_{i,t}-T_{h}^{-1}\sum_{t\in\mathcal{T}^{h}}\textbf{\ensuremath{\mathbf{W}}}_{i,t}$.
Let $\widehat{\mathbf{Q}}:=(\widehat{\mathbf{Q}}_{a}+\widehat{\mathbf{Q}}_{b})/2$.
We use $\boldsymbol{1}\left\{ \cdot\right\} $ to denote the indicator
function.
\begin{assumption}
\label{assu:limit_HJ} There are positive-definite matrices $\mathbf{Q}$
and $\mathbf{R}$ such that 
\[
\mathbf{Q}=\mathrm{plim}_{(N,T)\to\infty}\widehat{\mathbf{Q}}_{a}=\mathrm{plim}_{(N,T)\to\infty}\widehat{\mathbf{Q}}_{b}
\]
and 
\begin{equation}
\frac{1}{\sqrt{NT_{h}}}\sum_{i\in[N]}\sum_{t\in\mathcal{T}^{h}}\left[\mathbf{d}_{i,t}^{*}e_{i,t+h}^{(h)}-\mathbb{E}\left[\mathbf{d}_{i,t}^{*}e_{i,t+h}^{(h)}\right]\right]\stackrel{d}{\to}\mathcal{N}\left(0,\mathbf{\mathbf{R}}\right)\label{eq:assu_CLT_hj}
\end{equation}
as $(N,T)\to\infty$, where 
\[
\mathbf{R}=\lim_{(N,T)\to\infty}\frac{1}{N}\sum_{i\in[N]}\mathbb{E}\left[\frac{1}{T_{h}}\sum_{t,s\in\mathcal{T}^{h}}\mathbf{d}_{i,t}^{*}\mathbf{d}_{i,s}^{*\top}e_{i,t+h}^{(h)}e_{i,s+h}^{(h)}\right]
\]
with 
$$\mathbf{d}_{i,t}^{*}=(\textbf{\ensuremath{\mathbf{W}}}_{i,t}-\bar{\textbf{\ensuremath{\mathbf{W}}}}_{i,b})\cdot\boldsymbol{1}\left\{ t\in\mathcal{T}_{a}^{h}\right\} +(\textbf{\ensuremath{\mathbf{W}}}_{i,t}-\bar{\textbf{\ensuremath{\mathbf{W}}}}_{i,a})\cdot\boldsymbol{1}\left\{ t\in\mathcal{T}_{b}^{h}\right\} 
$$
and $\bar{\textbf{\ensuremath{\mathbf{W}}}}_{i,k}=(T_{h}/2)^{-1}\sum_{t\in\mathcal{T}_{k}^{h}}\textbf{\ensuremath{\mathbf{W}}}_{i,t}$
for $k\in\{a,b\}$. 
\end{assumption}
The SPJ estimator is asymptotically normal without a bias term if
$T$ is non-trivial relative to $N$ in that $N/T^{3}\to0$.  
\begin{thm}
\label{thm:hj} If Assumptions \ref{assu:indep_general} and \ref{assu:limit_HJ}
hold, then 
\begin{equation}
\sqrt{NT_{h}}\left(\hat{\boldsymbol{\theta}}^{(h){\rm spj}}-\boldsymbol{\theta}^{(h)}\right)\stackrel{d}{\to}\mathcal{N}(\boldsymbol{0},\mathbf{Q}^{-1}\mathbf{R}\mathbf{Q}^{-1})\label{eq:CLT_HJ}
\end{equation}
as $(N,T)\to\infty$ and $N/T^{3}\to0$. 
\end{thm}

The asymptotic distribution centered around zero allows us to invoke
the common procedure for inference. To construct a feasible $t$-statistic,
we compute the 
\[
\widehat{\mathbf{R}}=(NT_{h})^{-1}\sum_{i\in[N]}\sum_{t,s\in\mathcal{T}^{h}}\mathbf{d}_{i,t}^{*}\mathbf{d}_{is}^{*\top}\hat{e}_{i,t+h}^{(h){\rm spj}}\hat{e}_{i,s+h}^{(h){\rm spj}},
\]
where $\hat{e}_{i,t+h}^{(h){\rm spj}}=\tilde{y}_{i,t+h}-\tilde{\textbf{\ensuremath{\mathbf{W}}}}_{i,t}^{\top}\hat{\boldsymbol{\theta}}_{h}^{{\rm spj}}$
is the SPJ's estimated residual. Given some standard assumptions, $\widehat{\mathbf{R}}$
consistently estimates the individual-clustered variance $\mathbf{R}$.
The standard practice of statistical inference based on the $t$-statistic
is asymptotically valid.

The condition $N/T^3 \to 0$ for SPJ accommodates the leading case of $N/T \to c \in (0, \infty)$ in macro-finance applications. While not designed for large-$N$-small-$T$ panels, SPJ effectively eliminates bias and ensures valid inference under moderate $N/T$ ratios, as evidenced by our numerical and empirical results.
Real data has a finite number of observations. 
As a rule of thumb for empirical applications, we recommend our method for panels with $T\geq 30$ and $N/T \leq 10$.

\begin{rem}\label{rem: extensions}
The theory in this section covers the case with cross-sectional FE
and one-way clustered standard error only. There are many alternative
specifications in empirical applications. For example, researchers
may want to allow cross-sectional correlation in the standard error \citep{pesaran2015testing,juodis2022incidental},
to add time fixed effects to control temporal heterogeneity, and to
accommodate unbalanced panel data. To widen the applicability of the
SPJ approach, Appendix B.2 discusses the two-way
clustered standard error, the \citet{driscoll1998consistent} covariance matrix estimator, the two-way fixed effects, and the observation-splitting
procedure for unbalanced panels.
\end{rem}

\begin{rem}\label{rem: more split} 
Regarding the way of data splitting, \citet[Theorem 3.1]{Dhaene2015} point out that partitioning the full panel into two or more sub-panels over the time dimension and assigning proper weights to them can remove the Nickell bias and keep the same asymptotic variance. The half-panel splitting is the simplest and most intuitive implementation.
\end{rem}

\begin{rem}\label{rem: bias-variance}

Since our theory is based on asymptotics, we can compare the estimators by referring to features of the respective asymptotic distributions. 
\citet{Dhaene2015} shows that FE and SPJ have an asymptotically normal distribution with the same variance. However, the asymptotic normal distribution of the SPJ is centered at 0, whereas the location of that of FE is non-zero.\footnote{\label{footnote: mse} According to (\ref{eq:FE error order}), under the prototype model the asymptotic MSE of FE is ${\rm Asym.MSE}(\hat\beta^{(h){\rm fe}}) = \left[\beta^{(0)}\sigma_{u_x}^2f_{T,h}(\rho)/(T_h s_x^2)\right]^2 + \sigma^2_{xe,h}/(\sigma_x^4\cdot NT_h)$. In contrast, ${\rm Asym.MSE}(\hat\beta^{(h){\rm spj}}) =  \sigma^2_{xe,h}/(\sigma_x^4\cdot NT_h)$ , since SPJ is asymptotically centered at zero with the same asymptotic variance as FE. Therefore, ${\rm Asym.MSE}(\hat\beta^{(h){\rm spj}}) < {\rm Asym.MSE}(\hat\beta^{(h){\rm fe}})$ whenever $\beta^{(0)}\neq0$.  This fact is corroborated by the observations that the finite-sample RMSE of FE is larger than that of SPJ in all our simulation exercises.
} 
As a result, SPJ strictly dominates FE in terms of the MSE of the asymptotic distribution, which equals the variance plus the squared bias.
In finite sample though, we may observe that the confidence interval of SPJ is wider than that of FE due to the sampling error stemming from each of the half panels. 

\end{rem}

\subsection{Simulations} \label{sec:Simulations}

To illustrate the presence of Nickell bias and the finite sample performance
of SPJ, we conduct Monte Carlo simulation exercises based on the prototype model
of the simple form as in Equations (\ref{eq:DGP}) and (\ref{eq:AR1}). We
generate the innovations $\varepsilon_{i,t}$ and $u_{i,t}$ from independent
$\mathcal{N}(0,1)$, and the fixed effect is generated by 
$\mu_{i}^{(0)y}=0.2\sqrt{T}\bar{x}_{i}+\xi_i
$
where $\xi_i\sim\mathcal{N}(0,1)$ is independent of all other variables.
For the contemporaneous connection between $y_{i,t+1}$ and $x_{i,t+1}$, 
we specify $\beta^{(0)}= -0.6$, where the negative relationship is 
motivated by the impact of financial crises on the real economy.
The AR(1) coefficient
$\rho\in\{0,0.2,0.5,0.8\}$ varies the persistence of the predictor.
We consider the sample sizes $(N,T)\in\{(30,60),(30,120),(50,120)\}$
in line with the empirical examples. All simulation results are produced
by 1000 independent replications.

Proposition \ref{thm:biasFEAR1} offers the closed-form formula of the bias for 
the AR(1) specification of $x$.
It provides an ``oracle'' debiased (DB) estimator (See Appendix
A.4 for details) to take advantage of the closed-form
formula, which is based on the oracle of correct specification of AR(1). It will be compared with the plain FE and the SPJ estimators that keep an agnostic attitude about the dynamics of $x$.

Figure \ref{fig:IRF_simul} reports the average of estimated IRFs over the 1000 replications and compares them with the true IRFs. 
With $\beta^{(0)}<0$, the true IRFs are negative and decay to zero exponentially fast. 
The FE estimator exhibits an upward bias, 
indicating a positive Nickell bias with the opposite sign of the true IRFs. 
In other words, the FE estimator underestimates the negative effects of $x_{i,t}$ on $y_{i,t+h}$. Moreover, the bias worsens as the horizon $h$ increases. This finding is predicted by Proposition \ref{thm:biasFEAR1} and discussed in Remark \ref{rem:attentuation}.

Regarding the influence of $\rho$, FE's bias is relatively small when $\rho=0$ or 0.2.
As $\rho$ increases to 0.5 and 0.8, the bias becomes more pronounced. 
The bias is mitigated when $T$ grows from $60$ to $120$. When $N$ grows from $30$ to $50$, the bias becomes larger.
In comparison, the oracle DB estimator, as well as our recommended SPJ estimator, is very close to the true IRF. 

\par When it comes to the estimation errors measured by root-mean-square
error (RMSE), Figure \ref{fig:RMSE_simul} shows that the FE estimator
produces the largest error among the three methods in all cases, while
the oracle DB estimator yields the smallest RMSE. RMSEs of all the
three estimators get larger as $\rho$ gets closer to 1.
The recommended method, SPJ,
again exhibits performance similar to that of the oracle DB estimator in terms of RMSEs.

Figure \ref{fig:coverage} plots the empirical coverage 
probabilities of the confidence intervals (CIs) by inverting the $t$-statistic. 
The nominal 95\% probability
is marked by the horizontal dashed line. 
For each $(N,T)$
and $\rho$, the coverage rate of CIs from the FE estimator is close
to the nominal probability when $h=0$, while the bias emerges when $h>0$
and the coverage rate falls short of 0.95 as $h$ grows. Such Nickell
bias is present even though $\rho=0$ when the IRF equals zero as
$h>0$, and the performance further deteriorates as $x_{i,t}$ gets more persistent under a larger $\rho$.
When $N=30,$ the Nickell bias is alleviated as $T$ grows from 60
to 120, while it materializes again as $N$ increases to $50.$ These
numerical findings echo the closed-form expression of the bias in Proposition
\ref{thm:biasFEAR1}. The CIs from the DB estimator, as an oracle
estimator when the closed-form expression of the Nickell bias is known, have
coverage probabilities very close to 95\% in all cases. 
The SPJ estimator closely resembles the oracle DB estimator,
except for the slight deviation in the case of a small sample size $(N,T)=(30,60)$
and $\rho=0.8$; here SPJ still substantially mitigates the severe
bias from the FE. To summarize, the simulation results provide a clear picture of the
bias issue for the FE estimator in panel LP, and support the
practical values of the SPJ estimator in bias correction. 

If \label{para: bias in inference} statistical inference is of key research interest, Figure \ref{fig:coverage} shows that the Nickell bias in FE causes substantial under-coverage of the confidence interval (or equivalently, excessively high rejection rate in the $t$-test). Such distortion violates the principle of statistical inference and results in misleading empirical findings. Therefore, the Nickell bias must be removed in statistical inference, and SPJ is preferred to FE in this sense.

\begin{figure}
\begin{centering}
\includegraphics[width=1\textwidth]{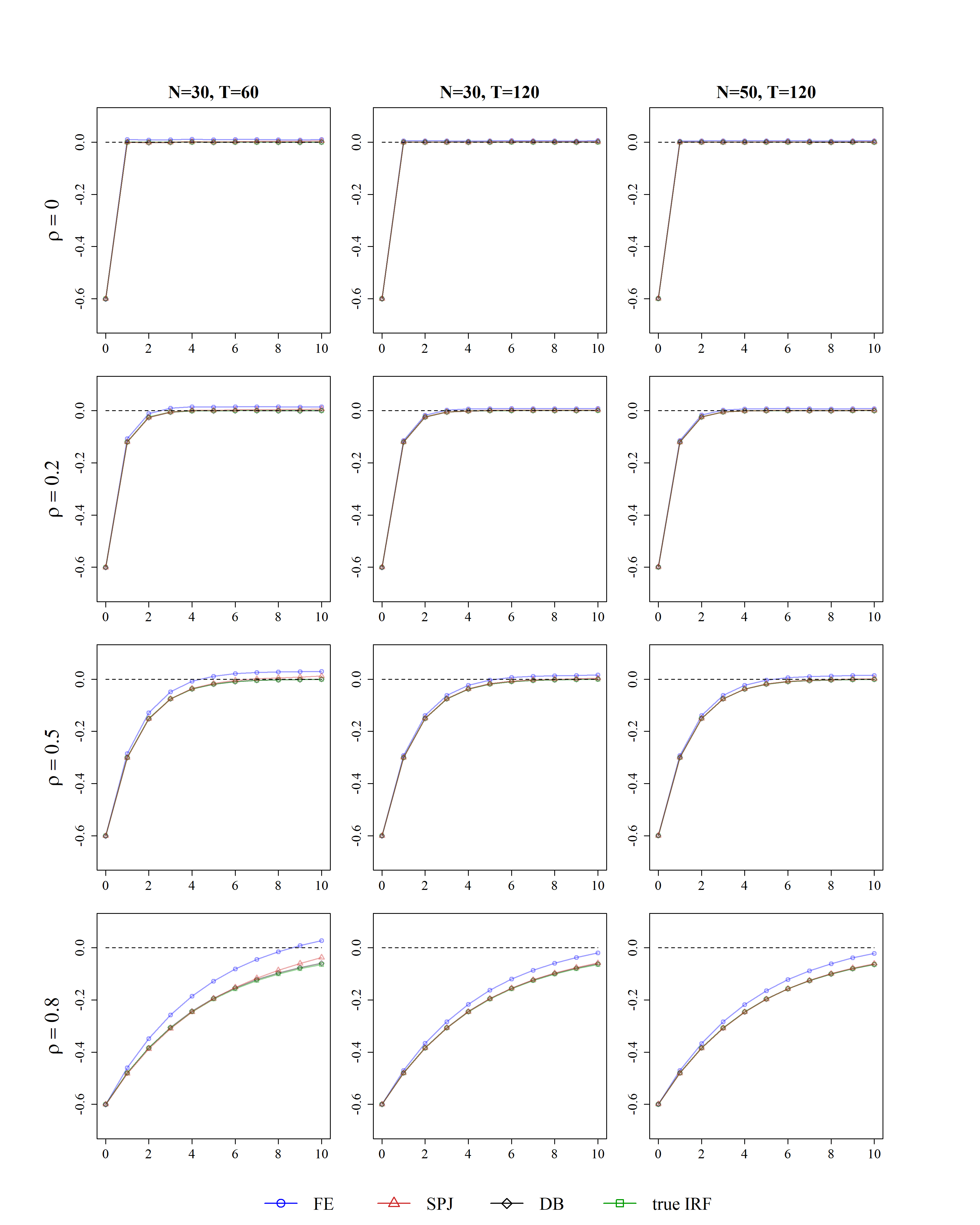}
\end{centering}
\footnotesize Note: In this figure SPJ (red), DB (black) and the true IRF (green) mostly overlap.
\caption{\label{fig:IRF_simul} Estimated IRFs Averaged Over Replications}
\end{figure}

\begin{figure}
\begin{centering}
\includegraphics[width=1\textwidth]{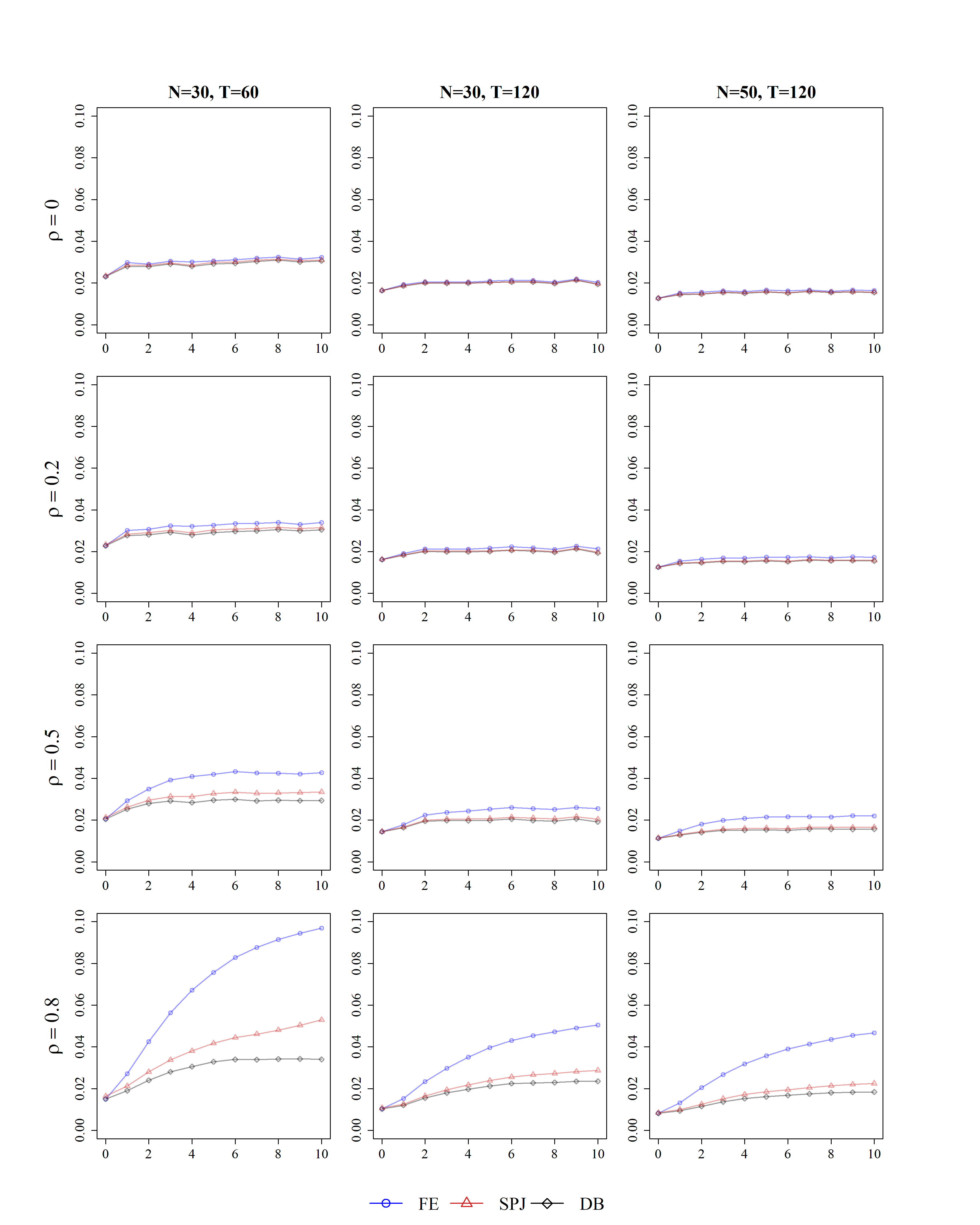}
\end{centering}
\caption{\label{fig:RMSE_simul} RMSEs of the Three Estimators}
\end{figure}

\begin{figure}
\begin{centering}
\includegraphics[width=1\textwidth]{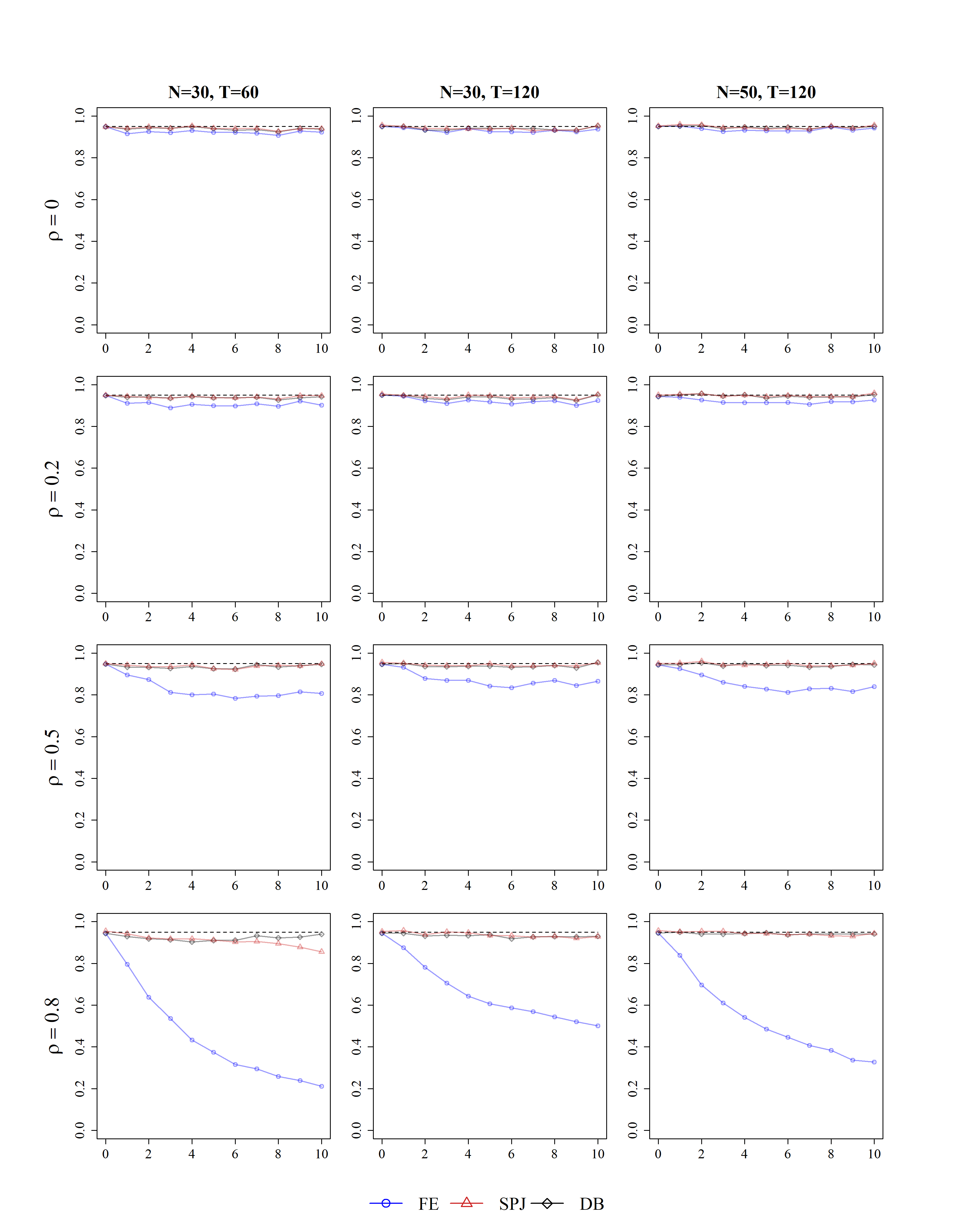}
\end{centering}
\footnotesize Note: The nominal level of 0.95 is marked by the horizontal
dashed line. 

\caption{\label{fig:coverage} Coverage Probability of the Confidence Interval Based on $t$-Statistic  }
\end{figure}

\begin{rem} \label{rem: add simul}
\par Besides the prototype model of Equations (\ref{eq:DGP}) and (\ref{eq:AR1}), 
real-data applications involve numerous specifications, such as cross-sectional correlations, the lagged dependent variable as a control for the dynamics of $y_{i,t}$, the time fixed effect, and the change of outcome $\Delta_h y_{i,t+h} = y_{i,t+h} - y_{i,t}$ as the dependent variable. These variations of panel LP are considered in our empirical applications in Section \ref{sec:Empirical-Application}. In addition, empirically-oriented simulations are useful for examining the robustness of SPJ to the structures of real data that deviate from our prototype model. We therefore conduct additional simulations to further justify the use of SPJ. To save space, we detail these additional simulation results in Appendix C.
These extensive simulations witness the robustness of SPJ under various settings, thereby preserving the advantages of LP in panel data.
\end{rem}

\begin{rem}
    In Appendices C.1 and C.3, we compare SPJ with other alternative methods that remove the Nickell bias, including the approximate bias correction \citep{herbst2021bias}, the likelihood estimator \citep{alvarez2022robust}, and the differenced GMM estimator \citep{arellano1991some}. The approximate bias correction and the likelihood estimator can be more efficient in finite samples, while they highly depend on analytical formulas that require case-by-case derivations. The GMM estimator suffers from weak IV issues when the regressor becomes moderately persistent, thereby exhibiting unstable performance. In contrast, our SPJ is easily implemented and shows robust performance. 
    The validity of weak-IV-robust estimators is currently unknown for panel LP; we leave it for future studies. 
\end{rem}

\section{Revisiting the Aftermath of Financial Crises}\label{sec:Empirical-Application}

In this section, we apply our method to the four renowned empirical studies on macro-finance linkage as mentioned in the Introduction. They all center on the aftermath of financial crises on output, but each of them places a distinct emphasis
on different financial shocks. 
\citet{romer2017new} design a new financial crisis chronology to discuss the responses of output and unemployment following the overall financial distress. \citet{baron2021banking} explore the role of banking crises in the credit crunch and output reduction  \citep{diamond1983bank,bernanke2018real}.
The third one and its related studies \citep{mian2017household,mian2010household}
highlight the rising household debts in depressing output in the mid-run. The fourth study \citet{cerra2008growth} examines the persistent damage of currency crises on national economies. These studies employ cross-country panel data and adopt the panel LP regression Eq.(\ref{eq:pred_h}):
$y_{i,t+h}$ denotes some measure of economic output,
and the regressor vector $\mathbf{W}_{i,t}$ includes a key variable of interest $x_{i,t}^{(1)}$---a certain measure of financial shocks.\footnote{The first three studies adopt the panel LP method, and the fourth study \citet{cerra2008growth} adopts the dynamic panel with distributed lag model, which also suffers the Nickell bias. The FE estimates are closely aligned with their estimation results for the currency crises, which will be discussed later in detail.} We revisit these works and contribute a methodological edge to the literature on financial crises.

\subsection{Financial Distress: \texorpdfstring{\citet{romer2017new}}{Lg} \label{subsec:rr}}

Crisis chronology can be traced back to the records
of \citet{caprio1996bank} about bank crisis events in the 1990s.
Thereafter, \citet{reinhart2009aftermath},
\citet{schularick2012credit}, \citet{laeven2013systemic}, and \citet{jorda2013credit} have adopted more precise subjective standards and quantitative data to measure financial crises and evaluate the associated economic losses. In a comprehensive study, \citet{romer2017new} (RR, henceforth) create a narrative semiannual measure of financial distress for 24 advanced countries from 1967 to 2012, based on the contemporaneous narrative accounts of country conditions on the disruptions to credit supply in the OECD Economic Outlook. To capture the variation in financial disruption across countries and time periods, they classify financial distress into five categories: credit disruption, minor crisis, moderate crisis, big crisis, and extreme crisis, and each category is further broken into minus, regular, and plus. Thus, their new index of financial distress has a scale of 16, where 0 represents no financial distress, and a higher value indicates a more severe financial crisis.
Based on this novel measure of general financial distress, they use FE to estimate panel LP and find that the average decline in output after a financial crisis is moderate in size, though persistent and statistically significant.

We revisit RR's open-access semi-annual dataset, an unbalanced panel of 24 OECD countries.
RR construct the financial distress index from 1967 to 2012, but between 1967 and 1979 only Germany had a ``Credit disruption (regular)'' in 1974. Our analysis restricts the time period from 1980 to 2012; this sample adjustment barely changes their estimation results.

RR's panel LP specification is
\begin{equation}
y_{i,t+h}=\mu_{i}^{(h)y}+g_{t}^{(h)}+\beta^{(h)}x_{i,t}^{{\rm FD}}+\sum_{j=1}^{4}\tau^{(h)}_{j}x_{i,t-j}^{{\rm FD}}+\sum_{k=1}^{4}\eta^{(h)}_{k}y_{i,t-k}+e^{(h)}_{i,t+h}, \label{eq:rr}
\end{equation}
for $h=0,1,...,10$,\footnote{
The only difference from Eq.(\ref{eq:pred_h}) is the presence of the time fixed-effect $g_t^{(h)y}$; See Sections B.2.3 and C.3 for such an extension.}
where the dependent variable $y_{i,t+h}$ is the logarithmic real GDP or the unemployment rate, and 
$x_{i,t}^{{\rm FD}}$ is their new index on financial distress. $\mu_{i}^{(h)y}$ and $g_{t}^{(h)}$ indicate country and time fixed effects for each horizon $h$. Moreover, four lags of the dependent variables and the financial distress index are included as control variables. 
The linear coefficient $\beta^{(h)}$ is the IRF of the economic activities in period $h$ ahead of financial shocks occurring at time $t$, which corresponds to a minor level shift in financial distress as RR defines. To obtain a typical shock in financial distress, RR multiplies the estimates of $\beta^{(h)}$ by 7, which corresponds to a ``Moderate crisis (minus)'' financial crisis.

Figure \ref{fig.rr_f4}'s panels (A) and (B) show the IRFs of the real GDP and the unemployment rate to a typical financial shock, respectively, estimated by RR's FE and our recommended SPJ.
SPJ confirms RR's finding that output falls and the unemployment rate rises significantly and persistently following the financial crisis, but its impulse responses exhibit considerable differences from those of the FE estimates, particularly for the long horizons. 

Panel (A) shows that the immediate aftermath of a moderate crisis is a fall in GDP of 2.2\% based on the SPJ estimate, which is close to the FE estimate. However, the differences between the SPJ and the FE estimates become more pronounced as the horizon $h$ rises. According to SPJ, the decline in output grows substantially 3.5 years after the financial shock and peaks at 6.3\%, which is about 16\% higher than the FE estimates. The differences between the two estimates are also persistent. After 5 years ($h=10$), the decline in output based on SPJ returns to about 5\%, which remains 28\% larger in magnitude than the corresponding FE. 

Consistent with the responses of real GDP, Panel (B) suggests that the FE estimator also tends to underestimate the aftermath of a moderate financial crisis on the unemployment rate. A financial shock raises the unemployment rate steadily. As with GDP, the unemployment rate continued to increase through 3.5 years following the impulse, peaking at 2.6 percentage points based on the SPJ estimate (24\% higher than the FE estimate), and after that, the differences between the two estimates remain.

\begin{figure}[ht]
\begin{centering}
\includegraphics[width=1\textwidth]{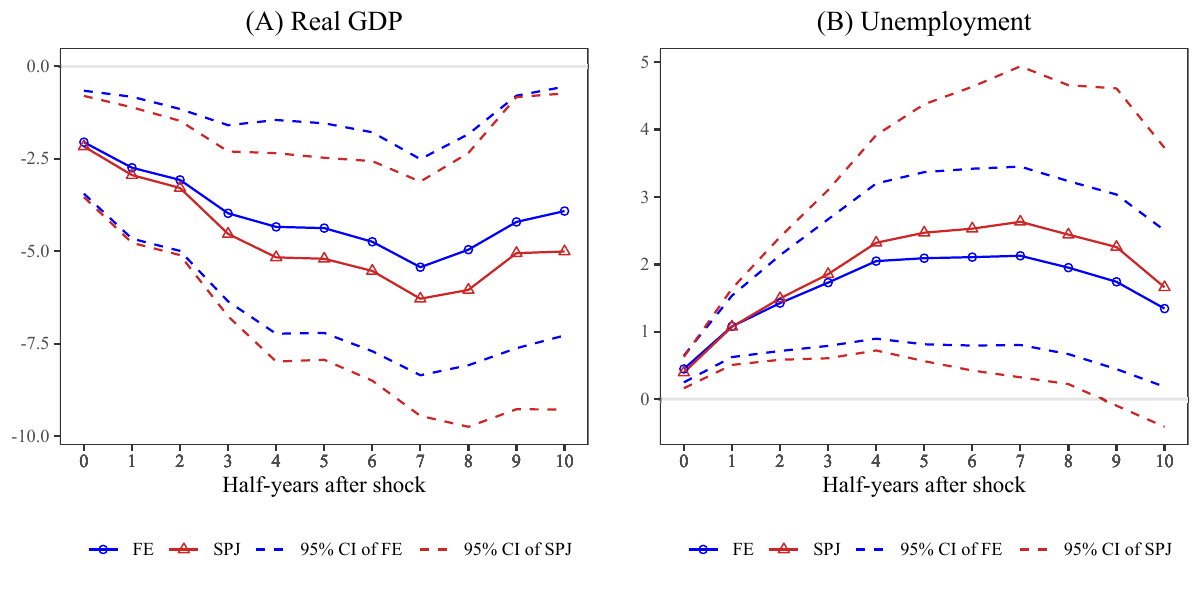} 
\end{centering}

\footnotesize Note: Panel (A) and Panel (B) present the impulse responses of the logarithmic  real GDP and the unemployment rate to financial distress, respectively. The impulse responses are estimated from Model (\ref{eq:rr}) and multiplied by 7 to indicate the impact of a “Moderate
crisis (minus)” financial crisis. The blue solid lines are the FE estimates as in RR, while the red solid lines indicate the SPJ estimates. The dashed lines show the 95\% confidence intervals based on standard errors clustered on country.

\caption{Impulse Responses of Real GDP and the Unemployment Rate to Financial Distress}
\label{fig.rr_f4}
\end{figure}

Overall, FE tends to underestimate the IRFs, which is likely caused in part by the persistence of financial distress. 
To check the persistence of the financial distress index, we adopt a simple panel AR(1) model:
$x_{i,t}^{{\rm FD}}=\mu_{i}+g_{t}+\rho x_{i,t-1}^{{\rm FD}}+e_{i,t},$ 
where $x_{i,t}^{{\rm FD}}$ is the financial distress index, and $\mu_{i}$ and $g_{t}$ denote country and year fixed effects. This regression yields the autoregressive coefficient \label{para: RR rho} $\hat{\rho}= 0.836,$ 
 $(s.e.=0.031, R^2= 0.8177)$, comparable to $\rho = 0.8$ in our simulation. Thus, the financial distress index is persistent, which may play a role in the differences between the FE and SPJ estimates.
 Though each of these two estimators falls into the 95\% confidence interval of the other estimator, 
it is still important to correct the Nickell bias in FE. As demonstrated by the simulation results in Section \ref{sec:Simulations}, the coverage of confidence intervals of the FE estimates can be severely distorted (Figure \ref{fig:coverage}) in the presence of Nickell bias (Figure \ref{fig:IRF_simul}), and the conventional hypothesis testing for the null hypothesis $\mathbb{H}_0:\beta^{(h)}=0$ using the $t$ statistics by FE is invalid.

\subsection{Banking Crises: \texorpdfstring{\citet{baron2021banking}}{Lg}}
\label{subsec: Baron2021} 
The banking system is vital to the modern financial system, and a vast literature has studied the financial intermediary of banks and how banking crises led to a sustained decline in economic activities \citep{diamond1983bank,calomiris2003fundamentals, gertler2010financial, brunnermeier2014macroeconomic, rampini2019financial}. Ben Bernake argues that bank runs played a critical role in the Great Depression of the 1930s, and the collapse of Lehman Brothers triggered global financial panics, which eventually turned into a worldwide financial crisis during 2007--2008 and the prolonged Great Recession \citep{bernanke1983nonmonetary, bernanke2018real}. The recent 2023 meltdown of Silicon Valley Bank, the second largest failure of a financial institution in U.S. history, again spurred worries about the global financial markets.

\citet{baron2021banking} (BVX, henceforth) study the aggregate output loss associated with bank crises.
Different from the traditional approach that takes banking panics as banking crises, BVX argue that a large decline in bank equity return is the prerequisite for banking crises.\footnote{Many scholars have emphasized the criticality of bank equity. \citet{gertler2010financial} introduce financial intermediaries into their model, showing that disruptions in financial intermediation depress aggregate economic activity. \citet{he2013intermediary} find the asymmetry of capital market risk premiums during banking crises, that is, when bank equity is low, bank losses have a significant impact on the risk premium, but when bank equity is high, it has little effect. \citet{brunnermeier2014macroeconomic} suggest that a decline in bank assets through the price channel creates greater risk as the economy moves away from a steady state. \citet{rampini2019financial} provide a dynamic model including bank net assets and find that low liquidity in banks will lead to more severe and prolonged recessions during banking crises.} 
Specifically, they construct a new data set of bank equity return for 46 advanced and emerging economies over 1870--2016, and define a bank crisis as a bank equity crash of more than 30\%. 
Their FE estimates find that bank equity crashes predict a sizable and long-lasting decline in future real GDP.

The baseline LP specification in BVX is
\begin{equation}
\begin{split}\Delta_{h}y_{i,t+h}=\mu_{i}^{(h)y} & +\beta^{(h)\mathrm{B}}x_{i,t}^{{\rm B}}+\beta^{(h)\mathrm{N}}x_{i,t}^{{\rm N}}+\sum_{j=1}^{3}\left(\tau^{(h)\mathrm{B}}_{j}x_{i,t-j}^{{\rm B}}+\tau^{(h)\mathrm{N}}_{j}x_{i,t-j}^{{\rm N}}\right)\\
 & +\sum_{k=0}^{3}\eta^{(h)}_{k}\Delta y_{i,t-k}+\sum_{\ell=0}^{3}\zeta^{(h)}_{\ell}\Delta y_{i,t-\ell}^{{\rm credit}}+e^{(h)}_{i,t+h}
\end{split}
\label{eq:bvx_t1_y}
\end{equation}
for $h=1,2,...,6$, where the dependent variable $\Delta_{h}y_{i,t+h}$
is the change in the logarithm of real GDP from year $t$ to $t+h$;
to be consistent with the other two empirical studies in this section, we multiply the change by 100 to represent percentages. 
The regressor $x_{i,t}^{{\rm B}}=\mathbf{1}\{r_{i,t}^{{\rm B}}\leq-30\%\}$
is a dummy variable for the bank equity crash when the bank equity index dropped more than 30\%, and similarly
$x_{i,t}^{{\rm N}}=\mathbf{1}\{r_{i,t}^{{\rm N}}\leq-30\%\}$ is a dummy for the nonfinancial equity crash. 
The control variables include three lags of bank and nonfinancial equity crash dummies, 
and the contemporaneous and up to three-year lagged real GDP growth and 
change in credit-to-GDP.

The dependent variable here is the $h$-order
difference $\Delta_{h}y_{i,t+h}$, 
and therefore the slope coefficients $\beta^{(h)\mathrm{B}}$ and $\beta^{(h)\mathrm{N}}$ 
represent cumulative responses. As detailed in Appendix B.1, 
FE under this specification also suffers from the Nickell bias and underestimates the magnitude of the IRFs, and the bias can be corrected by SPJ. 
Figure \ref{fig.bvx_t1_fe}'s Panel (A) shows 
the output declines significantly and persistently after a banking crisis. Compared with SPJ, FE tends to underestimate output loss,
and the gap widens over time.
The SPJ suggests that the cumulative decline in GDP peaks four years after the bank crisis at 4.2\%, while the corresponding output gap is 3.4\% by FE.
Panel (B) reiterates the real output losses following a crash in nonfinancial equity.
In the first three years after the nonfinancial stock crash, the IRF by SPJ is very close to the one by FE, whereas the differences in the two IRFs become visible afterward. FE again tends to underestimate.

\begin{figure}[ht]
\begin{centering}
\includegraphics[width=1\textwidth]{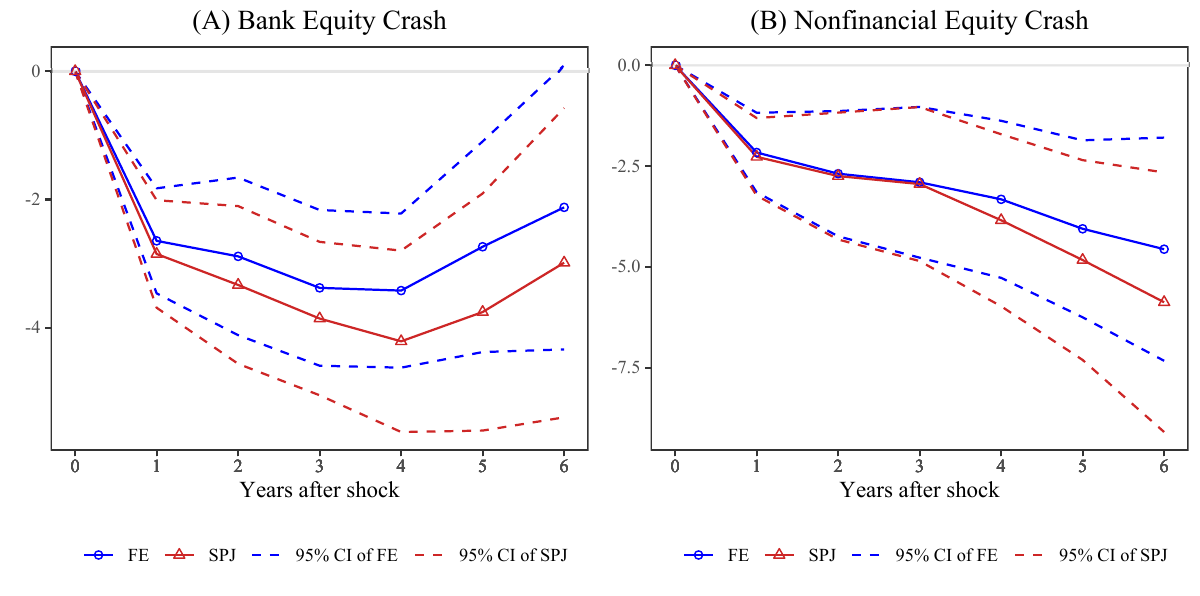} 
\end{centering}
\footnotesize Note: Panel (A) and Panel (B) present the impulse responses of accumulative changes in the logarithmic real GDP to bank and nonfinancial equity crashes respectively, which are estimated from the Model (\ref{eq:bvx_t1_y}). The blue solid lines represent the FE estimates in BVX, while the red solid lines indicate our SPJ. The dashed lines represent 95\% confidence intervals based on standard errors double-clustered on country and year.

\centering{}\caption{Impulse Responses of Real GDP to Bank and Nonfinancial Equity Crashes}
\label{fig.bvx_t1_fe}
\end{figure}

In the baseline specification, we do not control for time fixed effects as they may absorb the global impact of some big banking crises. BVX's FE estimates find that the inclusion of time fixed effects in the panel LP mitigates the output loss. We find similar results by SPJ, and thus the difference between the two estimates slightly declines. On the other hand, the gap between the two estimated IRFs of nonfinancial firms' equity crashes becomes more visible.

\subsection{Household Debt: \texorpdfstring{\citet{mian2017household}}{Lg}}
\label{subsec: Mian2017}

Prior to 2007, the subprime mortgage market in the United States developed rapidly due to the continued prosperity of the U.S. housing market
and the low interest rates. Household debt levels were on the rise, and mortgages were inversely proportional to household income \citep{mian2009consequences}.
The 2007--2008 global financial crisis reminds us of another important source of financial crisis: household debt.
Expansion of household debt boosted consumer demand \citep{mian2011house},
fueled housing speculation \citep{mian2022credit}, and raised
house prices \citep{justiniano2019credit}. 
A bubble cannot last forever. The cooling of the U.S. housing market, especially the hike in short-term interest rates, overloaded repayment burdens upon home buyers, leading to large-scale defaults that eventually triggered the subprime mortgage crisis.\footnote{In addition, the expansion of household debt in most developed economies such as the United States mainly affects the economy through the channel of consumption demand \citep{mian2018finance,mian2020does}. The higher the debt leverage, the greater the responsiveness of household consumption propensity to changes in housing wealth \citep{mian2013household}.} 
The global financial crisis highlights that household debts would not only affect economic fluctuations in the short term but also hamper economic growth in the medium and long term. 

In an influential study, \citet{mian2017household} (MSV, henceforth) explore the dynamics between household debt, nonfinancial firms debt, and economic fluctuations, with a cross-country panel via both VAR and LP. They find that the household debt shock leads to a boom-recession cycle, that is, GDP first increases for 2-3 years but then decreases substantially. More precisely, 
they find that a rise in the household debt to GDP ratio from four years ago to last year predicts a substantial decline in subsequent real GDP growth from the current year onward. By contrast, the output declines immediately following firm debt shocks, whereas recovers gradually afterward. We will compare the baseline panel LP results using the FE estimation in MSV's Figure 2 (p.1770) with the results based on SPJ.

Collecting an unbalanced panel of 30 countries from 1960
to 2012, MSV specify the LP regression as
\begin{equation}
y_{i,t+h}=\mu_{i}^{(h)y}+\beta^{(h)\mathrm{HH}}x_{i,t}^{{\rm HH}}+\beta^{(h)\mathrm{F}}x_{i,t}^{{\rm F}}+\sum_{j=1}^{4}\left(\tau^{(h)\mathrm{HH}}_{j}x_{i,t-j}^{{\rm HH}}+\tau^{(h)\mathrm{F}}_{j}x_{i,t-j}^{{\rm F}}\right)+\sum_{k=0}^{4}\eta^{(h)}_{k}y_{i,t-k}+e^{(h)}_{i,t+h}\label{eq:msv}
\end{equation}
for $h=1,2,...,10$, where $y_{i,t+h}$ is the logarithmic 
real GDP. The key explanatory variables of interest are the household
debt to GDP ratio $x_{i,t}^{{\rm HH}}$ and nonfinancial firm debt
to GDP ratio $x_{i,t}^{{\rm F}}$. Control variables include four
lags of household debt to GDP ratio and nonfinancial firm debt to
GDP ratio, and contemporaneous and up to four-year lagged logarithmic 
real GDP. 

\begin{figure}[ht]
\begin{centering}
\includegraphics[width=1\textwidth]{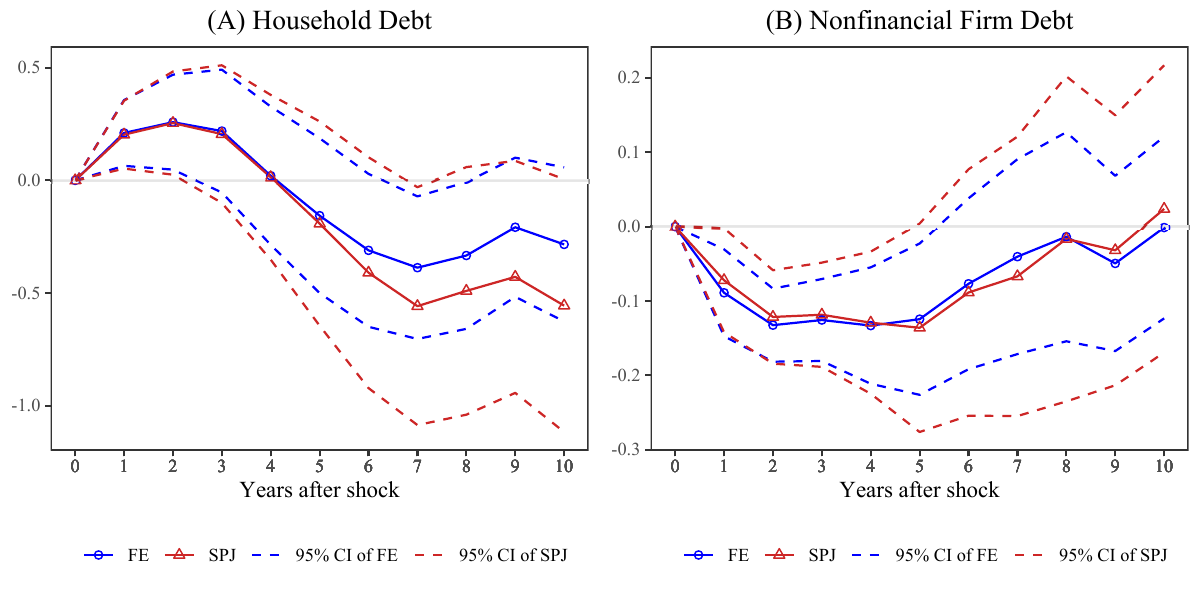} 
\end{centering}
\footnotesize Note: Panel (A) and Panel (B) present the impulse responses of the logarithmic  real GDP to household debt and nonfinancial firm debt, respectively, which are estimated from Model (\ref{eq:msv}). The blue solid lines represent the FE estimates in MSV, while the red solid lines indicate our SPJ. The dashed lines represent 95\% confidence intervals based on standard errors double-clustered on country and year.

\centering{}\caption{Impulse Responses of Real GDP to Household Debt and Nonfinancial Firm
Debt}
\label{fig.msv_fe}
\end{figure}

Our empirical analysis supports and reinforces MSV's main findings that household debt booms predict a sustained decline in output after a temporary rise.
SPJ's IRFs in Figure \ref{fig.msv_fe} maintain the boom-recession cycle, with long-term real GDP declines greater than MSV's FE.
In Panel (A) SPJ resembles FE in the first to three years, indicating a temporary rise in GDP due to the housing market boom. The output has started to decline persistently since the fourth year.
The trough is observed in the seventh year, 
where FE estimates a ten percentage point shock in the household debt to GDP ratio is associated with a 3.9\% drop, in contrast with SPJ's 5.6\% drop.
This is consistent with the literature that finds debt booms may distort resource allocation and human capital accumulation, and thus reduce long-run growth \citep{gopinath2017capital, charles2018housing, borio2016labour}. As a comparison, Panel (B) shows that the IRFs of output to firm debt shocks are close.

\subsection{Currency crises: \texorpdfstring{\citet{cerra2008growth}}{Lg}}  
\label{subsec: CS2008}

In our final application, we examine the impact of currency crises on output loss, a significant concern in international economics due to their severe and enduring effects on national economic output \citep{burnside2001prospective, hong2005recovery, cerra2008growth}. These crises, marked by sharp currency depreciations, can trigger widespread economic disruptions and persistent contractions. For instance, the Asian Financial Crisis of 1997-1998 led to substantial economic contractions in Southeast Asia, notably in Indonesia, Thailand, and South Korea. Similarly, the Argentine Crisis of 2001-2002 resulted in a severe economic depression and rising unemployment.

We used the data set from \cite{cerra2008growth}, which includes panel data on currency crises for 175 countries from 1965 to 2000. An exchange market pressure index (EMPI) is constructed as the sum of the percentage depreciation in the exchange rate and the percentage loss in foreign exchange reserves, allowing for cross-country comparability. A dummy variable for a currency crisis is assigned to a specific year and country if the EMPI is in the upper quartile of all observations within the panel. Our panel local projection (LP) specification is as follows:

\begin{equation}
\Delta_{h}y_{i,t+h}=\mu_{i}^{(h)y}+\beta^{(h)}x_{i,t}+\sum_{j=1}^{4}\tau_{j}^{(h)}x_{i,t-j}+\sum_{k=1}^{4}\eta_{k}^{(h)}\Delta y_{i,t-k}+e_{i,t+h}^{(h)},
\label{eq:cs}
\end{equation}
for $h=1,2,...,10$. Here, the dependent variable $\Delta_{h}y_{i,t+h}$ represents the change in the logarithm of real GDP from year $t-1$ to $t+h$. The regressor $x_{i,t}$ is a dummy variable indicating a currency crisis. Control variables include four lags of the currency crisis dummies and real GDP growth.

\begin{figure}[ht]
\begin{centering}
\includegraphics[width=0.6\textwidth]{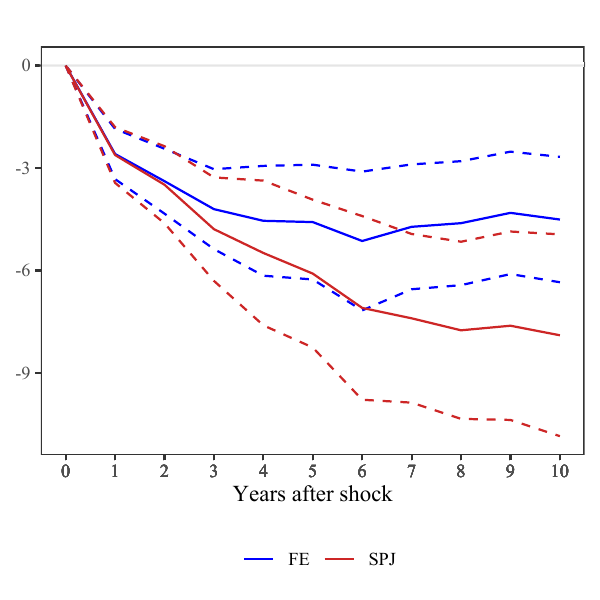}
\footnotesize\flushleft Note: This figure presents the impulse responses of cumulative changes in the logarithm of real GDP to the currency crises. The blue solid lines represent the FE estimates, while the red solid lines indicate our SPJ. The dashed lines represent 95\% confidence intervals based on standard errors.

\centering{}\caption{Impulse Responses of Real GDP to Currency Crises
\label{fig.cs_f3}}
\end{centering}
\end{figure}

Figure \ref{fig.cs_f3} illustrates the cumulative changes in real GDP following currency crises, using estimates from both FE and SPJ. This comparison highlights significant differences in how each method captures the economic impact of such crises.

The FE estimator shows a decline in real GDP immediately after a currency crisis. This decline persists over several years, indicating that currency crises have a lasting negative effect on economic output. However, the FE estimates suggest a moderate initial drop in GDP, which stabilizes relatively quickly. This might provide a somewhat reassuring picture to policymakers, suggesting that while the impact is negative, it is not overwhelmingly severe.

In contrast, the SPJ estimates reveal a more severe and protracted decline in GDP, suggesting that traditional FE estimates may significantly underestimate the true economic cost of currency crises. Quantitatively, the FE estimates show an initial GDP decline of 3\% following a currency crisis, stabilizing around a 4.5\% decrease after five years.\footnote{This quantitative result closely aligns with Figure 3 in \cite{cerra2008growth}, who employ a dynamic panel with distributed lags model, which also suffers from the Nickell bias.} The SPJ estimator, however, indicates a more substantial and persistent output drop, exceeding 7\% after five years. This finding corroborates \cite{hong2005recovery}, who argue that GDP levels remain permanently below their initial trend following currency crises. Furthermore, while the FE estimate shows a peak cumulative output loss of approximately 5.1\% in the sixth year post-crisis, the SPJ estimate suggests a 7.1\% output decline over the same horizon---about 40\% higher than FE. Moreover, SPJ suggests that the peak effect is about 7.9\% cumulative output loss about a decade after the crisis, which is about 75\% higher than that from FE. 

A notable aspect of Figure \ref{fig.cs_f3} is that the FE estimators fall outside the 95\% confidence intervals of the SPJ estimates for long horizons ($h>6$). This statistically significant difference between the SPJ and FE estimators underscores the substantial bias in the FE estimates, emphasizing the importance of robust estimation techniques in capturing the full extent of economic disruptions caused by currency crises. 
Our econometric theory posits that Nickell bias is more severe in panels with a larger number of cross-sectional units $N$ and a smaller number of time periods $T$, which is the case in this example.

\subsection{Summary and Policy Implications\label{subsec: summary empirical}}

Our empirical analysis demonstrates that financial shocks induce substantial output losses, corroborating the consensus that recessions triggered by financial crises are both deeper and more protracted than typical downturns. Moreover, our comparison of FE with SPJ reveals a systematic downward bias in the former, despite the distinct nature of financial crises, as well as variations in measurement and sample coverage. This bias leads to an underestimation of financial crisis costs, usually with large magnitudes in the long horizons. These findings underscore that the Nickell bias is pervasive, and it requires correction for valid inference.

Those case studies are in line with several key findings in our prototype model and simulation exercises.
First, in all four cases, FE exhibits attenuation relative to SPJ, consistent with Remark \ref{rem:attentuation} for the prototype model and Remark B1 in Appendix B.1 for the differenced dependent variable. This result indicates that previous studies on financial crises using FE may underestimate the impact of financial crises. Second, the gaps between SPJ and FE tend to be larger for long horizons than short horizons, indicating that the FE estimator is more likely to underestimate the long-term impact of financial crises. 
Lastly, among all applications, the discrepancies between FE and SPJ are the most statistically significant in \citet{cerra2008growth}, where the relative magnitude of the sample sizes $N/T$ is the largest. This result echoes the bias term in (\ref{eq:biasFEAR1}) of Proposition  \ref{thm:biasFEAR1} that a larger ratio $N/T$ undermines the statistical inference for FE.

Next, we discuss the magnitudes of the Nickell bias in the four cases and its importance for policy design. Table \ref{tab.re_diff} presents the absolute relative differences between the two estimates at the peak period of the output declines following financial shocks. Column (6) shows that the FE estimates of economic contraction following financial distress, banking crises, household debts, and currency crises as key results in those four studies suffer from considerable underestimation, ranging from 16\% to 75\%. Moreover, our results also indicate that the economic recovery from financial crises is slower than the FE estimates suggest. This helps partially explain why the recovery from the 2008 global financial crisis was weak and slow \citep{imf2018world}.

Note that the estimated quantitative effects of financial crises on output vary widely in the literature for various reasons, such as the measurement of financial crisis, the time periods, the sample of countries, and the measurement issues of GDP \citep{sufi2021financial}. 
Our study showcases the important estimation bias arising from the common econometric method used in the literature, and provides better econometric tools for practitioners to achieve more accurate estimates of financial crises on output loss.

Lastly, although the numerical differences between the SPJ and FE estimates are visually small in some empirical applications, they remain highly relevant for policymakers designing effective responses to financial crises. Fiscal stimulus policies are commonly employed to mitigate economic damage following severe financial crises. For example, in response to the subprime mortgage crisis of 2007-2008, the U.S. government under President Obama enacted the American Recovery and Reinvestment Act (ARRA) in 2009, which provided approximately \$831 billion, around 5.8\% of GDP, to avoid further economic deterioration.

Assuming a government spending multiplier of 1.0, a relatively high value according to \cite{ramey2019ten} and \cite{ramey2018government}, the choice of estimation methods between FE and SPJ influences the magnitude of a stimulus package designed to counteract the maximum output loss following a financial crisis.\footnote{\cite{ramey2019ten} conducted a recent survey of fiscal multipliers, finding that government purchase multipliers typically range between 0.6 and 1. A smaller multiplier implies the need for a larger stimulus package to offset the same amount of economic contraction. We select a multiplier of 1.0 for straightforward comparison. However, even under alternative fiscal multipliers, our qualitative conclusion---that the fiscal stimulus is likely insufficient based on FE estimates---remains unchanged.} Take \cite{romer2017new} as an example.  
Following a typical financial crisis, the FE estimates would require a fiscal stimulus of approximately 5.43\% of GDP, which equals the estimated output loss at the peak in Column (5) of Table \ref{tab.re_diff} divided by the fiscal multiplier, while the SPJ estimates recommend a larger stimulus of about 6.29\% of GDP. Consequently, failing to account for Nickell bias---using FE rather than SPJ---would lead policymakers to underestimate the necessary stimulus by roughly 0.85\% of GDP. Our simple calculation indicates that the size of ARRA may be smaller than the desired level, and thus a larger stimulus could have accelerated recovery from the Great Recession.

\begin{table}[ht]
\renewcommand\arraystretch{1.2}
\begin{center}
\caption{Relative Difference in the FE and SPJ estimates of IRFs}\label{tab.re_diff}
\resizebox{1\columnwidth}{!}{
\begin{tabular}{l c c c c c c }
\hline\hline
\multirow{2}*{Paper and Model} & \makecell{Dep. Variable}   & \makecell{Indep. Variable}     & \makecell{$h^{\rm peak}$}  & \makecell{$\widehat{\beta}^{(h)\mathrm{spj}}$}    & \makecell{$\widehat{\beta}^{(h)\mathrm{fe}}$}    & \makecell{Relative Diff.} \\
& (1)   & (2)       & (3)         & \makecell{(4)}   & \makecell{(5)}   & (6) \\
\hline
\multirow{1}*\makecell{\cite{romer2017new}:} \\
\multirow{2}*{Model (\ref{eq:rr})}& Real GDP   & \multirow{2}*{Financial distress}  & 7           & -6.285 & -5.432 & 15.69\% \\
& Unemployment  & & 7           & 2.632  & 2.128  & 23.69\% \\
\multirow{1}*\makecell{\cite{baron2021banking}:} \\
\multirow{2}*{Model (\ref{eq:bvx_t1_y})}& \multirow{2}*{Real GDP growth}   & Bank equity crash  & 4           &-4.213  &-3.422  &23.13\% \\
&    & Nonfinancial equity crash  & 6           &-5.876  &-4.561  &28.84\% \\
\multirow{1}*\makecell{\cite{mian2017household}:} \\
\multirow{2}*{Model (\ref{eq:msv})}& \multirow{2}*{Real GDP}   & Household debt & 7           &-0.558  &-0.387  &44.08\% \\
&    & Nonfinancial firm debt  & 5           &-0.136  &-0.124  &9.30\% \\

\multirow{1}*\makecell{\cite{cerra2008growth}:} \\
Model (\ref{eq:cs})& Real GDP   & Currency crisis & 10   & -7.893  & -4.507  & 75.14\% \\
\hline\hline
\end{tabular}
}
\end{center}
\begin{tablenotes}
\footnotesize
\item Note: Column (3) is the peak period when the economic activities contract the most, and 
      Columns (4) and (5) present the SPJ and FE estimates at this horizon. The period is half-year in RR and yearly in the other three studies.
      The relative difference, defined as $| \widehat{\beta}^{(h)\mathrm{spj}} / \widehat{\beta}^{(h)\mathrm{fe}}-1| \times 100\%$, is shown in Column (6).
\end{tablenotes}
\end{table}

This simple policy calibration applies to the other three cases. To save space, we next focus on the currency crisis, 
which witnesses the largest discrepancy between SPJ and FE estimates.  Based on the results in Table \ref{tab.re_diff}, following a typical currency crisis, FE would call for a fiscal stimulus of approximately 4.51\% of GDP, in comparison to that of 7.89\% of GDP according to SPJ. 
There is a 3.38\% gap to eliminate the potential loss of production at the peak. Developing countries experiencing a currency crisis are usually required to consolidate their fiscal expenditure. Our analysis implies that those fiscal consolidation measures may make their economies more difficult to recover from the currency crisis.  

As a summary, the Nickell bias of FE estimates for financial crises is not merely a statistical concern; it also has significant implications for policymaking and real-world economic outcomes. Ignoring this bias may lead to suboptimal policy decisions that prolong economic hardship and delay recovery from severe crises.

\section{Conclusion}

The Nickell bias for FE is well-known in linear dynamic panel data models where lagged dependent variables serve as regressors. Yet, it has not been studied in empirical macroeconomic applications invoking the panel LP.\footnote{
In panel LP, 
\cite{chong2012harrod}, 
\cite{choi2018oil} and \cite{hobijn2021using} are aware of the explicit Nickell bias and they suggested or used lagged dependent variables as IVs.}
We show that an asymptotic bias emerges in the FE estimator for panel LP even if no lagged dependent variables are present, which is the case in many empirical applications, for example, \citet{chodorow2019macro}, \citet{ottonello2020financial}, and \citet{bahaj2022employment}, to name a few. Furthermore, in the general specification, we illustrate that all predetermined regressors incur Nickell bias. It is imperative to correct this bias for valid asymptotic inference.

The SPJ method is capable of correcting the Nickell bias for general panel LP models. After bias correction, the standard statistical inference based on the $t$-statistic remains valid. By applying our bias correction procedure to four empirical studies on assessing the economic losses associated with financial crises, we show that FE tends to underestimate the economic contraction following financial crises in the mid- and long- runs. Our method extends beyond the literature on financial crises, and can be easily applied in other studies using panel LP.

\bigskip
\bigskip

\singlespacing \bibliographystyle{apalikeyear}
\small
\bibliography{FELP.bib}
\clearpage 
\begin{center}
{\Large Online Appendix for ``Nickell Bias in Panel 
Local Projection: 
 Financial Crises Are Worse Than You Think''} \\ \vspace{1.2em}
{\large Ziwei Mei$^a$, Liugang Sheng$^b$, Zhentao Shi$^b$ \vspace{0.6em}
\\   
 $^a$University of Macau \\  
$^b$The Chinese University of Hong Kong}
\vspace{0.6em}
\end{center}
 \vspace{0.5em}

\setcounter{footnote}{0} \setcounter{table}{0} \setcounter{figure}{0}
\setcounter{equation}{0} 
\setcounter{prop}{0} 
\setcounter{rem}{0} 
\global\long\def\thefootnote{\thesection\arabic{footnote}}%
 
\global\long\def\theequation{\thesection\arabic{equation}}%
 
\global\long\def\thefigure{\thesection\arabic{figure}}%
 
\global\long\def\thetable{\thesection\arabic{table}}%
\global\long\def\theprop{\thesection\arabic{prop}}%
\global\long\def\therem{\thesection\arabic{rem}}%
\appendix
\onehalfspacing 

\normalsize
\bigskip
This appendix consists of three sections. Section \ref{sec:proof}
provides derivations and proofs of the theoretical statements in the main text. 
Section \ref{sec:extensions_theory} enriches the estimation and inference
procedures to allow alternative specifications of the regressions
and dependence structures of the error terms. Section \ref{sec:Simulation-app}
gives additional simulation results under a variety of model specifications with a comprehensive comparison to other popular methods. 

\section{Technical Appendix}\label{sec:proof}\par 

\subsection{Derivation of Eq.(\ref{eq:pred_h})}\label{subsec:pred_h_deriv}

Given the structural equation (\ref{eq:stru_y}), elementary linear
algebra leads to the predictive regression for $y_{i,t+h}$. Notice
that the identification of the panel VAR system demands that $\mathbf{A}_{0,x}$,
the $K\times K$ right-bottom block of $\mathbf{A}_{0}$, must be
of full rank; otherwise the VAR system is not invertible. The block upper triangular feature of $\mathbf{A}_{0}$
allows us to write down its inverse $\mathbf{A}_{0}^{-1}=\begin{pmatrix}1 & \mathbf{a}_{0,yx}^{\prime}\mathbf{A}_{0,x}^{-1}\\
\boldsymbol{0} & \mathbf{A}_{0,x}^{-1}
\end{pmatrix}$ and therefore the reduced-form VAR is
\[
\mathbf{w}_{i,t+1}=\mathbf{A}_{0}^{-1}\left(\boldsymbol{\mu}_{i}^{(0)}+\sum_{s=1}^{p}\mathbf{A}_{s}\mathbf{w}_{i,t+1-s}+\mathbf{u}_{i,t+1}\right)=\boldsymbol{\mu}_{i}^{(1)}+\sum_{s=1}^{p}\mathbf{B}_{s}^{(1)}\mathbf{w}_{i,t+1-s}+\mathbf{e}_{i,t+1}^{(1)},
\]
where $\boldsymbol{\mu}_{i}^{(1)}=\mathbf{A}_{0}^{-1}\boldsymbol{\mu}_{i}^{(0)}$,
$\mathbf{B}_{s}^{(1)}=\mathbf{A}_{0}^{-1}\mathbf{A}_{s}$, and $\mathbf{e}_{i,t+1}^{(1)}=\mathbf{A}_{0}^{-1}\mathbf{u}_{i,t+1}$.
The above expression is the predictive regression for $h=1$. Assumption
\ref{assu:indep_general} (b) implies that all the roots of the determinant
equation 
\[
\mathrm{det}\left(\mathbf{I}_{K+1}-\sum_{s=1}^{p}\mathbf{B}_{s}^{(1)}z^{s}\right)=0
\]
 stay outside of the unit circle on the complex plane to rule out
unit roots and explosive roots. For $h\geq2$, we have 
\[
\mathbf{w}_{i,t+h}=\boldsymbol{\mu}_{i}^{(h)}+\sum_{s=1}^{p}\mathbf{B}_{s}^{(h)}\mathbf{w}_{i,t+1-s}+\mathbf{e}_{i,t+h}^{(h)}
\]
where the slope coefficients are defined recursively as 
\[
\mathbf{B}_{s}^{(h)}=\mathbf{B}_{1}^{(h-1)}\mathbf{B}_{s}^{(1)}+\mathbf{B}_{s+1}^{(h-1)}\cdot\boldsymbol{1}\left\{ s\leq p-1\right\} ,
\]
and the intercept and error term can be written in closed-forms as
\begin{equation}
\boldsymbol{\mu}_{i}^{(h)}=\sum_{s=0}^{h-1}\mathbf{B}_{1}^{(s)}\boldsymbol{\mu}_{i}^{(1)},\ \ \ \mathbf{e}_{i,t+h}^{(h)}=\sum_{s=0}^{h-1}\mathbf{B}_{1}^{(s)}\mathbf{e}_{i,t+h-s}^{(1)}\label{eq:error recursive}
\end{equation}
with $\mathbf{B}_{1}^{(0)}=\mathbf{I}.$ For panel LP, we are interested
in the first equation of the $h$-period-ahead predictive regression
\begin{align*}
y_{i,t+h} & =(1,\boldsymbol{0}^{\prime})\mathbf{w}_{i,t+h}=\mu_{i}^{(h)y}+\sum_{s=1}^{p}\boldsymbol{\theta}_{s}^{(h)\prime}\mathbf{w}_{i,t+1-s}+e_{i,t+h}^{(h)}
\end{align*}
 where $\boldsymbol{\theta}_{s}^{(h)\prime}=(1,\boldsymbol{0}^{\prime})\mathbf{B}_{s}^{(h)}$
is the first row of $\mathbf{B}_{s}^{(h)}$, and $e_{i,t+h}^{(h)}=(1,\boldsymbol{0}^{\prime})\mathbf{e}_{i,t+h}^{(h)}=\sum_{s=0}^{h-1}\boldsymbol{\theta}_{s}^{(h)\prime}\mathbf{e}_{i,t+h-s}^{(1)}$
is the first element of the vector $\mathbf{e}_{i,t+h}^{(h)}$. 

\subsection{Violation of Strict Exogeneity in Panel VAR}\label{subsec:VAR_bias}

The VAR($\infty$) considered by \citet{plagborg2021local} incorporates
as a special case of the VAR($p$) in the main text. VAR($\infty$)
is convenient for the discussion of the population model. Let $\mathbf{A}(z)=\sum_{s=0}^{\infty}\mathbf{A}_{s}z^{s}$
be an infinite-order polynomial for $(1+K)\times(1+K)$ matrices $\mathbf{A}_{s}$,
and the diagonal of $\mathbf{A}_{0}$ is standardized as\textbf{ }1.
Following the notations in (\ref{eq:compact_VAR}), we write the VAR($\infty$)
for a representative individual $i$ as 
\[
\mathbf{A}(\mathbb{L})\mathbf{w}_{i,t}=\boldsymbol{\mu}_{i}^{(0)}+\mathbf{u}_{i,t},
\]
where $\mathbb{L}$ is the lag operator. 

Assume all roots of the determinant equation $\mathrm{det}\left(\mathbf{A}(z)\right)=0$
are outside of the unit cycle, and thus there exists $\mathbf{C}(\mathbb{L})=\mathbf{A}(\mathbb{L})^{-1}$,
where $\mathbf{C}(z)=\sum_{s=0}^{\infty}\mathbf{C}_{s}z^{s}$ is another
infinite-order polynomial with $c_{y,0}=1$. It transforms the VAR($\infty$)
system into a corresponding vector moving average (VMA) system 
\[
\mathbf{w}_{i,t}=\mathbf{C}(\mathbb{L})\left(\boldsymbol{\mu}_{i}^{(0)}+\mathbf{u}_{i,t}\right)=\boldsymbol{\delta}_{i}+\sum_{s=0}^{\infty}\mathbf{C}_{s}\mathbf{u}_{i,t-s}
\]
where $\boldsymbol{\delta}_{i}=\mathbf{C}(1)\boldsymbol{\mu}_{i}^{(0)}$
is the unconditional mean vector. Partition every matrix 
\[
\mathbf{A}_{s}=\begin{pmatrix}a_{y,s} & \mathbf{a}_{yx,s}^{\top}\\
\mathbf{a}_{xy,s} & \mathbf{A}_{x,s}
\end{pmatrix},\ \ \mathbf{C}_{s}=\begin{pmatrix}c_{y,s} & \mathbf{c}_{yx,s}^{\top}\\
\mathbf{c}_{xy,s} & \mathbf{C}_{x,s}
\end{pmatrix}
\]
in a compatible manner with $(u_{i,t+1}^{y},\mathbf{u}_{i,t+1}^{x\prime})^{\prime}$.
Parallel to (\ref{eq:block_VAR}), the Wold-causal order implies $\mathbf{a}_{xy,0}=\boldsymbol{0}$
and leads to $\mathbf{c}_{xy,0}=\boldsymbol{0}$ as well. The first
equation of the VMA($\infty$) is
\[
y_{i,t+h}=\delta_{i}^{y}+\sum_{s=0}^{\infty}c_{y,s}u_{i,t+h-s}^{y}+\sum_{s=0}^{\infty}\mathbf{c}_{yx,s}^{\top}\mathbf{u}_{i,t+h-s}^{x}
\]
for $h\geq2$. 

Let $\mathcal{F}^{i,t}=\sigma\left((\mathbf{w}_{i,s})_{s=-\infty}^{t}\right)=\sigma((\mathbf{u}_{i,s})_{s=-\infty}^{t})$
be the information set from the infinite past up to time $t$, where
$\sigma\left(\cdot\right)$ is the sigma-field generated by the corresponding
random variables. Since $y_{i,t+h}$ is a linear combination of the
innovations, the conditional mean model is 
\begin{equation}
\mathbb{E}\left[y_{i,t+h}|\mathcal{F}^{i,t}\right]=\delta_{i}^{y}+\sum_{s=h}^{\infty}c_{y,s}u_{i,t+h-s}^{y}+\sum_{s=h}^{\infty}\mathbf{c}_{yx,s}^{\top}\mathbf{u}_{i,t+h-s}^{x}.\label{eq:full_reg}
\end{equation}
It follows that the error term against the information up to time
$t$ is
\begin{equation}
e_{i,t+h}^{(h)}=y_{i,t+h}-\mathbb{E}\left[y_{i,t+h}|\mathcal{F}^{i,t}\right]=\sum_{s=0}^{h-1}c_{y,s}u_{i,t+h-s}^{y}+\sum_{s=0}^{h-1}\mathbf{c}_{yx,s}^{\top}\mathbf{u}_{i,t+h-s}^{x}.\label{eq:full_error}
\end{equation}

Let $\mathcal{G}$ by the information set of all included regressors,
to be discussed below. Notice that $c_{y,0}=1$ by normalization.
Whether the local projection regression violates strict exogeneity
is equivalent to check if 
\begin{equation}
\mathbb{E}[e_{i,t+h}^{(h)}|\mathcal{G}]=\mathbb{E}\left[u_{i,t+h}^{y}|\mathcal{G}\right]+\sum_{s=1}^{h-1}c_{y,s}\mathbb{E}\left[u_{i,t+h-s}^{y}|\mathcal{G}\right]+\sum_{s=0}^{h-1}\mathbf{c}_{yx,s}^{\top}\mathbb{E}\left[\mathbf{u}_{i,t+h-s}^{x}|\mathcal{G}\right]\label{eq:cond_E}
\end{equation}
equals 0 or not. 
\begin{enumerate}
\item If $c_{y,s}\neq0$ for some $s\ge1$ ($c_{y,0}=1$ by standardization),
the researcher includes lagged dependent variables in the panel LP
regression we have $\sigma\left((y_{i,s})_{s=-\infty}^{T-h}\right)\subseteq\mathcal{G}$.
As a result, the first term in (\ref{eq:cond_E})
is $\mathbb{E}\left[u_{i,t+h}^{y}|\mathcal{G}\right]=u_{i,t+h}^{y}\neq0$ whenever $t\leq T-2h$,
and strict exogeneity must be violated. This is the explicit Nickell
bias due to lagged dependent variables.
\item If the researcher includes $\mathbf{x}_{i,t}$ in the regression and
thus $\sigma\left((\mathbf{u}_{i,s}^{x})_{s=-\infty}^{T-h}\right)\subseteq\mathcal{G}$.
There are two cases:
\begin{enumerate}
\item If in the true DGP $\mathbf{c}_{yx,s}\neq0$ for some $s\geq0$, then
the third term in the right-most expression of (\ref{eq:cond_E}) whenever $t\leq T-2h$
becomes 
\[
\sum_{s=0}^{h-1}\mathbf{c}_{yx,s}^{\top}\mathbb{E}\left[\mathbf{u}_{i,t+h-s}^{x}|\mathcal{G}\right]=\sum_{s=0}^{h-1}\mathbf{c}_{yx,s}^{\top}\mathbf{u}_{i,t+h-s}^{x}\neq0.
\]
\item If $\mathbf{c}_{xy,s}\neq\boldsymbol{0}$ for some $s\geq1$ ($\mathbf{c}_{xy,0}=\boldsymbol{0}$
is the implication of Wold-causality), then for some $\tau<t$ the
lagged $y_{i,\tau}$ Granger-causes $x_{i,t}$. In this case, again
in (\ref{eq:cond_E}) the first term $\mathbb{E}\left[u_{i,t+h}^{y}|\mathcal{G}\right]\neq0$ whenever $t\leq T-2h$
violates strict exogeneity. 
\end{enumerate}
\end{enumerate}
In summary, Point 1 translates to the familiar explicit
Nickell bias. Point 2(a) has $\mathbf{x}_{i,t}$ as regressors and
it incurs Nickell bias as well; the prototype model in Section \ref{subsec: implicit bias}
is a special case with $c_{y,s}=0$ but $\mathbf{c}_{yx,0}\neq\boldsymbol{0}$,
leading to Nickell bias despite the absence of the lagged
dependent variables. The information of $y_{i,\tau}$ is leaked via
$\mathbf{c}_{xy,s}\neq0$ into the future \textbf{$\mathbf{x}_{i,t}$}
in Point 2(b); the ``static model'' in Example 2 of \citet{chudik2018half}'s
Online Appendix is a special case of it.

When $\{c_{y,s}=0\}_{s=1}^{H-1}$, $\{\mathbf{c}_{xy,s}=\boldsymbol{0}\}_{s=1}^{H-1}$,
and $\{\mathbf{c}_{yx,s}=\boldsymbol{0}\}_{s=0}^{H-1}$, the true
DGP becomes a void model $y_{i,t+h}=\delta_{i}^{y}+u_{i,t+h}^{y}$
and it is completely unpredictable by the past information. But if
this happens, a researcher shall exclude all regressors from the linear
regression in the first place. Given that she has no prior knowledge
about the true regression coefficients, a researcher should worry that
each potentially relevant regressor based on her economic reasoning
is subject to Nickell
bias in panel LP. As a consequence, it is imperative to conduct
bias correction for asymptotically valid inference. 

\subsection{Proof of Proposition \ref{thm:biasFEAR1}}
\label{subsec: proof FEAR1}
We begin with the following decomposition 
\begin{align*}
\sqrt{NT_{h}}\left(\hat{\beta}^{(h){\rm \mathrm{fe}}}-{\beta}^{(h)}\right) & =\dfrac{1}{\sqrt{NT_{h}}s_{x}^{2}}\sum_{i\in[N]}\sum_{t\in\mathcal{T}^{h}}\tilde{x}_{i,t}e_{i,t+h}^{(h)}\\
 & =\frac{1}{\sqrt{N}s_{x}^{2}}\sum_{i\in[N]}\left(\dfrac{1}{\sqrt{T_{h}}}\sum_{t\in\mathcal{T}^{h}}\tilde{x}_{i,t}u_{i,t+h}^{y}+\dfrac{1}{\sqrt{T_{h}}}\beta^{(0)}\sum_{t\in\mathcal{T}^{h}}\sum_{s=0}^{h-1}\rho^{s}\tilde{x}_{i,t}u_{i,t+h-s}^{x}\right).
\end{align*}
The independence between $u_{i,t}^{y}$ and $u_{i,t}^{x}$ immediately
implies the first term $\mathbb{E}\left[T_{h}^{-1/2}\sum_{t\in\mathcal{T}^{h}}\tilde{x}_{i,t}u_{i,t+h}^{y}\right]=0$,
and next we focus on the second term $T_{h}^{-1/2}\beta^{(0)}\sum_{t\in\mathcal{T}^{h}}\sum_{s=0}^{h-1}\rho^{s}\tilde{x}_{i,t}u_{i,t+h-s}^{x}$.
Notice that $\mathbb{E}[x_{i,t}u_{i,t+h-s}^{x}]=0$ for any $0\leq s\leq h-1$
and hence 
\begin{equation}
\sum_{t\in\mathcal{T}^{h}}\mathbb{E}\left[x_{i,t}\sum_{s=0}^{h-1}\rho^{s}u_{i,t+h-s}^{x}\right]=0.\label{eq:cov1}
\end{equation}
By the AR(1) model (\ref{eq:AR1}), we have $x_{i,t}-\mu_i^x\sum_{\ell=0}^{t}\rho^{\ell}=\sum_{\ell=0}^{t}\rho^{\ell}u_{i,t-\ell}^{x}$,
which yields 
\[
\bar{x}_{i}=\dfrac{1}{T_{h}}\sum_{t\in\mathcal{T}^{h}}x_{i,t}=\dfrac{1}{T_{h}}\sum_{r=1}^{T_{h}}\sum_{\ell=0}^{r}\rho^{\ell}u_{i,r-\ell}^{x}+\mu_i^x\dfrac{1}{T_{h}}\sum_{r=1}^{T_{h}}\sum_{\ell=0}^{r}\rho^{\ell}.
\]
For any $s\in[h]$, the independence of $u_{i,t}^{x}$ across $i$
and $t$ gives 
\begin{eqnarray*}
 &  & \sum_{t=0}^{T_{h}}\mathbb{E}\left[\bar{x}_{i}\rho^{s}u_{i,t+h-s}^{x}\right]\\
 & = & \dfrac{\rho^{s}}{T_{h}}\sum_{t=0}^{T_{h}}\sum_{r=1}^{T_{h}}\sum_{\ell=0}^{r}\mathbb{E}\left[u_{i,t+h-s}^{x}\cdot\rho^{\ell}u_{i,r-\ell}^{x}\right]=\dfrac{\rho^{s}}{T_{h}}\sum_{t=0}^{T_{h}}\sum_{r=1}^{T_{h}}\sum_{\ell=0}^{r}\rho^{r-\ell}\mathbb{E}\left[u_{i,t+h-s}^{x}u_{i,\ell}^{x}\right]\\
 & = & \dfrac{\rho^{s}}{T_{h}}\sum_{t=h-s}^{T_{h}}\sum_{r=1}^{T_{h}}\sum_{\ell=0}^{r}\rho^{r-\ell}\mathbb{E}\left[u_{i,t}^{x}u_{i,\ell}^{x}\right]=\sigma_{u_{x}}^{2}\dfrac{\rho^{s}}{T_{h}}\sum_{t=h-s}^{T_{h}}\left(T_{h}-t\right)\rho^{t-\left(h-s\right)}
\end{eqnarray*}
where the last equality follows as $\mathbb{E}\left[u_{i,t}^{x}u_{i,\ell}^{x}\right]=\sigma_{u_{x}}^{2}\boldsymbol{1}\left\{ t=\ell\right\} $.
Define 
\begin{align*}
f_{T,h}(\rho) & :=\sum_{s=0}^{h-1}\sum_{t\in\mathcal{T}^{h}}\mathbb{E}\left[\bar{x}_{i}\rho^{s}u_{i,t+h-s}^{x}\right]=\sum_{s=0}^{h-1}\sum_{t=h-s}^{T_{h}}\left(1-\frac{t}{T_{h}}\right)\rho^{t-h+2s}\\
 & =\frac{\left(T-2h\right)\left(1-\rho^{2}\right)-\left(T-h\right)\rho^{h}\left(1-\rho^{2}\right)+\rho^{T-2h+1}\left(1-\rho^{2h}\right)}{\left(T-h\right)\left(1-\rho\right)^{2}\left(1-\rho^{2}\right)}\\
 & =\frac{\left(1-\rho^{h}\right)}{\left(1-\rho\right)^{2}}-\dfrac{h}{\left(T-h\right)\left(1-\rho\right)^{2}}+\frac{\rho^{T-2h+1}\left(1-\rho^{2h}\right)}{\left(T-h\right)\left(1-\rho\right)^{2}(1-\rho^2)}
\end{align*}
after long but elementary calculations for the third equality. We
deduce 
\[
\dfrac{1}{\sqrt{T_{h}}}\sum_{s=0}^{h-1}\sum_{t\in\mathcal{T}^{h}}\rho^{s}\mathbb{E}\left[\tilde{x}_{i,t}u_{i,t+h-s}^{x}\right]=-\dfrac{1}{\sqrt{T_{h}}}f_{T,h}(\rho).
\]
As a result, 
\begin{equation}
\mathbb{E}\left[\frac{1}{\sqrt{T_{h}}}\sum_{t\in\mathcal{T}^{h}}\tilde{x}_{i,t}e_{i,t+h}^{(h)}+\dfrac{\beta^{(0)}}{\sqrt{T_{h}}}f_{T,h}(\rho)\right]=0\label{eq:Exe}
\end{equation}
for all $i\in[N]$. The asymptotic normality assumed in (\ref{eq:clt_xe_bias})
gives 
\begin{eqnarray*}
 &  & \dfrac{1}{\sigma_{xe,h}\sqrt{NT_{h}}}\sum_{i\in[N]}\sum_{t\in\mathcal{T}^{h}}\left(\tilde{x}_{i,t}e_{i,t+h}^{(h)}+\dfrac{\beta^{(0)}}{\sqrt{T_{h}}}f_{T,h}(\rho)\right)\\
 & = & \dfrac{1}{\sigma_{xe,h}\sqrt{NT_{h}}}\sum_{i\in[N]}\sum_{t\in\mathcal{T}^{h}}\left(\tilde{x}_{i,t}e_{i,t+h}^{(h)}-\mathbb{E}\left(\tilde{x}_{i,t}e_{i,t+h}^{(h)}\right)\right)\stackrel{d}{\to}\mathcal{N}(0,1).
\end{eqnarray*}
Furthermore, the law of large numbers assumed in (\ref{eq:plim_varx})
yields 
\begin{align*}
\sqrt{NT_{h}}\left(\hat{\beta}^{(h){\rm \mathrm{fe}}}-\beta^{(h)}\right)+\frac{\beta^{(0)}\sigma_{u_{x}}^{2}}{s_{x}^{2}}\sqrt{\dfrac{N}{T_{h}}}f_{T,h}(\rho) & \stackrel{d}{\to}\mathcal{N}\left(0,\ \frac{\sigma_{xe,h}^{2}}{\sigma_{x}^{4}}\right)
\end{align*}
as stated in the theorem.

\subsection{Oracle Debiased Estimator \label{subsec:consist}}\par Proposition
\ref{thm:biasFEAR1} shows that under (\ref{eq:DGP}) and (\ref{eq:AR1}),
the Nickell bias in the asymptotic normal  distribution of $\hat{\beta}^{(h){\rm \mathrm{fe}}}$ is $-\frac{\beta^{(0)}\sigma_{u_{x}}^{2}}{s_{x}^{2}}\sqrt{\dfrac{N}{T_{h}}}f_{T,h}(\rho)$.
We thus define the following debiased FE estimator 
\[
\hat{\beta}^{(h){\rm db}}=\hat{\beta}^{(h){\rm \mathrm{fe}}}+\dfrac{\hat{\beta}^{(0)}}{T_{h}s_{x}^{2}}\cdot\hat{\sigma}_{u_{x}}^{2}\cdot f_{T,h}(\hat{\rho})
\]
where $\hat{\beta}^{(0)}$ is the FE estimator for (\ref{eq:DGP}),
$\hat{\rho}$ is the FE estimator of the AR(1) model (\ref{eq:AR1}).
For the variances, $\hat{\sigma}_{u_{x}}^{2}$ is the sum of squared
residuals of the AR(1) regression, and $s_{x}^{2}$ is defined 
in (\ref{eq:plim_varx}).\par We need to further estimate $\sigma_{xe,h}^{2}$
for inference. A consistent estimator is given by 
\begin{align}
\widehat{\sigma}_{xe,h}^{2} & :=\dfrac{1}{NT_{h}}\sum_{i\in[N]}\left[\sum_{t\in\mathcal{T}^{h}}\tilde{x}_{i,t}\widehat{e}_{i,t+h}^{(h){\rm \mathrm{db}}}\right]^{2}\label{eq:sighatxeh}
\end{align}
where $\widehat{e}_{i,t+h}^{(h)\mathrm{db}}=\tilde{y}_{i,t+h}-\widehat{\beta}^{(h)\mathrm{db}}\tilde{x}_{i,t}$$.$
Under standard regularity conditions we can show 
\[
\widehat{\sigma}_{xe,h}^{2}-\dfrac{1}{NT_{h}}\sum_{i\in[N]}\mathbb{E}\left[\sum_{t\in\mathcal{T}^{h}}\text{\ensuremath{\tilde{x}_{i,t}}}e_{i,t+h}^{(h)}\right]^{2}=o_{p}(1).
\]
Besides, we have deduced by (\ref{eq:Exe}) that 
\[
\dfrac{1}{NT_{h}}\sum_{i\in[N]}\left[\mathbb{E}\sum_{t\in\mathcal{T}^{h}}\tilde{x}_{i,t}e_{i,t+h}^{(h)}\right]^{2}=\dfrac{(\beta^{(0)})^{2}\sigma_{u_{x}}^{4}f_{T,h}^{2}(\rho)}{T_{h}}=o_{p}(1)
\]
as $T_{h}\to\infty.$ As a result $\widehat{\sigma}_{xe,h}^{2}-\sigma_{xe,h}^{2}=o_{p}(1)$
by the fact that 
\begin{align*}
\sigma_{xe,h}^{2} & =\lim_{(N,T)\to\infty}\dfrac{1}{NT_{h}}\sum_{i\in[N]}{\rm var}\left[\sum_{t\in\mathcal{T}^{h}}\tilde{x}_{i,t}e_{i,t+h}^{(h)}\right]\\
 & =\lim_{(N,T)\to\infty}\dfrac{1}{NT_{h}}\sum_{i\in[N]}\left\{ \mathbb{E}\left[\sum_{t\in\mathcal{T}^{h}}\tilde{x}_{i,t}e_{i,t+h}^{(h)}\right]^{2}-\left[\mathbb{E}\sum_{t\in\mathcal{T}^{h}}\tilde{x}_{i,t}e_{i,t+h}^{(h)}\right]^{2}\right\} .
\end{align*}
Familiar two-sided, symmetric (around the point estimate) confidence
intervals can be constructed given $\widehat{s}^{(h)\mathrm{db}}=\dfrac{\widehat{\sigma}_{xe,h}}{s_{x}\sqrt{NT_{h}}}.$

\subsection{Proof of Theorem \ref{thm:hj}} We first present a useful
lemma.
\begin{lem}
\label{lem:corr_b} If Assumption \ref{assu:indep_general} holds,
then 
\[
\sup_{i\in[N]}\text{\ensuremath{\Bigg\Vert}}\sum_{t\in\mathcal{T}_{a}^{h}}\mathbb{E}\left(\bar{\mathbf{W}}_{i,b}e_{i,t+h}^{(h)}\right)\ensuremath{\Bigg\Vert}=O\left(\frac{h}{T_{h}}\right)
\]
where $e_{i,t+h}^{(h)}=\sum_{s=0}^{h-1}\boldsymbol{\theta}_{s}^{(h)\prime}\mathbf{e}_{i,t+h-s}^{(1)}$
is the first entry of $\mathbf{e}^{(h)}_{i,t+h}$ defined in (\ref{eq:error recursive}).
\end{lem}
\begin{proof}[Proof of Lemma \ref{lem:corr_b}] Recall that $\mathbf{W}_{i,t}=(\mathbf{w}_{i,t}^{\prime},\mathbf{w}_{i,t-1}^{\prime},\ldots,\mathbf{w}_{i,t-p+1}^{\prime})^{\prime}$.
It suffices to show that $\sup_{i\in[N]}\Bigg\Vert\mathbb{E}\left[\sum_{t\in\mathcal{T}_{a}^{h}}\bar{\mathbf{w}}_{i,b}^{j}e_{i,t+h}^{(h)}\right]\Bigg\Vert=O\left(h/T_{h}\right)$
where $\bar{\mathbf{w}}_{i,b}^{j}=(T_{h}/2)^{-1}\sum_{t\in\mathcal{T}_{b}^{h}}\mathbf{w}_{i,t-j}$
for $j=0,1,\cdots,p-1.$ The VAR($p$) process $\mathbf{w}_{i,t}$
has the representation 
$\mathbf{w}_{i,t}=\boldsymbol{\delta}_{i}+\sum_{d=0}^{\infty}\boldsymbol{\Psi}_{d}\mathbf{e}_{i,t-d}^{(1)}$
with the coefficient matrices $\|\boldsymbol{\Psi}_{d}\|\leq C{\rm e}^{-cd}$ for some absolute
constants $C,c>0$ and the unconditional mean $\boldsymbol{\delta}_{i}$.
Besides, we
have 
\begin{eqnarray*}
\mathbb{E}\left[\sum_{t\in\mathcal{T}_{a}^{h}}\bar{\mathbf{w}}_{i,b}^{j}e_{i,t+h}^{(h)}\right] & = & \dfrac{1}{T_{h}/2}\sum_{t\in\mathcal{T}_{a}^{h}}\sum_{r\in\mathcal{T}_{b}^{h}}\mathbb{E}\left[\mathbf{w}_{i,r-j}^{j}e_{i,t+h}^{(h)}\right]\\
 & = & \dfrac{1}{T_{h}/2}\sum_{r\in\mathcal{T}_{b}^{h}}\sum_{t\in\mathcal{T}_{a}^{h}}\sum_{s=0}^{h-1}\sum_{d=0}^{\infty}\mathbb{E}\left[\boldsymbol{\Psi}_{d}\mathbf{e}_{i,r-j-d}^{(1)}\cdot\boldsymbol{\theta}_{s}^{(h)\prime}\mathbf{e}_{i,t+h-s}^{(1)}\right]
\end{eqnarray*}
where 
\begin{equation}\label{eq: expec cases}
   \mathbb{E}\left[\boldsymbol{\Psi}_{d}\mathbf{e}_{i,r-j-d}^{(1)}\cdot\boldsymbol{\theta}_{s}^{(h)\prime}\mathbf{e}_{i,t+h-s}^{(1)}\right] = \begin{cases}
   \boldsymbol{\Psi}_{d}\mathbf{A}_0^{-1}\boldsymbol{\Sigma}_{u}(\mathbf{A}_0^{-1})^\prime \boldsymbol{\theta}_{s}^{(h)}, &d = r - t - j - h +s\\
   \boldsymbol{0}, &\rm{otherwise}
    \end{cases}
\end{equation} 
where $\boldsymbol{\Sigma}_{u}=\mathbb{E}\left(\mathbf{u}_{i,t}\mathbf{u}_{i,t}^{\prime}\right)$. By (\ref{eq: expec cases}) and the fact that $j\in[p],0\leq s\leq h-1$, we have $$\Bigg\Vert \sum_{d=0}^{\infty}\mathbb{E}\left[\boldsymbol{\Psi}_{d}\mathbf{e}_{i,r-j-d}^{(1)}\cdot\boldsymbol{\theta}_{s}^{(h)\prime}\mathbf{e}_{i,t+h-s}^{(1)}\right]\Bigg\Vert \leq C_1\cdot C{\rm e}^{-c(r-t)}\cdot {\rm e}^{c(p+h)}$$ where $C_1 = \Vert\mathbf{A}_{0}^{-1}\Vert^2 \cdot \Vert\boldsymbol{\Sigma}_{u}\Vert \cdot \sup_{0\leq s\leq h-1}\Vert\boldsymbol{\theta}_{s}^{(h)}\Vert$. Hence,
\begin{align*}
\Bigg\Vert\mathbb{E}\left[\sum_{t\in\mathcal{T}_{a}^{h}}\bar{\mathbf{w}}_{i,b}^{j}e_{i,t+h}^{(h)}\right]\Bigg\Vert & \leq \dfrac{C_1\cdot C{\rm e}^{c(p+h)}}{T_{h}/2}\sum_{r\in\mathcal{T}_{b}^{h}} \sum_{t\in\mathcal{T}_{a}^{h}} \sum_{s=0}^{h-1} {\rm e}^{-c(r-t)}\\ 
 & = \dfrac{C_2\cdot h}{T_{h}}\sum_{r=T_h/2+1}^{T_h} {\rm e}^{-cr} \sum_{t=1}^{T_h/2} {\rm e}^{ct} \\
 & = \dfrac{C_2\cdot h}{T_{h}}\cdot \dfrac{{\rm e}^{-c(T_h/2+1)}(1-{\rm e}^{-cT_h/2})}{1-{\rm e}^{-c}} \cdot \dfrac{{\rm e}^{cT_h/2}-1}{1-{\rm e}^{-c}} \\
 & \leq \dfrac{C_3\cdot h}{T_{h}} {\rm e}^{-c(T_h/2 + 1)} \cdot {\rm e}^{cT_h/2} \leq \dfrac{C_3\cdot h}{T_{h}}
\end{align*}
where $C_2=2C_1\cdot C{\rm e}^{c(p+h)}$ and $C_3 = \frac{C_2}{(1-{\rm e}^{-c})^2}$. This non-asymptotic bound holds uniformly for all $i\in[N]$.
\end{proof}\par \bigskip{}
\par We proceed with the main proof. The FE estimator is 
\[
\hat{\boldsymbol{\theta}}^{(h)\mathrm{fe}}=\widehat{\mathbf{Q}}^{-1}\dfrac{1}{NT_{h}}\sum_{i\in[N]}\sum_{t\in\mathcal{T}^{h}}\tilde{\mathbf{W}}_{i,t}y_{i,t+h}
\]
 and thus 
\[
\hat{\boldsymbol{\theta}}^{(h)\mathrm{fe}}-\boldsymbol{\theta}^{(h)}=\widehat{\mathbf{Q}}^{-1}\boldsymbol{\zeta}
\]
where $\boldsymbol{\zeta}=\dfrac{1}{NT_{h}}\sum_{i\in[N]}\sum_{t\in\mathcal{T}^{h}}\tilde{\mathbf{W}}_{i,t}e_{i,t+h}^{(h)}$.
The estimation error can be decomposed as 
\begin{equation}
\hat{\boldsymbol{\theta}}^{(h)\mathrm{spj}}-\boldsymbol{\theta}^{(h)}=\widehat{\mathbf{Q}}^{-1}\left(2\boldsymbol{\zeta}-\dfrac{1}{2}\left(\boldsymbol{\zeta}_{a}+\boldsymbol{\zeta}_{b}\right)\right)+\left(\widehat{\mathbf{Q}}^{-1}-\widehat{\mathbf{Q}}_{a}^{-1}\right)\dfrac{\boldsymbol{\zeta}_{a}}{2}+\left(\widehat{\mathbf{Q}}^{-1}-\widehat{\mathbf{Q}}_{b}^{-1}\right)\dfrac{\boldsymbol{\zeta}_{b}}{2}\label{eq:theta_error}
\end{equation}
where $\boldsymbol{\zeta}_{k}=\dfrac{2}{NT_{h}}\sum_{i\in[N]}\sum_{t\in\mathcal{T}_{k}^{h}}\left(\mathbf{w}_{i,t}-\bar{\mathbf{w}}_{i,k}\right)e_{i,t+h}^{(h)}$
for $k\in\left\{ a,b\right\} $.\par Multiply $\sqrt{NT_{h}}$ on
both side of (\ref{eq:theta_error}). We focus on 
\[
\sqrt{NT_{h}}\left(2\boldsymbol{\zeta}-\dfrac{1}{2}\left(\boldsymbol{\zeta}_{a}+\boldsymbol{\zeta}_{b}\right)\right)=\dfrac{1}{\sqrt{NT_{h}}}\sum_{i\in[N]}\sum_{t\in\mathcal{T}^{h}}\mathbf{d}_{i,t}^{*}e_{i,t+h}^{(h)}.
\]
The summand $\left\{ T_{h}^{-1/2}\sum_{t\in\mathcal{T}^{h}}\left[\mathbf{d}_{i,t}^{*}e_{i,t+h}-\mathbb{E}\left(\mathbf{d}_{i,t}^{*}e_{i,t+h}\right)\right]\right\} _{i=1}^{N}$
is independent across $i\in[N]$ under Assumption \ref{assu:indep_general}.
The definition of $\mathbf{d}_{i,t}^{*}$ yields the following decomposition
\[
\frac{1}{\sqrt{NT_{h}}}\sum_{i\in[N]}\sum_{t\in\mathcal{T}^{h}}\mathbf{d}_{i,t}^{*}e_{i,t+h}^{(h)}=\frac{1}{\sqrt{NT_{h}}}\sum_{i\in[N]}\sum_{t\in\mathcal{T}_{a}^{h}}(\mathbf{W}_{i,t}-\bar{\mathbf{W}}_{i,b})e_{i,t+h}^{(h)}+\frac{1}{\sqrt{NT_{h}}}\sum_{i\in[N]}\sum_{t\in\mathcal{T}_{b}^{h}}(\mathbf{W}_{i,t}-\bar{\mathbf{W}}_{i,a})e_{i,t+h}^{(h)}.
\]
Recall that $e_{i,t+h}^{(h)}$ is the first entry of $\mathbf{e}_{i,t+h}^{(h)}=\sum_{s=0}^{h-1}\mathbf{B}_{1}^{(s)}\mathbf{e}_{i,t+h-s}^{(1)}$
defined in (\ref{eq:error recursive}), a linear combination of $\left(\mathbf{e}_{i,t+\ell}^{(1)}\right)_{\ell=1}^{h}$.
By construction, $\mathbb{E}\left[\mathbf{W}_{i,t}e_{i,t+h}^{(h)}\right]=0$
for all $t\geq1,$ and $\mathbb{E}\left[\bar{\mathbf{W}}_{i,a}e_{i,t+h}^{(h)}\right]=0$
for all $t\in\mathcal{T}_{b}^{h}$, since $\mathbb{E}\left[\mathbf{W}_{i,s}e_{i,t+h}^{(h)}\right]=0$
whenever $s<t.$ Lemma \ref{lem:corr_b} then gives 
\begin{equation}
\frac{1}{\sqrt{NT_{h}}}\sum_{i\in[N]}\sum_{t\in\mathcal{T}^{h}}\mathbb{E}\left[\mathbf{d}_{i,t}^{*}e_{i,t+h}^{(h)}\right]=O\left(\frac{1}{\sqrt{NT_{h}}}\cdot N\dfrac{h}{T_{h}}\right)=O\left(h\frac{N^{1/2}}{T_{h}^{3/2}}\right)\to0\label{eq:HJ-bias}
\end{equation}
under the condition $N/T^{3}\to0$ with a fixed $h$.\par Parallel
calculation delivers $\sqrt{NT_{h}}\boldsymbol{\zeta}_{a}=O_{p}(1)$
and $\sqrt{NT_{h}}\boldsymbol{\zeta}_{b}=O_{p}(1).$ We therefore
conclude 
\begin{eqnarray*}
 &  & \sqrt{NT_{h}}\left(\hat{\boldsymbol{\theta}}^{(h){\rm spj}}-\boldsymbol{\theta}^{(h)}\right)\\
 & = & \widehat{\mathbf{Q}}^{-1}\left(2\boldsymbol{\zeta}-\dfrac{\boldsymbol{\zeta}_{a}}{2}-\dfrac{\boldsymbol{\zeta}_{b}}{2}\right)+\left(\widehat{\mathbf{Q}}^{-1}-\widehat{\mathbf{Q}}_{a}^{-1}\right)\cdot O_{p}\left(1\right)+\left(\widehat{\mathbf{Q}}^{-1}-\widehat{\mathbf{Q}}_{b}^{-1}\right)\cdot O_{p}\left(1\right)\\
 & = & \widehat{\mathbf{Q}}^{-1}\frac{1}{\sqrt{NT_{h}}}\sum_{i\in[N]}\sum_{t\in\mathcal{T}^{h}}\left[\mathbf{d}_{i,t}^{*}e_{i,t+h}^{(h)}-\mathbb{E}\left[\mathbf{d}_{i,t}^{*}e_{i,t+h}^{(h)}\right]\right]+\frac{1}{\sqrt{NT_{h}}}\sum_{i\in[N]}\sum_{t\in\mathcal{T}^{h}}\mathbb{E}\left[\mathbf{d}_{i,t}^{*}e_{i,t+h}^{(h)}\right]+o_{p}\left(1\right)\\
 & = & \widehat{\mathbf{Q}}^{-1}\frac{1}{\sqrt{NT_{h}}}\sum_{i\in[N]}\sum_{t\in\mathcal{T}^{h}}\left[\mathbf{d}_{i,t}^{*}e_{i,t+h}^{(h)}-\mathbb{E}\left[\mathbf{d}_{i,t}^{*}e_{i,t+h}^{(h)}\right]\right]+o(1)+o_{p}\left(1\right)\\
 & \stackrel{d}{\to} & \mathcal{N}\left(\boldsymbol{0},\boldsymbol{{\rm Q}}^{-1}\boldsymbol{{\rm R}}\boldsymbol{{\rm Q}}^{-1}\right)
\end{eqnarray*}
where the second equality follows by Assumption \ref{assu:limit_HJ}
as $\mathbf{Q}=\mathrm{plim}\widehat{\mathbf{Q}}_{a}=\mathrm{plim}\widehat{\mathbf{Q}}_{b}$
and so does $\mathbf{Q}=\mathrm{plim}\widehat{\mathbf{Q}}$, the third
equality from (\ref{eq:HJ-bias}), and the limiting distribution
follows by (\ref{eq:assu_CLT_hj}).

\section{Extensions for Alternative Specifications}\label{sec:extensions_theory}\par 

\subsection{Differenced Dependent Variable \label{subsec:bias_specification_FD}}

Instead of using the level of the dependent variable, another popular specification of LP uses the cumulative difference in $y_{i,t}$ over the horizon $h$ as the dependent variable, for example \citet{baron2021banking}. Specifically, the regression is given by 
\begin{equation}
\Delta_{h}y_{i,t+h}=\mu_{i}^{(h)}+\beta^{(h)}x_{i,t}+e{}_{i,t+h}^{(h)},\label{eq:DGP-hD}
\end{equation}
where $\Delta_{h}y_{i,t+h}:=y_{i,t+h}-y_{i,t}$. We slightly abuse the notations $\mu^{(h)}_i$, $\beta^{(h)}$, and $e_{i,t+h}^{(h)}$ to represent the fixed effect, the IRFs, and the error term, respectively. 
The independent variable $x_{i,t}$ accommodates various forms of stationary panels, such as the raw data, first-order difference, and percentage change. Similar to the prototype panel LP model (\ref{eq:modelh}),  this specification is not immune from the Nickell bias. 
\par Suppose that the true DGP now becomes 
\begin{equation}
\Delta y_{i,t}=\mu_{i}^{(0)y}+\beta^{(0)}x_{i,t}+u_{i,t}^{y},\label{eq:DGP-FD}
\end{equation}
and the dynamics of $x_{i,t}$ remains (\ref{eq:AR1}). Similar derivations as in Section \ref{sec:bias} yield 
\begin{equation}
\Delta y_{i,t+r}=\alpha_{i}^{(r)}+b^{(r)}x_{i,t}+\varepsilon_{i,t+r}^{(r)}\label{eq:DGP Delta FD s}
\end{equation}
for all $r\geq1$, where $b^{(r)}=\rho^{r}\beta^{(0)}$, $\alpha_{i}^{(r)}=\mu_{i}^{(0)y}+\beta^{(0)}\mu_{i}^{x}\sum_{s=0}^{r-1}\rho^{s}$,
and $\varepsilon_{i,t+r}^{(r)}=u_{i,t+r}^{y}+\beta^{(0)}\sum_{s=0}^{r-1}\rho^{s}u_{i,t+r-s}^{x}$.
The components in (\ref{eq:DGP-FD}) can be found by taking $r$ from
1 to $h$ and summing up for (\ref{eq:DGP Delta FD s}):
\begin{equation}
    \label{eq: sum up}
    \begin{aligned}
\beta^{(h)} & =\sum_{r=1}^{h} b^{(r)}=\beta^{(0)}\sum_{r=1}^h \rho^r,\\
\mu_{i}^{(h)} & =\sum_{r=1}^{h} \alpha_i^{(r)},\text{ and }
e_{i,t+h}^{(h)}  =\sum_{r=1}^{h}\varepsilon^{(r)}_{i,t+r}.
    \end{aligned}
\end{equation} 
Some elementary calculations yield the following closed-form expressions 
\begin{align*}
\beta^{(h)} & =\dfrac{\beta^{(0)}\rho(1-\rho^{h})}{1-\rho},\\
\mu_{i}^{(h)} & =h\mu_{i}^{(0)y}+\dfrac{\beta^{(0)}\mu_{i}^{x}[h(1-\rho)-\rho(1-\rho^{h})]}{(1-\rho)^{2}},\\
e_{i,t+h}^{(h)} & =\sum_{r=1}^{h}u_{i,t+r}^{y}+\beta^{(0)}\sum_{r=1}^{h}\sum_{s=0}^{r-1}\rho^{s}u_{i,t+r-s}^{x}.
\end{align*}
The FE estimator of panel LP for the model (\ref{eq:DGP-hD}) is given as 
\[\hat\beta^{(h)\text{fe-dif}} =  {\sum_{i\in[N]} \sum_{t\in\mathcal{T}^h} \tilde x_{i,t} \Delta_h y_{i,t+h} } / { \sum_{i\in[N]} \sum_{t\in\mathcal{T}^h} \tilde x_{i,t}^2 }. \]
The Nickell bias of the FE estimator $\hat\beta^{(h)\text{fe-dif}}$ for  (\ref{eq:DGP-hD}) stems from the second term
$\beta^{(0)}\sum_{r=1}^{h}\sum_{s=0}^{r-1}\rho^{s}u_{i,t+r-s}^{x}$
of $e_{i,t+h}^{(h)}$ by parallel arguments as in Section \ref{sec:bias}. The following proposition formally quantifies the Nickell bias. 
\begin{prop} \label{thm:biasFEAR1_dif} 
Suppose $\Delta y_{i,t}$ and $x_{i,t}$ are generated from the DGPs (\ref{eq:DGP-FD}) and (\ref{eq:AR1}). Under the conditions of Proposition \ref{thm:biasFEAR1},  
\begin{equation}
\sqrt{NT_{h}}\left(\hat\beta^{(h){\rm fe}\text{-}{\rm dif}} -\beta^{(h)}\right)+\beta^{(0)}\cdot\dfrac{\sigma_{u_{x}}^{2}}{s_{x}^{2}}\sqrt{\dfrac{N}{T_{h}}}\sum_{r=1}^h f_{T,r}(\rho)\stackrel{d}{\to}   \mathcal{N}\left(0,\ \frac{\sigma_{xe,h}^{2}}{\sigma_{x}^{4}}\right),
\label{eq:biasFEAR1_dif}
\end{equation}
where $\{f_{T,r}(\rho)\}_{ 1 \leq r \leq h}$,
$\sigma_{u_{x}}^{2}$, and 
$\sigma_{x}^{2}$ follow the definitions in Proposition \ref{thm:biasFEAR1}.  
\end{prop}
Compared to the asymptotic normality in Proposition \ref{thm:biasFEAR1} for the prototype model, the only change in Proposition \ref{thm:biasFEAR1_dif} is that the function $f_{T,h}(\rho)$ in the bias term of (\ref{eq:biasFEAR1}) becomes the cumulative sum $\sum_{r=1}^h f_{T,r}(\rho)$. The Nickell bias with the differenced dependent variable has the same root as that of the prototype model (\ref{eq:modelh}), and thus SPJ is still a valid solution to the Nickell bias in (\ref{eq:DGP-hD}). 

\begin{rem}
\label{rem:attentuation_dif} Suppose $\rho \in (0,1)$. Following the analysis in Remark \ref{rem:attentuation}, the following properties hold for the panel LP model (\ref{eq:DGP-hD}). 
\begin{enumerate}
\item Given all
the parameters in the DGP, the bias worsens with a larger $h$, since $\sum_{r=1}^h f_{T,r}(\rho)/\sqrt{T_h}$ increases with $h$.
\item The bias shrinks the FE estimator toward zero since 
\[
\hat{\beta}^{(h){\rm \mathrm{fe}}\text{-}{\rm dif}}\stackrel{a}{\sim}\beta^{(h)}\left(1-\frac{1}{T_{h}}\cdot\dfrac{\sigma_{u_{x}}^{2}\sum_{r=1}^h f_{T,r}(\rho)}{\sigma_{x}^{2}\sum_{r=1}^h \rho^{r}}\right)+\frac{\mathcal{Z}}{\sqrt{NT_{h}}}
\]
by noticing the true IRF $\beta^{(h)}=\beta^{(0)}\sum_{r=1}^h\rho^{r}$. When $T$ is sufficiently large, we have 
\[ 1-\frac{1}{T_{h}}\cdot\dfrac{\sigma_{u_{x}}^{2}\sum_{r=1}^h f_{T,r}(\rho)}{\sigma_{x}^{2}\sum_{r=1}^h \rho^{r}}>0.\]
Therefore, no matter whether the true $\beta^{(h)}$ is positive or negative,
the FE estimator exhibits an \emph{attenuation bias} and underestimates
$|\beta^{(h)}|$. This conclusion coincides with that in Remark \ref{rem:attentuation}, and the empirical evidence in Sections \ref{subsec: Baron2021} and \ref{subsec: CS2008} with differenced dependent variables.

\item 

Under a small sample $T$ and a relatively long horizon $h$, it is possible that 
\[ 1-\frac{1}{T_{h}}\cdot\dfrac{\sigma_{u_{x}}^{2}\sum_{r=1}^h f_{T,r}(\rho)}{\sigma_{x}^{2}\sum_{r=1}^h \rho^{r}} < 0.\]
Thus, in some particular settings with finite samples the FE estimates and the true IRFs may have opposite signs. This phenomenon is observed in the simulation study in Section \ref{subsec: empirical oriented}, as shown in Figure \ref{fig:dif}. 
\end{enumerate}
\end{rem}
\begin{proof}[Proof of Proposition \ref{thm:biasFEAR1_dif}]Following the proof of Proposition \ref{thm:biasFEAR1} in Section \ref{subsec: proof FEAR1}, the essential procedure is to calculate $\mathbb{E}\left[ T_h^{-1/2}\sum_{t\in\mathcal{T}^h} \tilde x_{i,t} e_{i,t+h}^{(h)}\right]$. Note that the $\varepsilon^{(r)}_{i,t+r}$ in (\ref{eq:DGP Delta FD s}) for the model with a differenced dependent variable has the same expression as the error term in (\ref{eq: close form prototype}) for the prototype model. Following the arguments for (\ref{eq:Exe}), we have 
\begin{equation*}
\mathbb{E}\left[\frac{1}{\sqrt{T_{h}}}\sum_{t\in\mathcal{T}^{h}}\tilde{x}_{i,t}\varepsilon_{i,t+r}^{(r)}+\dfrac{\beta^{(0)}}{\sqrt{T_{h}}}f_{T,r}(\rho)\right]=0.
\end{equation*}
Therefore, by $e_{i,t+h}^{(h)} = \sum_{r=1}^{h}\varepsilon_{i,t+r}^{(r)}$ as  in (\ref{eq: sum up}), we have 
\begin{equation*}
     \mathbb{E}\left[\frac{1}{\sqrt{T_{h}}}\sum_{t\in\mathcal{T}^{h}}\tilde{x}_{i,t}e_{i,t+h}^{(h)}\right] = \sum_{r=1}^{h}\mathbb{E}\left[\frac{1}{\sqrt{T_{h}}}\sum_{t\in\mathcal{T}^{h}}\tilde{x}_{i,t}\varepsilon_{i,t+r}^{(r)}\right] = -\dfrac{\beta^{(0)}}{\sqrt{T_{h}}}\sum_{r=1}^h f_{T,r}(\rho).
\end{equation*}
With this essential result, all other arguments for Proposition \ref{thm:biasFEAR1_dif} follow the proof of Proposition \ref{thm:biasFEAR1} in Section \ref{subsec: proof FEAR1}. 
\end{proof}
\subsection{Estimation and Inference with Split-Panel Jackknife  \label{subsec:extension HJ}}\par In
the main text, we present the theory in simple forms to highlight that the Nickell bias
in LP can be effectively resolved by the split-panel
jackknife (SPJ) estimator. In practice, applied researchers may attempt
alternative specifications in both the predictors and the correlation
structure in the error terms. These extensions can be accommodated
by adapting existing econometric methods.

\par \subsubsection{Two-way Clustered Standard Error\label{subsubsec: twoway}}\par We first
provide the formula for two-way clustered standard error of the SPJ
estimator. Recall that the asymptotic variance estimator $\hat{\mathbf{V}}_{N}:=\hat{\mathbf{Q}}^{-1}\widehat{\mathbf{R}}_{N}\hat{\mathbf{Q}}^{-1}$
below Theorem \ref{thm:hj} is clustered by individual, where 
\[
\widehat{\mathbf{R}}_{N}:=(NT_{h})^{-1}\sum_{i\in[N]}\sum_{t,s\in\mathcal{T}^{h}}\mathbf{d}_{i,t}^{*}\mathbf{d}_{i,s}^{*\top}\hat{e}_{i,t+h}^{(h){\rm spj}}\hat{e}_{i,s+h}^{(h){\rm spj}}.
\]
Similarly, we can construct an asymptotic variance estimator $\hat{\mathbf{V}}_{T}:=\hat{\mathbf{Q}}^{-1}\widehat{\mathbf{R}}_{T}\hat{\mathbf{Q}}^{-1}$
with 
\[
\widehat{\mathbf{R}}_{T}:=(NT_{h})^{-1}\sum_{t\in\mathcal{T}^{h}}\sum_{i,j\in[N]}\mathbf{d}_{i,t}^{*}\mathbf{d}_{j,t}^{*\top}\hat{e}_{i,t+h}^{(h){\rm spj}}\hat{e}_{j,t+h}^{(h){\rm spj}}
\]
to cluster by time. Finally, the robust variance estimator
(or White estimator) is $\hat{\mathbf{V}}_{NT}:=\hat{\mathbf{Q}}^{-1}\widehat{\mathbf{R}}_{NT}\hat{\mathbf{Q}}^{-1}$
with 
\[
\widehat{\mathbf{R}}_{NT}:=(NT_{h})^{-1}\sum_{t\in\mathcal{T}^{h}}\sum_{i\in[N]}\mathbf{d}_{i,t}^{*}\mathbf{d}_{i,t}^{*\top}(\hat{e}_{i,t+h}^{(h){\rm spj}})^{2}.
\]
Following \citet{cameron2011robust}, we can construct the two-way
clustered variance estimator as $\hat{\mathbf{V}}_{\text{TW}}=\hat{\mathbf{V}}_{N}+\hat{\mathbf{V}}_{T}-\hat{\mathbf{V}}_{NT}.$

\subsubsection{Cross-Sectional Correlation\label{subsubsec: Driscoll Kraay}}
\par We provide the formula for a robust standard error of the SPJ
estimator under cross-sectional dependence. Following \citet{driscoll1998consistent}, 
 the robust variance estimator is $\hat{\mathbf{V}}_{\rm DK}:=\hat{\mathbf{Q}}^{-1}\widehat{\mathbf{S}}_{NT}\hat{\mathbf{Q}}^{-1}$
with 
\[
\widehat{\mathbf{S}}_{NT} := N^{-1}\left[\hat{\mathbf{\Delta}}_0 + \sum_{j=1}^{m(T_h)}w(j,m(T_h))\left(\hat{\mathbf{\Delta}}_j + \hat{\mathbf{\Delta}}_j^\top \right)\right],
\]
where $m(T_h) = \lfloor T_h^{1/4}\rfloor$, $w(j,m(T_h)) = 1  - j/[m(T_h)+1]$, and 
\[
\hat{\mathbf{\Delta}}_j = \dfrac{1}{T_h}\sum_{t=j+1}^{T_h} \mathbf{g}_{N,t} \mathbf{g}_{N,t-j}^\top\text{ with } \mathbf{g}_{N,t} = \sum_{i\in[N]}\mathbf{d}_{i,t}^{*}\hat{e}_{i,t+h}^{(h){\rm spj}}
\]
for any $j=1,2,\dots,m(T_h)$. Note that $\widehat{\mathbf{S}}_{NT}$ consistently estimates the covariance matrix
\[{\mathbf{S}}_{NT} = \dfrac{1}{N T_h}\sum_{i,j\in[N]}\sum_{t,s\in\mathcal{T}^h}\mathbb{E}(\mathbf{d}_{i,t}^{*}\mathbf{d}_{j,s}^{*\top}{e}_{i,t+h}^{(h)} {e}_{j,s+h}^{(h)}).\]
It is precisely the \citet{newey1987simple} covariance matrix estimator robust to heteroskedasticity and serial correlation, applied to the sequence of cross-sectional averages of $\mathbf{d}_{i,t}^{*}\hat{e}_{i,t+h}^{(h){\rm spj}}$.

\subsubsection{Time Fixed Effect }\label{subsec: TE HJ}
In the main text we only discuss the cross-sectional fixed effects. Time fixed effects are ubiquitous in panel data empirical applications.
The two-way FE estimator captures potential time FE along with the
individual FE. \citet{chudik2018half} study SPJ for a generic two-way
FE regression. We borrow their approach in panel LP. Specifically,
we can redefine the demeaned variable 
\[
\tilde{\mathbf{W}}_{i,t}=\mathbf{W}_{i,t}-T_{h}^{-1}\sum_{t\in\mathcal{T}^{h}}\mathbf{W}_{i,t}-N^{-1}\sum_{i\in[N]}\mathbf{W}_{i,t}+(NT_{h})^{-1}\sum_{t\in\mathcal{T}^{h}}\sum_{i\in[N]}\mathbf{W}_{i,t}
\]
and similar two-way demeaning also applies to other observable variables.
All other steps will follow the formulas for the one-way fixed effect.
As for the standard error, we update the expression for $\boldsymbol{\mathbf{d}}_{i,t}^{*}$
in Assumption \ref{assu:limit_HJ} given as 
\begin{align*}
\mathbf{d}_{i,t}^{*}= & \left(\mathbf{W}_{i,t}-\bar{\mathbf{W}}_{i,b}-\frac{1}{T_{h}}\sum_{t\in\mathcal{T}^{h}}(\mathbf{W}_{i,t}-\bar{\mathbf{W}}_{i,b})\right)\cdot\boldsymbol{1}\left\{ t\in\mathcal{T}_{a}^{h}\right\} \\
 & +\left(\mathbf{W}_{i,t}-\bar{\mathbf{W}}_{i,a}-\frac{1}{T_{h}}\sum_{t\in\mathcal{T}^{h}}(\mathbf{W}_{i,t}-\bar{\mathbf{W}}_{i,a})\right)\cdot\boldsymbol{1}\left\{ t\in\mathcal{T}_{b}^{h}\right\} 
\end{align*}
where $\bar{\mathbf{W}}_{i,k}=(T_{h}/2)^{-1}\sum_{t\in\mathcal{T}_{k}^{h}}\mathbf{W}_{i,t}$
for $k\in\{a,b\}$.

\subsubsection{Unbalanced Panel }\par While the balanced panel facilitates
theoretical analysis, in practice unbalanced panels are the norm rather
than the exception. We follow Section A7 of \citet{chudik2018half}
to handle unbalanced panels for SPJ with the observation-splitting
procedure. Suppose that for the cross-sectional unit $i$ we observe
all variables over $t=T_{f_{i}},T_{f_{i}+1},\cdots,T_{l_{i}}$, where
$T_{f_{i}}$ and $T_{l_{i}}$ are the first and last time periods,
respectively. Without loss of generality we assume $T_{i}=T_{l_{i}}-T_{f_{i}}+1$
is an even number for every $i$. We implement SPJ by dividing the
$T_{i}$ observations into two sub-samples. The first sub-sample,
denoted by subscript $a$, contains the first $T_{i}/2$ observations;
the second half, denoted by subscript $b$, consists of the remaining
$T_{i}/2$ observations.

\section{Additional Simulation Results\label{sec:Simulation-app} }

This section collects several additional simulation results. Section \ref{subsec:analytical} compares SPJ to the analytical debiased estimator and the likelihood estimator that are highly dependent on closed-form formulas that vary with model specifications. Section \ref{subsec:corr} investigates the numerical performance of SPJ for panel LP with cross-sectional correlations. Section \ref{subsec:lag} considers a simulation setting with the lagged dependent variable as a control and the time fixed effects, and compares SPJ to the widely used difference GMM estimator by \citet{arellano1991some}. Section \ref{subsec: empirical oriented} includes an empirically-oriented simulation based on the data features of \citet{cerra2008growth}.
Section \ref{subsec: high rho} examines the robustness of SPJ when the regressor is persistent with the AR(1) coefficient $\rho$ close to one. Section \ref{subsec: summary simul} summarizes the simulation results. All results in this section are based on 1000 replications. 

\subsection{Comparison to Analytical Bias Correction and Likelihood Estimator\label{subsec:analytical}}
\par We compare the SPJ method to the analytical debiased estimator by \citet{herbst2021bias} and the likelihood estimator (LE) by \citet{alvarez2022robust} (See also \citet{bai2024likelihood}). These estimators are based on closed-form formulas, which are case-by-case according to the model specifications. For simplicity, we still focus on the prototype simulation setup based on (\ref{eq:DGP}) and (\ref{eq:AR1}) as in Section \ref{sec:Simulations}.  
\par Under this prototype setup, our Section \ref{subsec:consist} provides convenient closed-form formulas for the analytical bias correction method in \citet{herbst2021bias}, which is the same as the expressions for our Oracle debiased (DB) estimator in Section \ref{sec:Simulations}. As for the LE estimator, recall that the function $f_{T,h}(\rho)$ is defined below (\ref{eq:bias_limit}) in Proposition \ref{thm:biasFEAR1} about the Nickell bias of FE estimator. The LE estimators of IRFs $\{\beta^{(h)}\}_{0\leq h\leq H}$, the AR(1) coefficient $\rho$, and the variance of the error in the AR(1) model $\sigma_{u_x}^2$, solve the following bias-corrected scores: 
\begin{align}
&\dfrac{1}{NT}\sum_{i\in[N]}\sum_{t\in\mathcal[T]}\tilde x_{i,t}(\tilde y_{i,t} - \tilde{x}_{i,t}\beta^{(0)}) = 0, \label{eq: beta 0 score}\\
&\dfrac{1}{NT_h}\sum_{i\in[N]}\sum_{t\in\mathcal{T}^h}\tilde x_{i,t-h}(\tilde y_{i,t} - \tilde{x}_{i,t-h}\beta^{(h)}) + \dfrac{\beta^{(0)}\sigma_{u_x}^2}{T_h}f_{T,h}(\rho) = 0\text{ for }h=1,2,\dots,H, \label{eq: beta h score} \\
    &\dfrac{1}{N(T-1)}\sum_{i\in[N]}\sum_{t\in\mathcal{T}^h}\tilde x_{i,t-1}(\tilde x_{i,t} - \tilde{x}_{i,t-1}\rho) + \dfrac{\sigma_{u_x}^2}{T-1}f_{T,1}(\rho)  = 0,\label{eq: rho score}\\
    &\sigma_{u_x}^2 = \dfrac{1}{N(T-2)}\sum_{i=1}^n\sum_{t=2}^T(\tilde x_{i,t} - \tilde{x}_{i,t-1}\rho)^2. \label{eq: variance score}
\end{align}
Combining (\ref{eq: rho score}) and (\ref{eq: variance score}), we can solve $\rho$ by the Newton-Raphson method. Other parameters are readily solved with a consistent estimator of $\rho$. In terms of the difference between the two methods, the DB estimator relies on the analytical formula of the Nickell bias in the FE estimator, and estimates the unknown parameters in the bias formula ex-post for bias correction. The LE estimator depends on the bias-corrected scores before conducting the estimation.  
\par Figures \ref{fig:simple_irf}-\ref{fig:simple_coverage} display the simulation results. The DB and LE estimators produce precise estimation with small bias and RMSE, and confidence intervals with accurate coverage probabilities. These two methods are expected to perform the best under the correct model specification since they use most of the model structure. Nevertheless, we highlight that analytical formulas are indispensable for these two methods. These formulas vary with the model specification, requiring complex derivations and theoretical justifications when the order of lags increases, when the lagged dependent variables are included, when cross-sectional correlation exists, and when the dependent variable is specified as the change of outcome $\Delta_h y_{i,t+h} = y_{i,t+h} - y_{i,t}$ (Section \ref{subsec:bias_specification_FD}). Local projection is popular due to its easy implementation and robustness to model specifications. The inconvenience of analytical methods undermines these attractions. 

In contrast to the analytical methods, our recommended SPJ estimator is always a linear combination of three FE estimators following (\ref{eq:spj}), with proper modifications of the standard errors if the cross-sectional correlation exists. Figures \ref{fig:simple_irf}-\ref{fig:simple_coverage} show that SPJ also produces confidence intervals with coverage probabilities close to the nominal level, with RMSEs comparable to the DB and LE methods. The mild inflation of RMSE in SPJ is compensated with easy implementation and robustness to model specifications, maintaining the advantages of LP in panel data.

\subsection{Cross-Sectional Correlation\label{subsec:corr}}
\par In Section \ref{sec:bias}, we follow the common practice in the literature of panel data to assume independent cross-sections. Cross-sectional correlations are common in real data  \citep{pesaran2015testing,juodis2022incidental}. 
This section extends the prototype simulation setup as in Section \ref{sec:Simulations} to cross-sectional correlation. In particular, we generate the error term of the DGP (\ref{eq:DGP}) by 
\[
(u^y_{1,t}, u^y_{2,t}, \ldots, u^y_{N,t})' \sim \mathcal{N}(0_N,\Sigma_N),\]
where $\Sigma_N = (\Sigma_{i,j})_{1\leq i\leq N,1\leq j\leq N}$ is the cross-sectional covariance matrix. We follow \citet{driscoll1998consistent} to generate the covariance matrix by $\Sigma_{i,j} = \textbf{1}\{i=j\} + \lambda_i\lambda_j \textbf{1}\{i\neq j\}$, where $\{\lambda_i\}_{1\leq i\leq N}$  collects $N$ independent random variables of uniform distribution in the interval $[0,1/\sqrt{2}]$. For preciseness, for each $(N,T)$ we generate the covariance matrix $\Sigma_N$ before conducting the simulations, and consider $\Sigma_N$ as constant throughout all replications of the simulation study. Other settings of the parameters and error terms follow Section \ref{sec:Simulations}. 
\par The two-way clustered standard error and the \citet{driscoll1998consistent} standard error (Section \ref{subsubsec: Driscoll Kraay}) are commonly used to address cross-sectional correlations in panel data. This section examines the performance of these two standard errors. We still compare the SPJ and FE estimators as in Section \ref{sec:Simulations}. For SPJ that allows for cross-sectional correlation, the asymptotic variance estimators follow  Sections \ref{subsubsec: twoway} and \ref{subsubsec: Driscoll Kraay}. For FE in the prototype LP model with a single AR(1) regressor, the variance estimators also have similar expressions as those in Sections \ref{subsubsec: twoway} and \ref{subsubsec: Driscoll Kraay}, with $\mathbf{d}^*_{i,t}$ replaced by the demeaned regressor $\tilde x_{i,t}$. 
\par To save space, we only present the results with the \citet{driscoll1998consistent} standard error in Figures \ref{fig:correlation_irf}-\ref{fig:correlation_coverage}; the results with the two-way clustered standard error are almost the same. The FE estimator still suffers from the Nickell bias that causes large estimation errors and severe low convergence probability of the confidence intervals. In contrast, SPJ produces robust estimation and inference with cross-sectional correlations.

\subsection{Panel LP with Lagged Dependent Variable and Time Effect \label{subsec:lag}}
We consider a more general DGP with a lagged dependent variable given as  
\begin{align}y_{i,t} & =\mu_{i}^{(0)y}+g_{t}^{(0)y}+\tau y_{i,t-1}+\beta^{(0)}x_{i,t}+u_{i,t}^{y},
\label{eq: VAR simul y}\\
x_{i,t} & =\mu_{i}^{x}+\kappa y_{i,t-1}+\rho x_{i,t-1}+u_{i,t}^{x}.
\label{eq: VAR simul x}
\end{align} 
The time fixed effect $g_{t}^{(0)y}$ as well as the lagged dependent variable
$y_{i,t-1}$ are added to mimic common practice in applied works.
From (\ref{eq: VAR simul y}) and (\ref{eq: VAR simul x}) we can deduce the following panel VAR model
\begin{equation}
\left(\begin{array}{c}
y_{i,t} \\
x_{i,t}
\end{array}\right)  =  \left(\begin{array}{c}
\mu^{(0)y}_{i} + \beta^{(0)}\mu_i^x + g_t^{(0)y}\\
\mu^x_i
\end{array}\right)  + \mathbf{P}\left(\begin{array}{c}
y_{i,t-1} \\
x_{i,t-1}
\end{array}\right) + \left(\begin{array}{c}
u^y_{i,t} + \beta^{(0)}u^x_{i,t} \\
u^x_{i,t}
\end{array}\right).
\label{eq: VAR simul deduce}
\end{equation}
Therefore, the true impulse response function for Period $h\geq 1$ becomes the $(1,2)$th element of the matrix ${\bf P}^{h}$,
where \[{\bf P}=\left( \begin{array}{cc}
\tau + \beta^{(0)}\kappa &  \beta^{(0)}\rho \\
\kappa & \rho 
\end{array} \right) \]
is the coefficient matrix of the VAR model in (\ref{eq: VAR simul deduce}). 
Based on the insights from Section \ref{subsec: implicit bias}, this
panel VAR(1) model also incurs the Nickell bias issue, though the
closed-form bias formula for this setting is too complicated to be deduced
and coded. Therefore in our simulation we consider the FE estimator
and the SPJ estimator, but omit the DB estimator.

\par In our setting, $\mu_{i}^{(0)y}$, $\mu_{i}^{x}$, $u_{i,t}^{y}$,
$u_{i,t}^{x}$ and the sample sizes follow the simulation setting
of simple models (\ref{eq:DGP}) and (\ref{eq:AR1}). The time effect
is set as $g_{t}^{(0)}:=0.025t+0.001t^{2}$. We fix $\beta^{(0)}=-0.25$,
$\kappa=-0.5$, $\tau=0.5$ and vary $\rho\in\{0.2,0.4,0.5,0.6\}$;
to ensure that the VAR(1) system is stationary (net of the deterministic
trend) we cannot choose $\rho=0.8$ as in Section \ref{sec:Simulations}.
We keep the abbreviation ``FE'' to denote the (two-way)
individual-time fixed effect estimator. The SPJ estimator that allows
time effects is introduced in Appendix \ref{subsec: TE HJ}. In addition, the differenced GMM (DGMM) estimator by \citet{arellano1991some} has been widely used in dynamic panel models with the lagged dependent variable. Therefore, we also include DGMM in this simulation study for comparison. 
\par For $h=0$ we estimate the regression model (\ref{eq: VAR simul y}). For $h\geq 1$, following
the spirit of LP we estimate the model 
\[
y_{i,t+h}=\mu_{i}^{(h)y}+g_{t}^{(h)y}+\tau^{(h)}y_{i,t}+\beta^{(h)}x_{i,t}+e_{i,t+h}^{(h)}
\]
 where $\beta^{(h)}$ is the IRF of interest. For a generic panel variable $w_{i,t}$, let $\check w_{i,t} = w_{i,t} - N^{-1}\sum_{j=1}^N {w}_{j,t}$ denote the demeaned variable over the cross-sections to remove the time effect. DGMM estimates the following differenced model 
\[
\Delta \check y_{i,t+h}  = \tau^{(h)}\Delta \check y_{i,t}+\beta^{(h)}\Delta \check x_{i,t}+ \Delta \check e_{i,t+h}^{(h)}
\]
and uses the first-order lag $(\check y_{i,t-1},\check x_{i,t-1})$ as instrumental variables. 
\par Figures \ref{fig:lag_irf}-\ref{fig:lag_coverage} compare the results of FE, SPJ, and DGMM estimators. The key findings for FE and SPJ are similar to those in Section \ref{sec:Simulations}. Specifically, FE incurs an obvious attenuation bias and its substantial under-coverage
of the confidence interval. The SPJ produces robust estimation and inference as in Section \ref{sec:Simulations} even in this more complex setting with the lagged dependent variable and time fixed effects included in the regression. 
\par In contrast, the performance of DGMM is unstable, highly dependent on the essential parameters in the DGP. When $\rho$ is relatively small (0.2 or 0.4), DGMM produces the largest RMSE among all considered methods under a small $h$ with a coverage probability close to the nominal level, while its RMSE becomes the smallest when $h$ becomes large. When the regressor $x_{i,t}$ becomes more persistent with a larger $\rho$  (0.5 or 0.6), the RMSE of DGMM increases substantially and becomes the largest under all horizons $h$ due to the weak IV issue. In particular, when $(N,T) = (30,60)$ and $\rho = 0.6$, DGMM has extremely unstable performance as shown in Figure \ref{fig:lag_irf}. In this scenario, the RMSE of DGMM is always greater than 1 and out of range of the plot as shown in Figure \ref{fig:lag_rmse}.
The potential lack of efficiency undermines the robustness of local projection if we use DGMM for estimation and inference.

\subsection{Empirically-Driven Simulation Based on \citet{cerra2008growth}\label{subsec: empirical oriented}  }  
 
\par In the prototype model considered in Section \ref{sec:Simulations}, we consider $(N,T) \in \{(30,60),(30,120),(50,120)\}$ and $\rho\in\{0,0.2,0.5,0.8\}$. Though these settings can cover common features of panel data in empirical macroeconomics and finance, including the first three empirical examples in Section \ref{sec:Empirical-Application}, real-data applications may be involved with different settings. In our fourth real-data example (\citet{cerra2008growth}, CS08 hereafter) in Section \ref{subsec: CS2008}, the number of cross-sections $N$ is larger than the time span $T$. In this scenario, the Nickell bias of the FE estimator is more severe according to our Proposition \ref{thm:biasFEAR1}. SPJ and FE estimators showcase the most statistically significant discrepancy in CS08 among the four applications in our Section \ref{sec:Empirical-Application}. In addition, the dependent variable in CS08 is the differenced outcome $\Delta_h y_{i,t+h} = y_{i,t+h} - y_{i,t}$, which is not covered by the prototype model in Section \ref{sec:Simulations}. 
\par We therefore use the DGP (\ref{eq:DGP-FD}) to generate the outcome variable as in Section \ref{subsec:bias_specification_FD}, where the regressor $x_{i,t}$ still follows the AR(1) process (\ref{eq:AR1}) with individual fixed effects. We calibrate the essential parameters for the simulation based on CS08. We estimate the AR(1) coefficient of the essential regressor of interest in CS08 by SPJ in the presence of individual fixed effects, and the result is 0.0953. In addition, in the data of CS08, totally $N=82$ countries have a full length of time series data with $T = 36$. We therefore set $\rho = 0.0953$ and $(N,T) = (82,36)$ for this simulation. Other settings of the coefficients, the fixed effects, and the error terms follow those in Section \ref{sec:Simulations}. We use FE and SPJ to estimate the panel LP model (\ref{eq:DGP-hD}) over the replications. 
\par Figure \ref{fig:dif} exhibits the simulation results. As usual, the FE estimator suffers from the Nickell bias that causes large estimation errors and severely low coverage probabilities of the confidence intervals. In contrast, the SPJ still shows accurate estimation of IRFs and valid confidence intervals with coverage probabilities close to the nominal level, even in this simulation setting calibrated based on our representative real-data example. These results again demonstrate the robustness of SPJ.

\subsection{SPJ for Panel LP with Highly Persistent Regressor\label{subsec: high rho}}

Up to now, we have been maintaining the stationarity of the VAR model throughout the theoretical discussions. To examine the robustness of SPJ with highly persistent regressors, we revisit the prototype model in Section \ref{sec:Simulations} with $\rho$ as large as 0.9 and 0.95, which are close to 1. The econometric literature models the AR(1) coefficient as a sequence $\rho = \rho_T \to 1$ as $T\to\infty$, and names such AR(1) models the \textit{mildly integrated} and \textit{local-to-unity} processes \citep{phillips1987towards}. 
\par Figures \ref{fig:highrho_irf}-\ref{fig:highrho_coverage} display the simulation results. As expected, the FE estimator produces unsatisfactory performance with large estimation errors and extremely low coverage probabilities of confidence intervals. SPJ still substantially mitigates the Nickell bias of FE, with much smaller RMSEs than FE. When $T=120$, the coverage probabilities of the confidence intervals by SPJ do not substantially deviate from the nominal level even when $h = 10$. These results by SPJ improve those by FE with coverage probabilities below 0.1.  
\par It is not surprising to observe some deviations of the SPJ confidence intervals' coverage probabilities from the nominal level, since the theoretical foundations of SPJ do not cover highly persistent regressors. Even in the context of time series, under nonstationarity the ordinary least squares (OLS) suffers from the \textit{Stambaugh} bias \citep{stambaugh1999predictive}, which causes size distortions in the $t$-test. \citet{montiel2021local} develops an inferential procedure for time series LP robust to nonstationary time series with a lag augmentation. The asymptotic normal distribution of this lag-augmented estimator is achieved with efficiency loss for persistent regressors. Specifically, for unit roots this lag-augmented estimator has $\sqrt{T}$-consistency in the sense that $\hat\beta-\beta=O_p(1/\sqrt{T})$, in contrast to the $T$-consistency of OLS with Stambaugh bias. In time series predictive regressions, another estimator called IVX proposed by \citet{phillips2009econometric} enjoys asymptotic normality and maintains high efficiency for persistent regressors. This IVX estimator applies to LP, as LP is a sequence of predictive regressions. 
\par In panel LP with a large $N$, the time series Stambaugh bias will be carried over and fused with the Nickell bias, leading to a composite and substantially enlarged \textit{Nickell-Stambaugh} bias in the $t$-statistic \citep*{liao2024nickell}. The construction of an estimator free from the Nickell-Stambaugh bias for nonstationary panel LP is thus more challenging, and its theoretical foundation requires complex technical proofs. Therefore, nonstationary panel LP is beyond the reach of the current paper and left for a standalone research topic. \citet*{liao2024nickell} establishes a bias-free estimator for panel predictive regression based on a panel version of IVX (called \textit{IVX-X-jackknife}, IVXJ), which applies to panel LP. IVXJ asymptotically removes the Nickell-Stambaugh bias in panel LP with an agnostic attitude toward the persistence of the regressors. Interested readers can find more details in \citet*{liao2024nickell}. 

\subsection{Summary\label{subsec: summary simul}}
\par We summarize the comparative studies of all the methods considered in Section \ref{sec:Simulation-app}. As usual, FE incurs the Nickell bias and causes severe size distortion in inference. The DB and LP considered in Section \ref{subsec:analytical} rely on analytical formulas that vary with model specifications. The derivations of these formulas are inconvenient with further lags, additional control variables, and the change of outcome as the dependent variable. The popular DGMM method considered in Section \ref{subsec:lag} suffers from the weak IV issue when the persistence of $x_{i,t}$ measured by $\rho$ becomes moderate, which causes its unsatisfactory performance. 

SPJ always follows the formula (\ref{eq:spj}) and produces small estimation errors and accurate coverage probabilities of confidence intervals. Among all methods for comparison, SPJ remains robust under a variety of simulation settings, including the prototype model in Section \ref{subsec:analytical}, panel LP with cross-sectional correlations as in Section \ref{subsec:corr}, penal LP with lagged dependent variable as in Section \ref{subsec:lag}, and some empirically oriented scenarios as in Section \ref{subsec: empirical oriented}. Even in the challenging setup with $\rho$ close to 1 as in Section \ref{subsec: high rho}, SPJ still substantially mitigates the Nickell bias in FE. The robustness of SPJ maintains the advantages of local projection when it is extended from time series to panel data.

\begin{figure}
\begin{centering}
\includegraphics[width=1\textwidth]{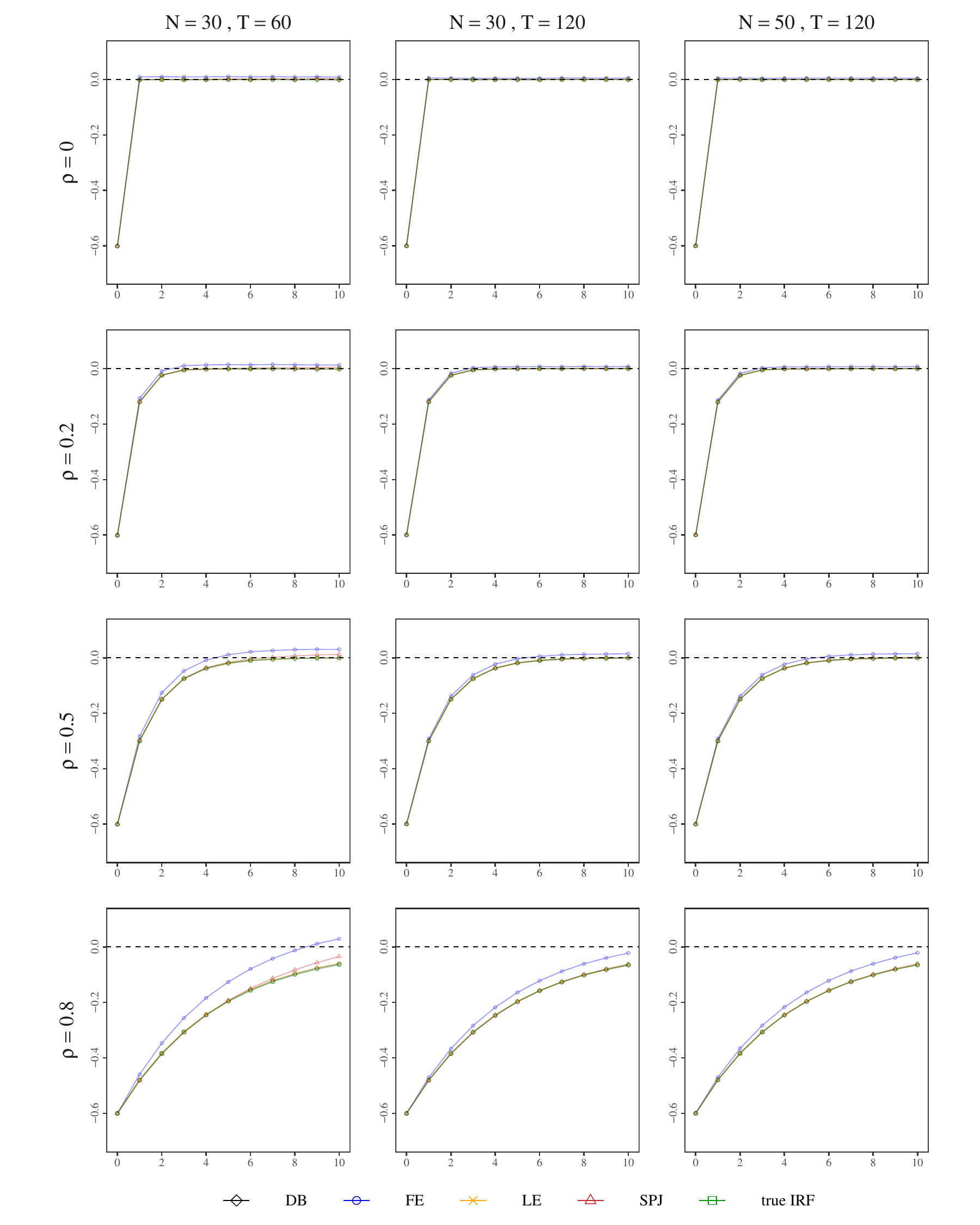}
\end{centering}
\caption{\label{fig:simple_irf} Estimated IRFs Averaged Over Replications with Comparison to DB and LE}
\end{figure}

\begin{figure}
\begin{centering}
\includegraphics[width=1\textwidth]{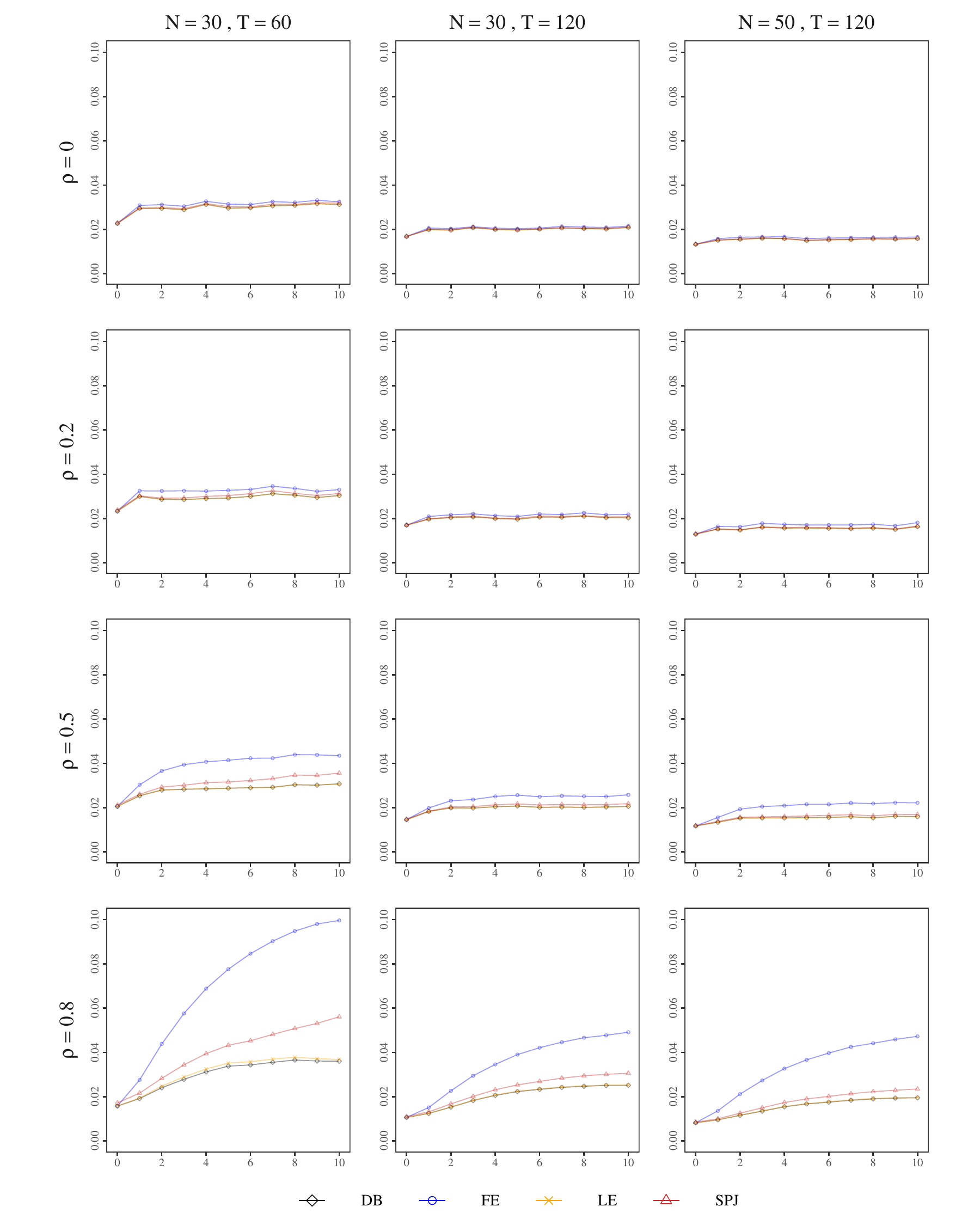}
\end{centering}
\caption{\label{fig:simple_rmse} RMSEs of IRFs with Comparison to DB and LE}
\end{figure}

\begin{figure}
\begin{centering}
\includegraphics[width=1\textwidth]{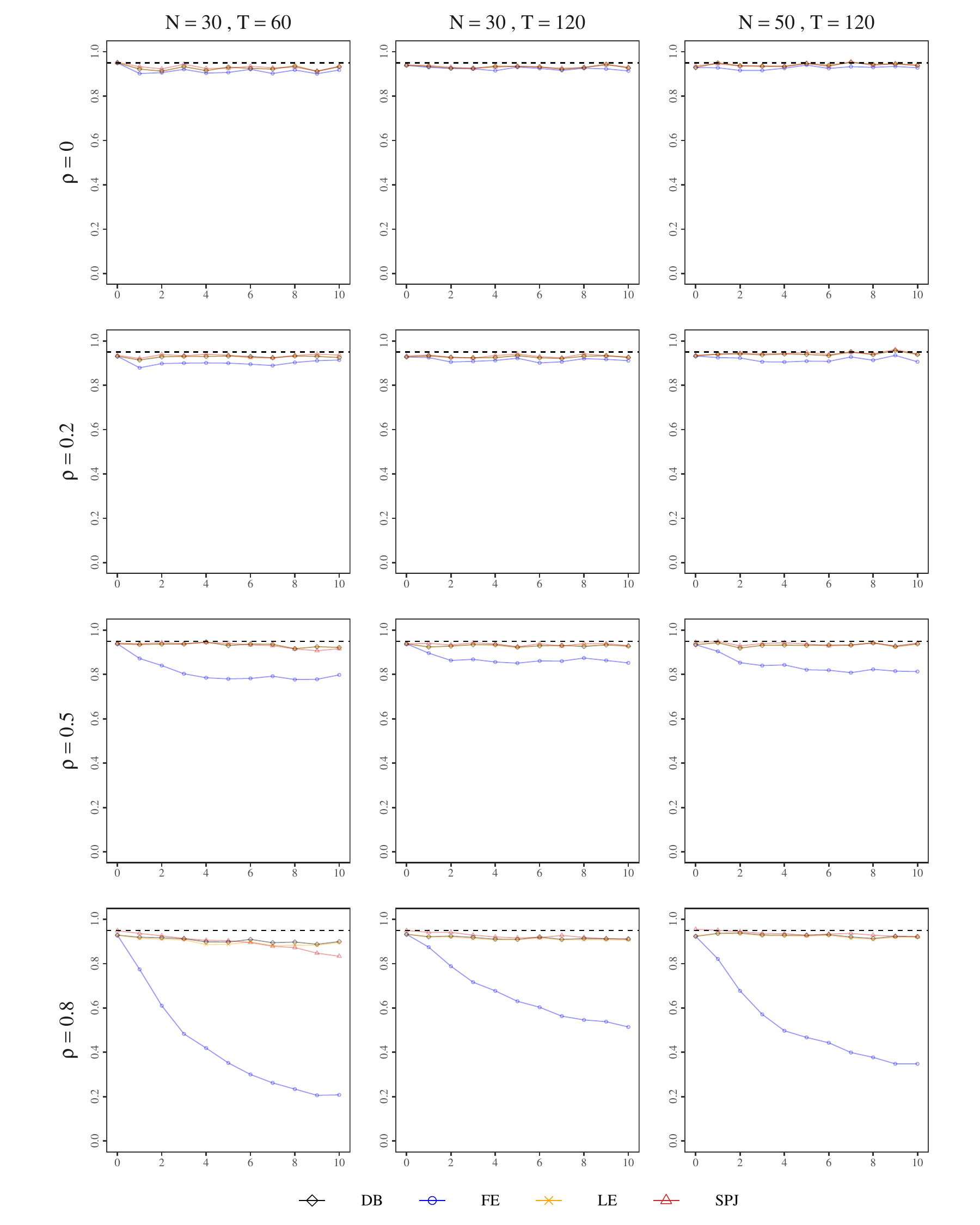}
\end{centering}
\caption{\label{fig:simple_coverage} Coverage Probability of Confidence Interval Based on $t$-Statistic with Comparison to DB and LE}
\end{figure}

\begin{figure}
\begin{centering}
\includegraphics[width=1\textwidth]{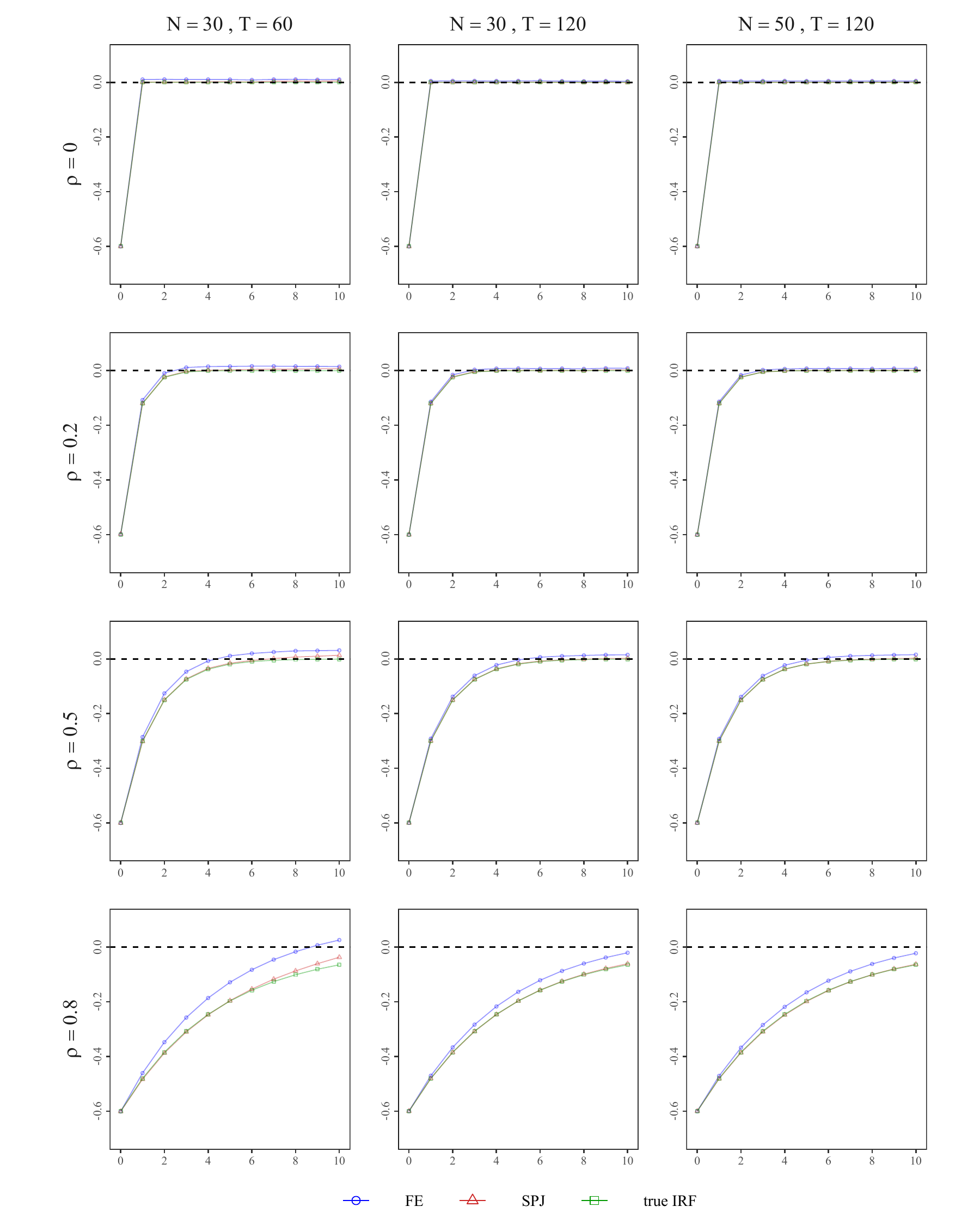}
\end{centering}
\caption{\label{fig:correlation_irf} Estimated IRFs Averaged Over Replications with Cross-Sectional Correlation}
\end{figure}

\begin{figure}
\begin{centering}
\includegraphics[width=1\textwidth]{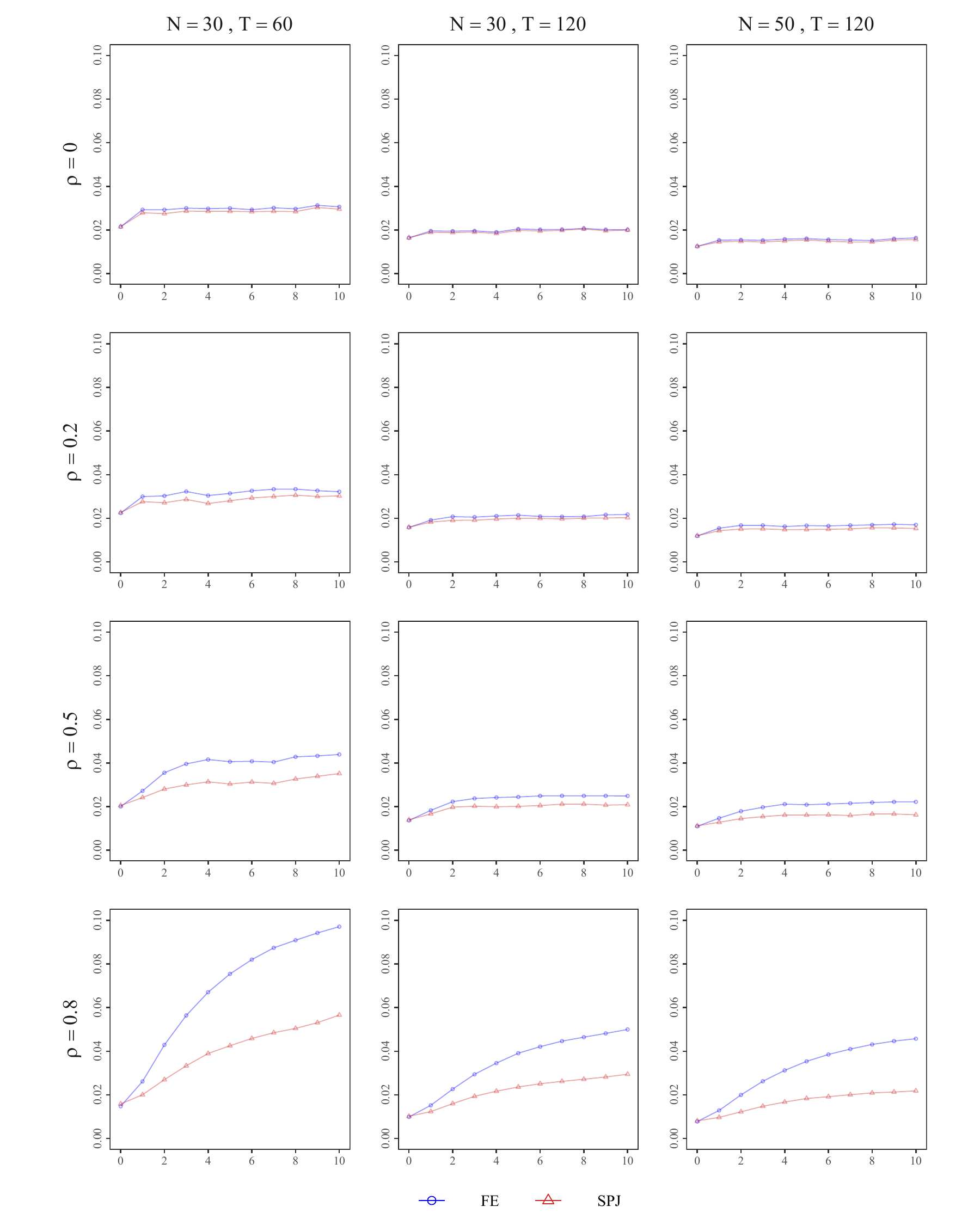}
\end{centering}
\caption{\label{fig:correlation_rmse} RMSEs of IRFs with Cross-Sectional Correlation}
\end{figure}

\begin{figure}
\begin{centering}
\includegraphics[width=1\textwidth]{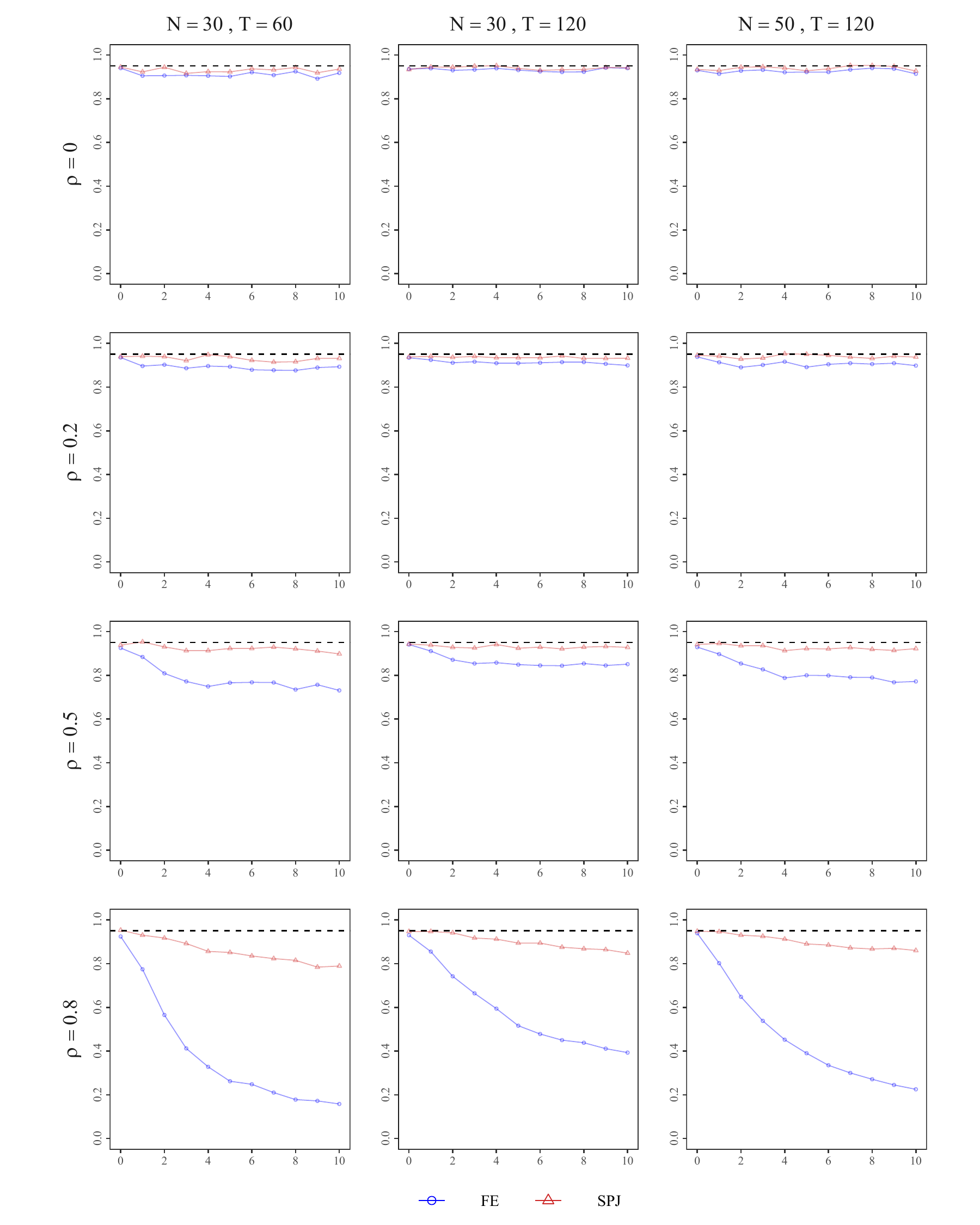}
\end{centering}
\caption{\label{fig:correlation_coverage} Coverage Probability of Confidence Interval Based on $t$-Statistic with Cross-Sectional Correlation}
\end{figure}

\begin{figure}
\begin{centering}
\includegraphics[width=1\textwidth]{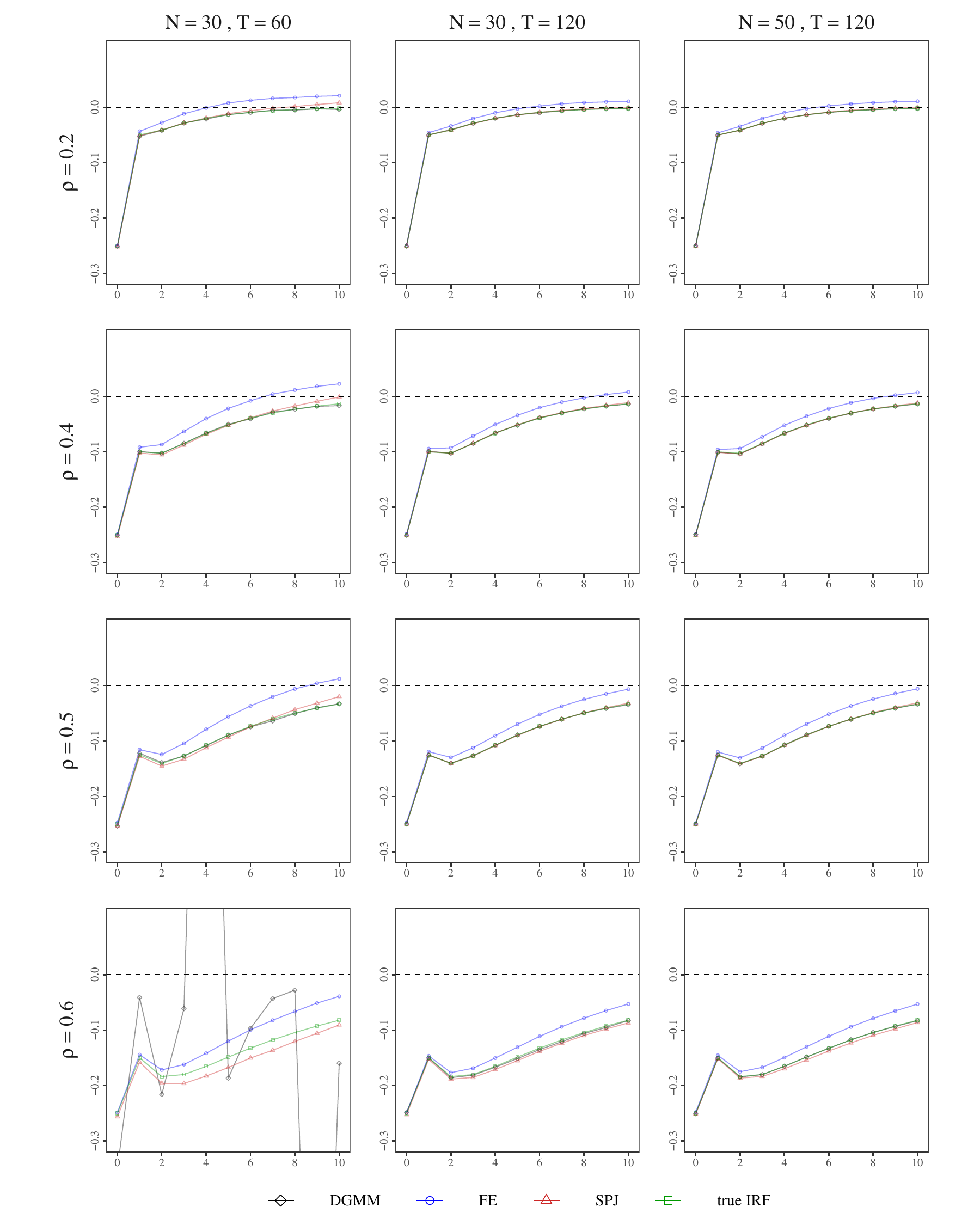}
\end{centering}
\caption{\label{fig:lag_irf} Estimated IRFs Averaged Over Replications with Lagged Dependent Variable and Time Effect}
\end{figure}

\begin{figure}
\begin{centering}
\includegraphics[width=1\textwidth]{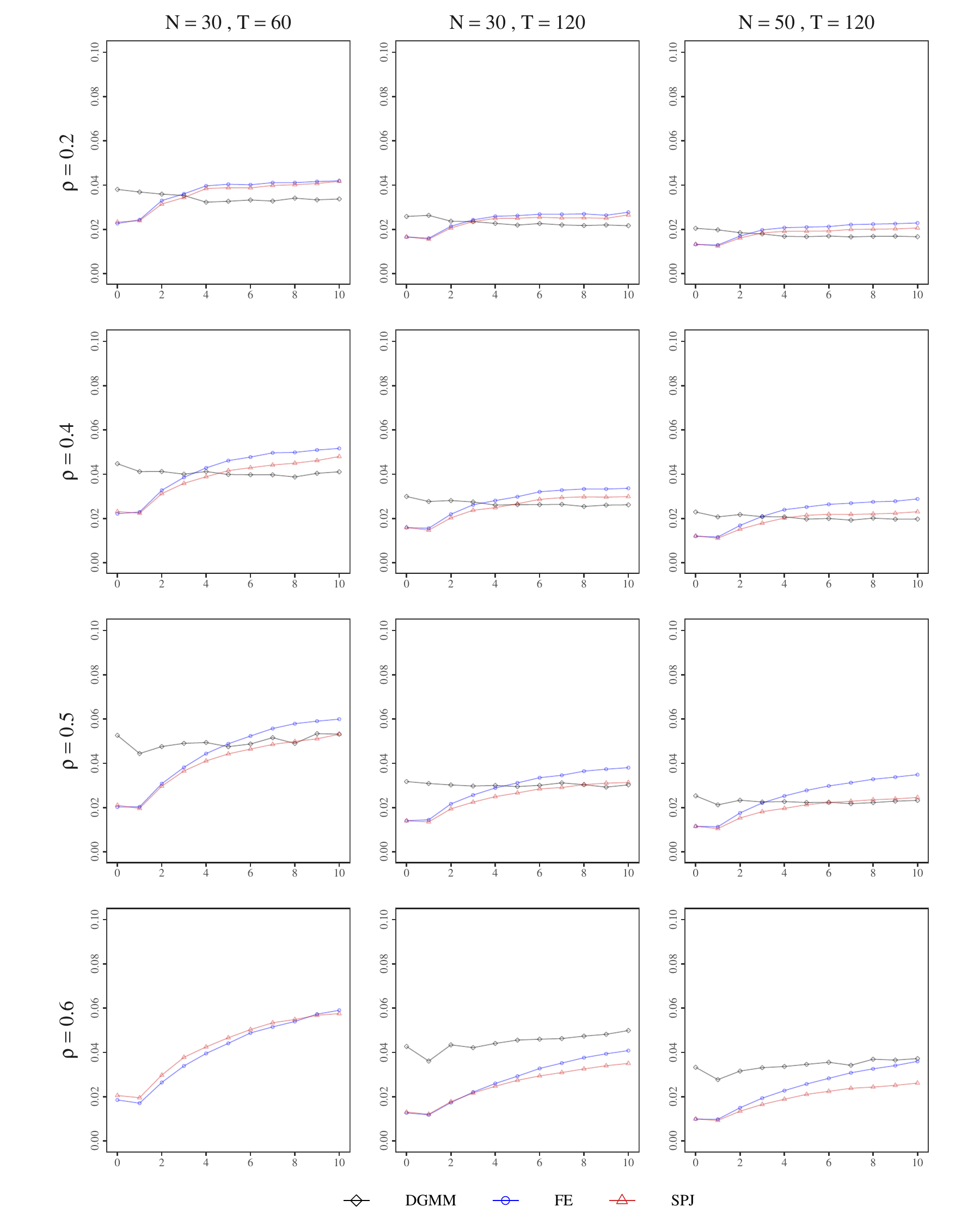}
\end{centering}
\caption{\label{fig:lag_rmse} RMSEs of IRFs   with Lagged Dependent Variable and Time Effect}
\end{figure}

\begin{figure}
\begin{centering}
\includegraphics[width=1\textwidth]{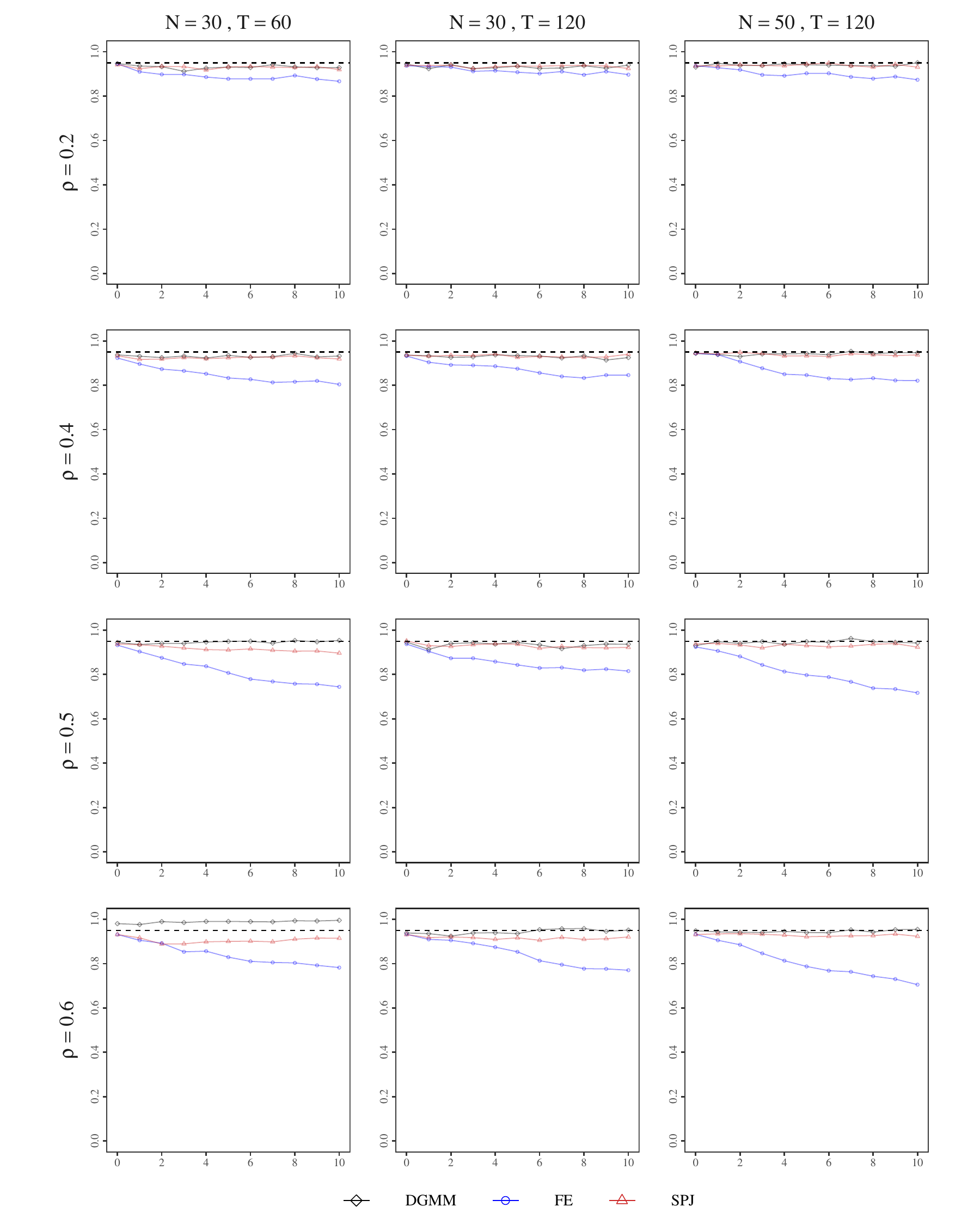}
\end{centering}
\caption{\label{fig:lag_coverage} Coverage Probability of Confidence Interval Based on $t$-Statistic with Lagged Dependent Variable and Time Effect}
\end{figure}

\begin{figure}
\begin{centering}
\includegraphics[width=1\textwidth]{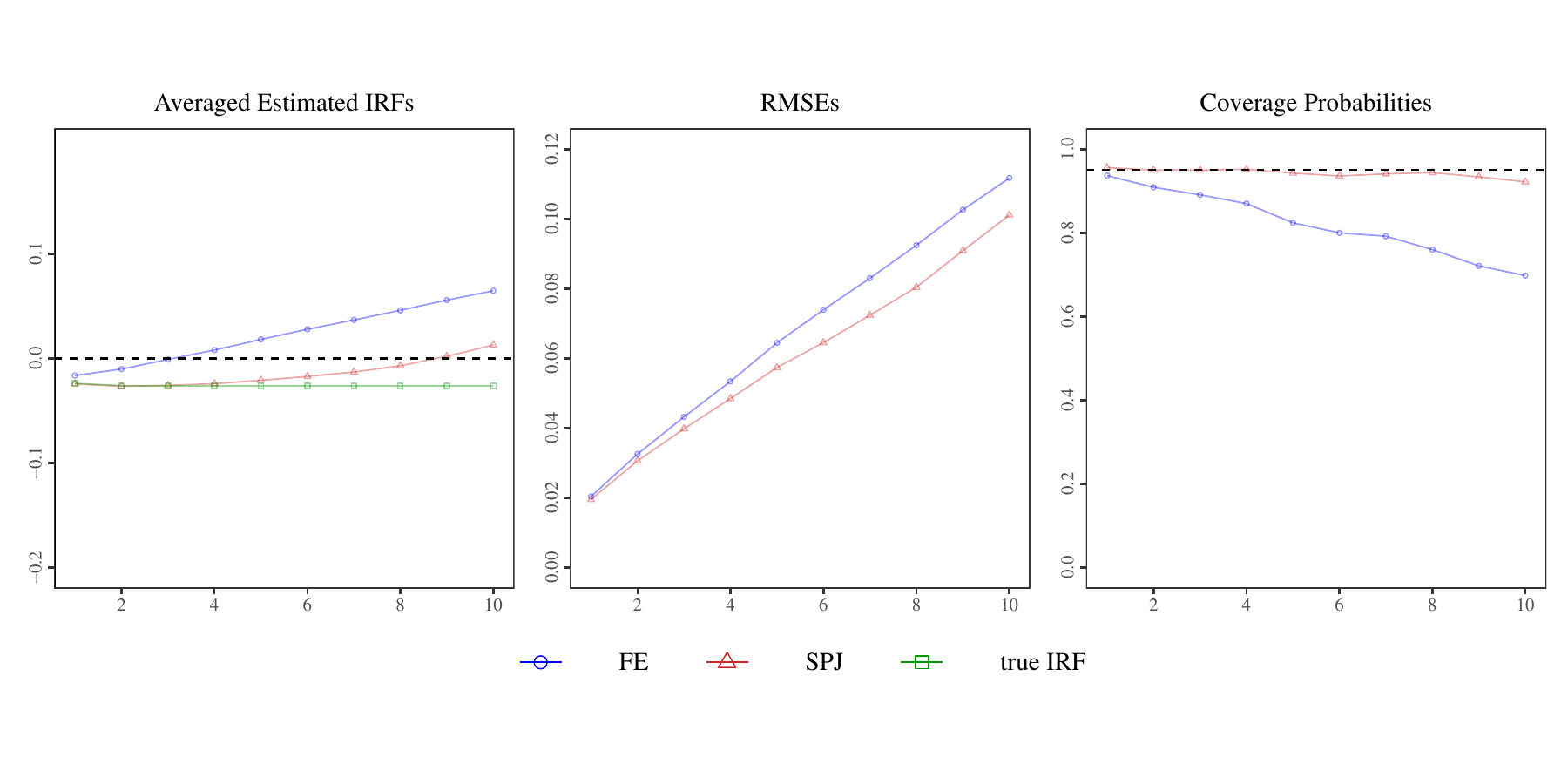}
\end{centering}
\caption{\label{fig:dif} Simulation Results Based on CS08}
\end{figure}


\begin{figure}
\begin{centering}
\adjincludegraphics[width = 1\textwidth,trim={0 {0.2\height} 0 {0.2\height}},clip]{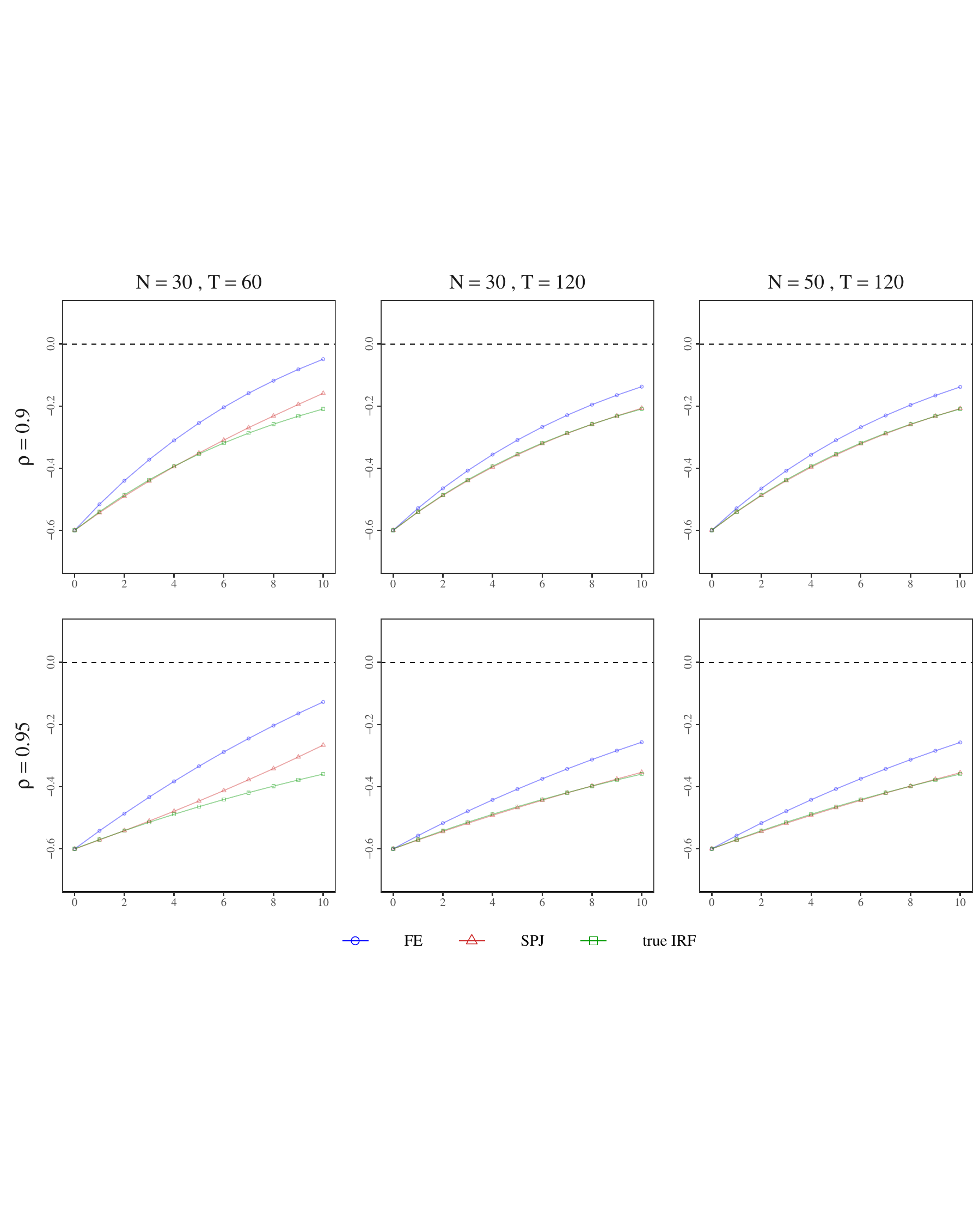}
\end{centering}
\caption{\label{fig:highrho_irf} Estimated IRFs Averaged Over Replications with Large $\rho$}
\end{figure}

\begin{figure} 
\centering 
\adjincludegraphics[width = 1\textwidth,trim={0 {0.2\height} 0 {0.2\height}},clip]{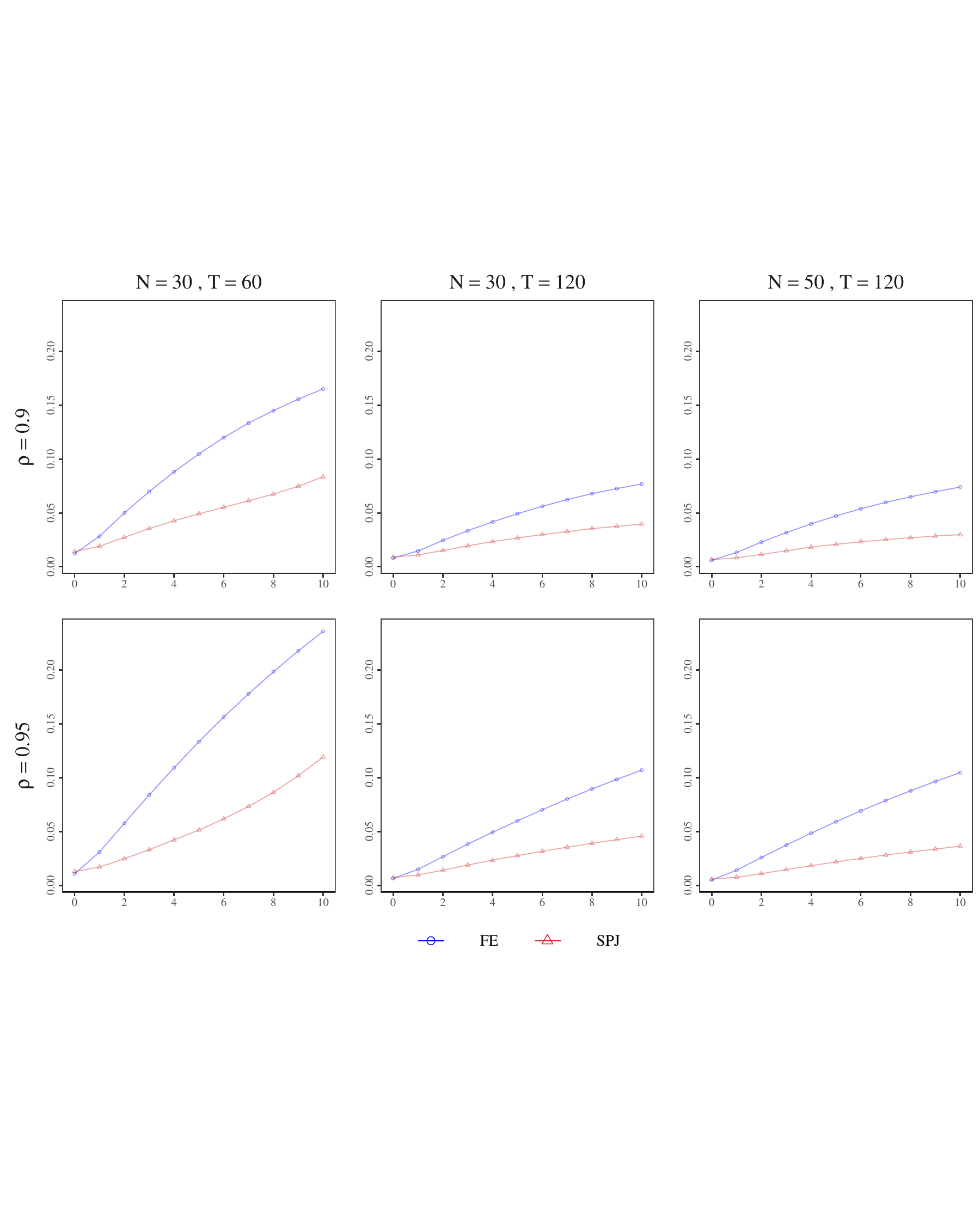} 
\caption{\label{fig:highrho_rmse} RMSEs of IRFs with Large $\rho$}
\end{figure}

\begin{figure}
\centering 
\adjincludegraphics[width = 1\textwidth,trim={0 {0.2\height} 0 {0.2\height}},clip]{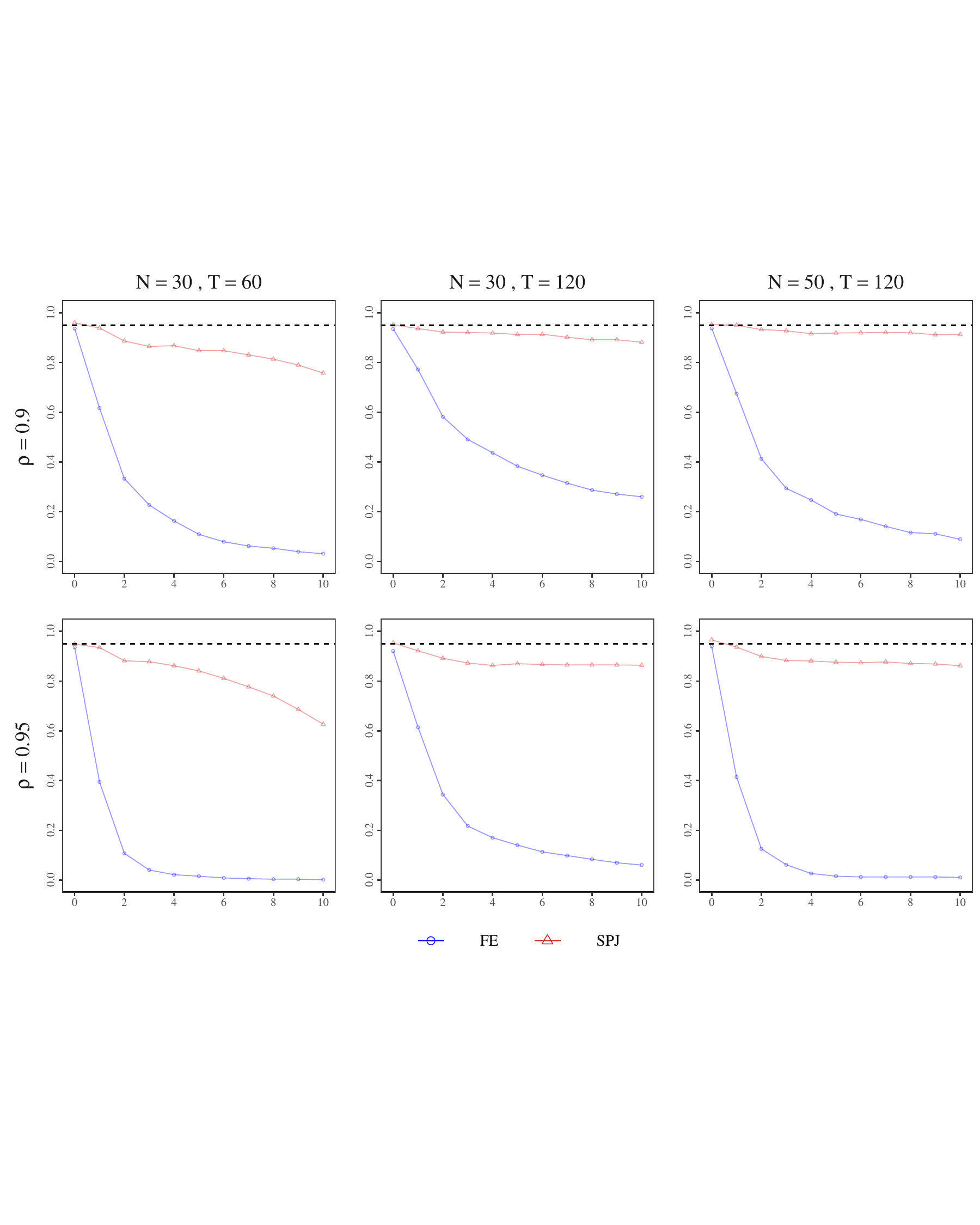}
\caption{\label{fig:highrho_coverage} Coverage Probability of Confidence Interval Based on $t$-Statistic with Large $\rho$}
\end{figure}

\end{document}